\documentclass[a4paper,UKenglish,cleveref, autoref, thm-restate]{lipics-v2021}

\pdfoutput=1 
\hideLIPIcs  
\nolinenumbers

\usepackage{hyperref}
\hypersetup{colorlinks=true,citecolor=blue,linkcolor=blue,urlcolor=blue}
\usepackage{blkarray}
\usepackage{bm}
\usepackage{amsmath}
\usepackage{amssymb}
\usepackage{amsthm}   
\usepackage{enumerate}
\usepackage{comment}
\usepackage{thm-restate}
\usepackage{graphicx}
\usepackage{tikz, tikz-cd}
\usetikzlibrary{matrix}
\usetikzlibrary{shapes}
\usetikzlibrary{arrows,decorations.markings, tikzmark}
\usetikzlibrary{decorations.pathreplacing,calligraphy,backgrounds}
\usetikzlibrary{arrows.meta,arrows}
\usepackage{xcolor}

\newtheorem{fact}[theorem]{Fact}

 \def\cE{c_{\text{E}}}
 \def\gamE{\gamma_{\text{E}}}
 
 \def\calP{\mathcal{P}}
 \def\calR{\mathcal{R}}

 \let\phi=\varphi

 \def\calT{\mathcal{T}}

\let\epsilon=\varepsilon
\def\calH{\mathcal{H}}

\newcommand{\homs}[2]{\mbox{\ensuremath{\mathsf{Hom}(#1 \to #2)}}}

\newcommand{\embs}[2]{\mbox{\ensuremath{\mathsf{Emb}(#1 \to #2)}}}

\newcommand{\subs}[2]{\mbox{\ensuremath{\mathsf{Sub}(#1 \to #2)}}}
\newcommand{\colsubs}[2]{\mbox{\ensuremath{\mathsf{ColSub}(#1 \to #2)}}}

\newcommand{\colsubsval}[2]{\mbox{\ensuremath{\mathsf{ColSub}_{\mathsf{val}}(#1 \to #2)}}}

\newcommand{\auts}[1]{\ensuremath{\mathsf{Aut}(#1)}}

\newcommand{\fptred}{\ensuremath{\leq^{\mathrm{fpt}}_{\mathrm{T}}}}

\newcommand{\cphomsprob}{\text{\sc{cp-Hom}}}

\newcommand{\subsprob}{\text{\sc{Sub}}}

\newcommand{\W}[1]{\ensuremath{\mathsf{W[#1]}}}

\newcommand{\leaves}{\mathsf{L}}
\newcommand{\nonleaves}{\mathsf{NL}}

\newcommand{\degl}{\mathsf{deg}_{\leaves}}
\newcommand{\degla}[1]{\mathsf{deg}_{\leaves,#1}}
\newcommand{\degnl}{\mathsf{deg}_{\nonleaves}}

\newcommand{\tmax}{{t_\mathsf{max}}}

\newcommand{\fork}{\mathrm{F}}
\newcommand{\hal}{\mathrm{C}}
\newcommand{\starnum}{\mathrm{S}}

\newcommand{\colT}{{T_\mathsf{col}}}
\newcommand{\colPhi}{{\Phi_\mathsf{col}}}
\newcommand{\colTprime}{{T'_\mathsf{col}}}
\newcommand{\colS}{{\mathcal{S}_\mathsf{col}}}
\newcommand{\colSprime}{{\mathcal{S}'_\mathsf{col}}}
\newcommand{\colR}{\mathcal{R}_\mathsf{col}}
\newcommand{\colP}{\mathcal{P}_\mathsf{col}}

\newcommand{\redc}[2][red,fill=red]{\tikz[baseline=-0.5ex]\draw[#1,radius=#2] (0,0) circle ;}
\newcommand{\greenc}[2][green!80!blue,fill=green!80!blue]{\tikz[baseline=-0.5ex]\draw[#1,radius=#2] (0,0) circle ;}
\newcommand{\bluec}[2][blue,fill=blue]{\tikz[baseline=-0.5ex]\draw[#1,radius=#2] (0,0) circle ;}
\newcommand{\yellowc}[2][yellow!50!orange,fill=yellow!50!orange]{\tikz[baseline=-0.5ex]\draw[#1,radius=#2] (0,0) circle ;}
\newcommand{\blackc}[2][black,fill=black]{\tikz[baseline=-0.5ex]\draw[#1,radius=#2] (0,0) circle ;}
\newcommand{\brownc}[2][cyan,fill=cyan]{\tikz[baseline=-0.5ex]\draw[#1,radius=#2] (0,0) circle ;}
\def\fracture#1#2{\ensuremath{#1\raisebox{.2ex}{\rotatebox[origin=c]{-15}{$\sharp$}}#2}}

\usetikzlibrary{conference}
\usetikzlibrary{cd}

\title{
Parameterised and Fine-grained Subgraph Counting, modulo $2$\thanks{For the purpose of Open Access, the
authors have applied a CC BY public copyright licence to any Author Accepted Manuscript version arising
from this submission. All data is provided in full in the results section of this paper.}}

\titlerunning{Parameterised and Fine-grained Subgraph Counting, modulo $2$} 

\author{Leslie Ann Goldberg}{Department of Computer Science, University of Oxford}{}{}{}

\author{Marc Roth}{Department of Computer Science, University of Oxford}{}{}{}

\authorrunning{L.\,A. Goldberg and M. Roth} 

\Copyright{Leslie Ann Goldberg and Marc Roth} 



\category{} 

\relatedversion{An extended abstract of this work, not containing the proof of the classification for trees, is accepted for publication at the 50th EATCS International Colloquium on Automata, Languages and Programming (ICALP 2023).} 

\acknowledgements{We want to thank Radu Curticapean, Holger Dell and Thore Husfeldt for insightful discussions on an early draft of this work.}

\begin{document}

\maketitle

\begin{abstract}
Given a class of graphs $\calH$, the problem $\oplus\subsprob(\calH)$ is defined as follows. The input is  a graph  $H\in \calH$ together with an arbitrary graph $G$.
The problem is to compute, 
modulo $2$, the number of subgraphs of $G$ that are isomorphic to $H$. 
The goal of this research is to determine for which classes  $\calH$ the problem 
$\oplus\subsprob(\calH)$
is fixed-parameter tractable (FPT), i.e., solvable in time $f(|H|)\cdot |G|^{O(1)}$.

Curticapean, Dell, and Husfeldt (ESA 2021) conjectured that $\oplus\subsprob(\calH)$ is FPT if and only if the class of allowed patterns $\calH$ is \emph{matching splittable}, 
which means that for some fixed $B$, every $H \in \calH$ can be turned into a matching (a graph in which every vertex has degree at most~$1$) by removing at most~$B$ vertices. 

Assuming the randomised Exponential Time Hypothesis, we prove their conjecture for (I) all hereditary pattern classes $\calH$, and (II) all tree pattern classes, i.e., all classes $\calH$ 
such that every $H\in \calH$ is a tree. 

We also establish almost tight fine-grained upper and lower bounds for the case of hereditary patterns (I).
\end{abstract}
\section{Introduction}
The last two decades have seen remarkable progress in the classification of subgraph counting problems: Given a small \emph{pattern} graph $H$ and a large \emph{host} graph $G$, how often does $H$ occur as a subgraph if $G$? Since it was discovered that subgraph counts from small patterns reveal global properties of complex networks~\cite{Miloetal02,Miloetal04}, subgraph counting has also found several applications in fields such as biology~\cite{Nogaetat08,Schilleretal15} genetics~\cite{Tranetal13}, phylogeny~\cite{Kuchaievetal10}, and data mining~\cite{Babis17}. Moreover, the theoretical study of subgraph counting and related problems has led to many deep structural insights, establishing both new algorithmic techniques and tight lower bounds under the lenses of fine-grained and parameterised complexity theory~\cite{FlumG04,DalmauJ04,ChenTW08,CurticapeanM14,CurticapeanDM17,Bressan21,BeraGLSS22}.
 
Without any additional restrictions, the subgraph counting problem is infeasible.
The complexity class ~$\#\W{1}$ is the parameterised complexity class analogous to NP (see Section~\ref{sec:prelims} for more detail). Under standard assumptions,
problems that are $\#\W{1}$-hard are not \emph{fixed-parameter tractable} (FPT). The canonical complete problem for~$\#\W{1}$, the problem of counting $k$-cliques, corresponds to the 
special case of the subgraph counting problem where $H$ is a clique of size~$k$.
This problem   cannot be solved in time $f(k)\cdot n^{o(k)}$ for any function $f$ unless the Exponential Time Hypothesis (ETH) fails~\cite{Chenetal05,Chenetal06}. 
Due to this hardness result, the research focus in this area  shifted to the question: Under which restrictions on the patterns $H$ and the hosts $G$ is algorithmic progress possible? More precisely, under which restrictions can the problem be solved in time $f(|H|)\cdot |G|^{O(1)}$, for some computable function $f$? Instances that can be solved within such a run time bound are called \emph{fixed-parameter tractable} (FPT); allowing a potential super-polynomial overhead in the size of the pattern $|H|$ formalises the assumption that $H$ is assumed to be (significantly) smaller than $G$.

If only the patterns are restricted, then the situation is fully understood. Formally, given a class $\calH$ of patterns, the problem $\#\subsprob(\calH)$ asks, given as input a graph $H\in \calH$ and an arbitrary graph $G$, to compute the number of subgraphs of $G$ that are isomorphic to $H$. Following initial work by Flum and Grohe~\cite{FlumG04} and by Curticapean~\cite{Curticapean13}, Curticapean and Marx~\cite{CurticapeanM14} proved that, under standard assumptions, $\#\subsprob(\calH)$ is FPT if and only if $\calH$ has bounded matching number, that is, if there is a positive integer $B$ such that the size of any matching in any graph in $\calH$ is at most $B$. They also proved that all FPT cases are polynomial-time solvable.

In stark contrast, almost nothing is known for the decision version $\subsprob(\calH)$. Here, the task is to correctly decide whether there is a copy of $H\in \calH$ in $G$, rather than to count the copies. It is known that $\subsprob(\calH)$ is FPT whenever $\calH$ has bounded treewidth (see e.g.\ \cite[Chapter 13]{FlumG06}), and it is conjectured that those are all FPT cases. However, resolving this conjecture belongs to the ``most infamous'' open problems in parameterised complexity theory~\cite[Chapter 33.1]{DowneyF13}.

\subsection{Counting Modulo 2}
To interpolate between the fully understood realm of (exact) counting and the barely understood realm of decision, Curticapean, Dell and Husfeldt proposed the study of counting subgraphs, modulo $2$~\cite{CurticapeanDH21}. Formally, they introduced the problem $\oplus\subsprob(\calH)$, which expects as input a graph $H\in \calH$ and an arbitrary graph $G$, and the goal is to compute \emph{modulo 2} the number of subgraphs of $G$ isomorphic to $H$.

The study of counting modulo $2$ received significant attention from the viewpoint of classical, structural, and fine-grained complexity theory. For example, one way to state Toda's Theorem~\cite{Toda91} is $\mathrm{PH}\subseteq \mathrm{P}^{\oplus \mathrm{P}}$, implying that counting satisfying assignments of a CNF, modulo $2$, is at least as hard as the polynomial hierarchy. Another example is the quest to classify the complexity of counting modulo $2$ the homomorphisms to a fixed graph, which was very recently resolved by Bulatov and Kazeminia~\cite{BulatovK22}. There has also been work 
by Abboud, Feller, and Weimann~\cite{AbboudFW20}
on the fine-grained complexity of counting modulo~2 the number of 
triangles in a graph that satisfy certain weight constraints. 

In their work~\cite{CurticapeanDH21}, Curticapean, Dell and Husfeldt proved that the problem of counting $k$-matchings modulo $2$, that is, the problem $\oplus\subsprob(\calH)$ where $\calH$ is the class of all $1$-regular graphs, is fixed-parameter tractable, where the parameter~$k$ is $|H|$. 
Since the exact counting version of this problem is $\#\W{1}$-hard~\cite{Curticapean13},  their result 
provides an example where counting modulo $2$  is \emph{strictly} easier than exact counting (subject to complexity assumptions).
The complexity class $\oplus\W{1}$ can be defined via the complete problem of counting $k$-cliques modulo~$2$.
Crucially, $\oplus\W{1}$-hard problems are not fixed-parameter tractable, unless the randomised ETH (rETH) fails.
Curticapean et al.~\cite{CurticapeanDH21} proved that counting $k$-paths modulo $2$ is 
$\oplus\W{1}$-hard. 
Since \emph{finding} a $k$-path in a graph $G$ is fixed-parameter tractable via colour-coding~\cite{AlonYZ95}, this hardness result 
provides an example where 
counting modulo $2$  is \emph{strictly} harder than decision (subject to complexity assumptions).
Combining those observations, it appears that
counting subgraphs modulo $2$ may lie strictly in between the complexity of decision and the complexity of exact counting.

A \emph{matching} is a graph whose maximum degree is at most~$1$.
The \emph{matching-split number} of a graph $H$ is the minimum size of a set $S\subseteq V(H)$ such that $H\setminus S$ is a matching. A class of graphs $\calH$ is called \emph{matching splittable} if there is a positive integer $B$ such that the matching-split number of any $H\in\calH$ is at most $B$. For example, the class of all matchings   is matching splittable while the class of all cycles is not. 
Curticapean, Dell and Husfeldt extended their 
FTP algorithm for counting $k$-matchings modulo~$2$ to obtain an FPT algorithm for
$\oplus\subsprob(\calH)$  for any matching-splittable class $\calH$. On this basis, 
they   then made the following conjecture.

\begin{conjecture}[\cite{CurticapeanDH21}]\label{conj:subsmod2}
    $\oplus\subsprob(\calH)$ is FPT if and only if $\calH$ is matching splittable.
\end{conjecture}

A class $\mathcal{H}$ of graphs is called \emph{hereditary}
 if it is closed under vertex removal.
Intriguingly, if Conjecture~\ref{conj:subsmod2} is true, then the FPT criterion for counting subgraphs modulo $2$ ($\oplus\subsprob(\calH)$) would coincide with the polynomial-time criterion for finding subgraphs ($\subsprob(\calH)$) for hereditary pattern classes $\calH$ as established by Jansen and Marx.

\begin{theorem}[\cite{JansenM15}]\label{thm:decsubspoly}
Let $\calH$ be a hereditary class of graphs and assume $\mathrm{P}\neq \mathrm{NP}$. Then
    $\subsprob(\calH)$ is solvable in polynomial time if and only if $\calH$ is matching splittable.
\end{theorem}

\noindent Jansen and Marx also conjecture that the condition of $\calH$ being hereditary can~be~removed.
\begin{conjecture}[\cite{JansenM15}]\label{conj:decsubspoly}
    $\subsprob(\calH)$ is solvable in polynomial time if and only if $\calH$ is matching splittable.
\end{conjecture}

Conjectures~\ref{conj:subsmod2} and~\ref{conj:decsubspoly} have the remarkable consequence that $\oplus\subsprob(\calH)$ is FPT if and only if $\subsprob(\calH)$ is solvable in polynomial time. 
In the current work we establish this consequence
for all hereditary pattern classes.

\subsection{Our Contributions}\label{sec:contributions}
We resolve Conjecture~\ref{conj:subsmod2} for all hereditary classes $\calH$, as well as for every class $\calH$ consisting only of trees; note that the upper bounds were shown in~\cite{CurticapeanDH21} and that the lower bounds are the novel part.

\begin{restatable}{theorem}{mainhereditary}\label{thm:main_hereditary}
Let $\mathcal{H}$ be a hereditary class of graphs. If $\mathcal{H}$ is matching splittable, then $\oplus\subsprob(\mathcal{H})$ is fixed-parameter tractable. Otherwise, the problem is $\oplus\W{1}$-complete and, assuming rETH, cannot be solved in time $f(|H|)\cdot |G|^{o(|V(H)|/\log |V(H)|)}$ for any function $f$.
\end{restatable}

\begin{restatable}{theorem}{main}\label{thm:main}
Let $\mathcal{T}$ be a recursively enumerable class of trees. If $\mathcal{T}$ is matching splittable, then $\oplus\subsprob(\calT)$ is fixed-parameter tractable. Otherwise $\oplus\subsprob(\calT)$ is $\oplus\W{1}$-complete.
\end{restatable}

The requirement that the class of trees $\mathcal{T}$ needs to be recursively enumerable is a standard technicality - the reason for it is that the function~$f$ in the running time in the standard definition of an FPT algorithm is required to be computable. It turns out that having $\mathcal{T}$ recursively enumerable is enough for this.

In order to prove our classifications, we   adapt the by-now-standard technique for analysing subgraph counting problems established by Curticapean, Dell and Marx~\cite{CurticapeanDM17}. 
Let $\#\subs{H}{G}$ denote the number of subgraphs of a graph~$G$ that are isomorphic to a graph~$H$ and let $\#\homs{F}{G}$ denotes the number of homomorphisms (edge-preserving mappings) from a graph~$F$ to a graph~$G$. 
Given a graph~$H$, there is a function $a_H$ from graphs to rationals with finite support such that the following holds for any graph $G$:
\begin{equation}\label{eq:intro_hombasis}
    \#\subs{H}{G} = \sum_F a_H(F) \cdot \#\homs{F}{G}\,,
\end{equation}
where the sum is over all (isomorphism types of) graphs. 
Since $a_H$   has finite support,   $a_H(F)=0$ for all but finitely-many graphs~$F$.
Thus, equation~\eqref{eq:intro_hombasis} allows us to express the solution to the
\emph{exact} counting problem  
as a finite linear combination of homomorphism counts. In a nutshell, the framework of~\cite{CurticapeanDM17} states that computing the function $G \mapsto \#\subs{H}{G}$ is hard to compute if and only if there is a graph $F$ of high treewidth with $a_H(F)\neq 0$. This translates the complexity of (exact) subgraph counting to the purely combinatorial problem of understanding the coefficients $a_H$. One might hope that this strategy transfers to counting modulo $2$ as well. Unfortunately, this is not possible as Equation~\eqref{eq:intro_hombasis} might not be well-defined if arithmetic is done modulo $2$. The reason for this is the fact that the coefficients $a_H(F)$ are of the form $\mu(F,H) \times |\auts{H}|^{-1}$, where $\mu(F,H)$ is an integer, and $\auts{H}$ is the automorphism group of the graph $H$~\cite{CurticapeanDM17}. Thus there is, a priori, no hope to extend the framework to counting modulo $2$ for pattern graphs with an even number of automorphisms. In fact, according to Curticapean, Dell and Husfeldt~\cite{CurticapeanDH21}, the absence of a comparable framework for counting modulo $2$ is one of the main challenges for establishing the hardness part of Conjecture~\ref{conj:subsmod2}, and it is the main reason why the reductions in~\cite{CurticapeanDH21} use more classical, gadget-based reductions.

In this work, we solve the problem of patterns with an even number of automorphisms by considering a colourful intermediate problem. More concretely, we will equip each edge of the pattern $H$ with a distinct colour and show that it will be sufficient to consider only automorphisms that preserve the colours.
If $H$ has no isolated vertices, then this is only the trivial automorphism.  Formally, the coloured approach will be based on the notion of so-called \emph{fractured} graphs introduced by Peyerimhoff et al.\ \cite{PeyerimhoffRSSVW22}.

\subsection*{Organisation of the Paper}
We start by introducing some basic terminology in Section~\ref{sec:prelims}. The formal definitions of our graph colourings, as well as colour-preserving homomorphisms and embeddings can be found in Section~\ref{sec:col_prelims}, and the majority of the paper will consider the coloured setting as it allows us to get rid of automorphism groups of even size. This is formalised in Section~\ref{sec:fractures} using the framework of \emph{fractured graphs} originally introduced in~\cite{PeyerimhoffRSSVW22}. An introduction to parameterised and fine-grained complexity theory, including the definition of our computational problems and the statement of the randomised Exponential Time Hypothesis, can be found in Section~\ref{sec:param}; moreover, this section contains a self-contained and formal exposition of the complexity monotonicty principle for coloured graphs in the modular setting, stating intuitively that the computation, modulo $2$, of a finite linear combination of homomorphism counts between coloured graphs is precisely as hard as computing, modulo $2$, the hardest term with an odd coefficient. Additionally, Section~\ref{sec:param} contains the formal statement of the reduction from the coloured setting to the uncoloured setting via the principle of inclusion and exclusion. Note that this reduction is necessary for obtaining our main results (Theorems~\ref{thm:main_hereditary} and~\ref{thm:main}), which classify the complexity of the uncoloured problem $\oplus\subsprob(\mathcal{H})$.

Having completed the set-up, we continue in Section~\ref{sec:hereditary} with the treatment of $\oplus\subsprob(\mathcal{H})$ for hereditary $\mathcal{H}$, i.e., with the proof of Theorem~\ref{thm:main_hereditary}. We note that, on a technical level, understanding the hereditary case is much easier than the case of trees. However, almost all of the key techniques and ideas that become necessary to classify the case of trees are already used in Section~\ref{sec:hereditary}, although in a much simpler way. For this reason, we consider Section~\ref{sec:hereditary} also as a warm-up for getting used to the framework of fractured graphs. Concretely, we can outline our treatment of hereditary classes as follows: Using a result of Jansen and Marx~\cite{JansenM15}, each hereditary class of graphs $\mathcal{H}$ is either matching splittable, or it fully contains one of the following four subclasses: (I) The class of all cliques, (II) the class of all bicliques, (III) the class of all triangle packings (disjoint unions of triangles), or (IV) the class of all $P_2$-packings (disjoint unions of paths with two edges). For proving the classification of $\oplus\subsprob(\mathcal{H})$ for hereditary $\calH$ (Theorem~\ref{thm:main_hereditary}), it thus suffices to show that each of the four cases (I) - (IV) is hard. Since the problems of \emph{deciding} whether a graph contains a $k$-clique or whether a graph contains a $k$-by-$k$-biclique are already hard, the problem of counting their respective occurences modulo $2$ (cases (I) and (II)) can easily shown to be hard using a variation of the Isolation Lemma due to Williams et al.\ \cite{WilliamsWWY15}. The majority of Section~\ref{sec:hereditary} is thus dedicated to establishing hardness for triangle packings (III) in Section~\ref{sec:triangle} and for $P_2$-packings (IV) in Section~\ref{sec:p2_pack}. 

The classification of $\oplus\subsprob(\mathcal{T})$ for classes of trees $\mathcal{T}$ (Theorem~\ref{thm:main}) can be found in Section~\ref{sec:trees}. In the first step, we establish a graph-theoretical classification of classes of trees that are not matching splittable. To this end, we first introduce three structural invariants of trees (the definitions are rather technical and can be found right at the beginning of Section~\ref{sec:trees}): The \emph{fork number}, the \emph{star number}, and the \emph{$\hal$-number}. We then show that each class $\calT$ of trees is either matching splittable, or it satisfies at least one of the following properties:
\begin{itemize}
    \item[(1)] $\calT$ has unbounded $\hal$-number,
    \item[(2)] $\calT$ has unbounded star number, or
    \item[(3)] $\calT$ has unbounded fork number.
\end{itemize}
The central steps of the proof of Theorem~\ref{thm:main} are then hardness proofs for the previous three cases: Case (1) is treated in Section~\ref{sec:caseH}, Case (2) is treated in Section~\ref{sec:caseStar}, and Case (3) is treated in Section~\ref{sec:caseF}. Finally, we collect the intractabilty results for all cases in Section~\ref{sec:collect_trees} to prove Theorem~\ref{thm:main}.

\section{Preliminaries}\label{sec:prelims}

Let $f:A_1\times A_2 \rightarrow B$ be a function. For each $a_1\in A_1$ we write $f(a_1,\star): A_{2} \to B$ for the function that maps $a_2\in A_{2}$ to $f(a_1,a_2)$.

Graphs in this work are undirected and without self loops. A \emph{homomorphism} from a graph $H$ to a graph $G$ is a mapping $\varphi$ from the vertices $V(H)$ of $H$ to the vertices $V(G)$ of $G$ such that for each edge $e=\{u,v\}\in E(H)$ of $H$, the image $\varphi(e)=\{\varphi(u),\varphi(v)\}$ is an edge of $G$. A homomorphism is called an \emph{embedding} if it is injective. We write $\homs{H}{G}$ and $\embs{H}{G}$ for the sets of homomorphisms and embeddings, respectively, from $H$ to $G$. An embedding $\varphi\in\embs{H}{G}$ is called an \emph{isomorphism} if it is bijective and $\{u,v\}\in E(H) \Leftrightarrow \{\varphi(u),\varphi(v)\}\in E(G)$. We say that $H$ and $G$ are \emph{isomorphic}, denoted by $H\cong G$, if an isomorphism from $H$ to $G$ exists. A \emph{graph invariant} $\iota$ is a function from graphs to rationals such that $\iota(H)=\iota(G)$ for each pair of isomorphic graphs $H$ and $G$.

A \emph{subgraph} of $G$ is a graph $G'$ with $V(G')\subseteq V(G)$ and $E(G')\subseteq E(G)$. We write $\subs{H}{G}$ for the set of all subgraphs of $G$ that are isomorphic to $H$. Given a subset of vertices $S\subseteq V(G)$ of a graph $G$, we write $G[S]$ for the graph induced by $S$, that is, $G[S]$ has vertices $S$ and edges $\{ \{u,v\} \subseteq S \mid \{u,v\} \in E(G)\}$.

We denote by $\mathsf{tw}(G)$ the \emph{treewidth} of the graph $G$. Since we will rely on treewidth purely in a black-box manner, we omit the technical definition and refer the reader to~\cite[Chapter 7]{CyganFKLMPPS15}.

Given any graph invariant $\iota$ 
(such as treewidth) and a class of graphs $\mathcal{G}$, we say that $\iota$ is \emph{bounded} in $\mathcal{G}$ if there is a non-negative integer $B$ such that, for all $G\in \mathcal{G}$, 
$\iota(G)\leq B$. 
Otherwise we say that $\iota$ is \emph{unbounded} in $\mathcal{G}$.

Given a graph $H=(V,E)$, a \emph{splitting set} of $H$ is a subset of vertices $S$ such that 
every vertex in
$H[V\setminus S]$ 
has degree at most~$1$.  The \emph{matching-split number} of $H$ is the minimum size of a splitting set of $H$. A class of graphs $\calH$ is called \emph{matching splittable} if the matching-split number of $\calH$ is bounded.

\subsection{Colour-Preserving Homomorphisms and Embeddings}\label{sec:col_prelims}
A homomorphism~$c$ from a graph~$G$ to a graph~$Q$ is sometimes called a ``$Q$-colouring'' of~$G$.
A $Q$-\emph{coloured} graph is a pair consisting of a graph $G$ and a homomorphism $c$  from~$G$ to~$Q$.  
Note that the identity function $\mathsf{id}_Q$ on $V(Q)$ is a $Q$-colouring of~$Q$. If 
a homomorphism~$c$ from~$G$ to~$Q$
is vertex surjective, then we call $(G,c)$ a \emph{surjectively} $Q$-coloured graph.

\begin{definition}[$c_{\text{E}}$]\label{def:cE}
 \label{rem:Qcol_is_Ecol}
A $Q$-colouring $c$ of a graph $G$ induces a (not necessarily proper) edge-colouring $c_{\text{E}}\colon E(G) \to E(Q)$ given by 
$c_{\text{E}}(\{u,v\}) = \{c(u),c(v)\}$. 
\end{definition}

\noindent \textbf{Notation:} Given a $Q$-coloured graph $(G,c)$ and a vertex $u\in V(Q)$, we will use the capitalised letter $U$ to denote the subset of vertices of $G$ that are coloured by $c$ with $u$, that is, $U:=c^{-1}(u)\subseteq V(G)$. 

\medskip

Given two $Q$-coloured graphs $(H,c_H)$ and $(G,c_G)$, we call a homomorphism $\varphi$ from~$H$ to~$G$ \emph{colour-preserving} if for each $v\in V(H)$ we have $c_G(\varphi(v))=c_H(v)$. We note the special case in which $Q=H$ and $c_H$ is the identity~$\mathsf{id}_Q$; then the condition simplifies to $c_G(\varphi(v))=v$.
A colour-preserving embedding of~$(H,c_H)$ in~$(G,c_G)$
is a vertex injective colour-preserving homomorphism
from $(H,c_H)$ to $(G,c_G)$. We write $\homs{(H,c_H)}{(G,c_G)}$ and $\embs{(H,c_H)}{(G,c_G)}$ for the sets of all colour-preserving homomorphisms and embeddings, respectively, from $(H,c_H)$ to $(G,c_G)$.

Let $k$ be a positive integer, let $H$ be a graph with $k$ edges, and let $(G,\gamma)$ be a pair consisting of a graph $G$ and a 
function that  maps each edge of $G$ to one of $k$ distinct colours.
We refer to $\gamma$ as a ``$k$-edge colouring'' of~$G$. 
For example, in most of our applications 
we will fix a graph $Q$ with $k$~edges and a $Q$-colouring $c$ of~$G$
and we will take $\gamma$ to be the edge-colouring $c_E$ from Definition~\ref{def:cE}.  We write $\colsubs{H}{(G,\gamma)}$ for the set of all subgraphs of $G$ that are isomorphic to $H$ and that contain each of the $k$ edge colours precisely once.

\subsection{Fractures and Fractured Graphs}\label{sec:fractures}
In this work, we will crucially rely on and extend the framework of \emph{fractured} graphs as introduced in~\cite{PeyerimhoffRSSVW22}.

\begin{definition}[Fractures]\label{def:fracGeneral}
Let $Q$~be a graph. For each vertex $v$ of~$Q$, let $E_Q(v)$ be the set of edges of~$Q$ that are incident to~$v$.
A \emph{fracture} of $Q$ is a tuple $\rho=(\rho_v)_{v\in
    V(Q)}$, where for each vertex $v$ of~$Q$, $\rho_v$ is a partition of   $E_Q(v)$.
\end{definition}

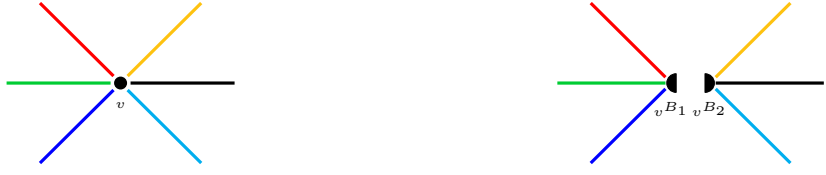
\begin{figure}[t]
    \centering
    \begin{tikzpicture}[scale=1.5]
        \node[vertex,inner sep=.4ex,label={[label distance=.03]below:\(v\)}] (m) at (0, 0) {};

        \draw[very thick,red] (m) -- ++(135:1);
        \draw[very thick,green!80!blue] (m) -- ++(180:1);
        \draw[very thick,blue] (m) -- ++(-135:1);

        \draw[very thick,yellow!50!orange] (m) -- ++(45:1);
        \draw[very thick] (m) -- ++(0:1);
        \draw[very thick,cyan] (m) -- ++(-45:1);

        \begin{scope}[shift={(5,0)}]
            \begin{scope}[scale=2.4]
                \kowaen{0,0}{-90/90/white,90/270/white}{1};
            \end{scope}
            \node[label={[label distance=.03]below:\(v^{B_1}\)}]  at (1-2) {};
            \draw[very thick,red] (1-2) -- ++(135:1);
            \draw[very thick,green!80!blue] (1-2) -- ++(180:1);
            \draw[very thick,blue] (1-2) -- ++(-135:1);
            \node[label={[label distance=.03]below:\(v^{B_2}\)}]  at
                (1-1) {};
            \draw[very thick,yellow!50!orange] (1-1) -- ++(45:1);
            \draw[very thick] (1-1) -- ++(0:1);
            \draw[very thick,cyan] (1-1) -- ++(-45:1);
            \begin{scope}[scale=1.8]
                \kowaen{0,0}{-90/90/black,90/270/black}{1};
            \end{scope}
        \end{scope}
    \end{tikzpicture}
    \caption{\label{fig:fracture} Illustration of the construction of a
        fractured graph from~\cite{PeyerimhoffRSSVW22}. The left picture shows a vertex $v$ of a graph~$Q$ with incident
        edges $E_Q(v)=\{
        \redc{2pt},\greenc{2pt},\bluec{2pt},\yellowc{2pt},\blackc{2pt},\brownc{2pt}\}$. The right
        picture shows the splitting of $v$ in the construction of the fractured
        graph~$\fracture{Q}{\sigma}$ for a fracture $\sigma$ satisfying that the partition
        $\sigma_v$ contains two blocks $B_1 =\{ \redc{2pt},\greenc{2pt},\bluec{2pt}\}$, and
    $B_2=\{\yellowc{2pt},\blackc{2pt},\brownc{2pt}\}$.}
\end{figure}

Note that a fracture describes how to split (or how to \emph{fracture}) each vertex of a
given graph: for each vertex~$v$, create a vertex $v^B$ for each block $B$ in the partition
$\rho_v$; edges originally incident to $v$ are made incident to $v^B$ if and only if they
are contained in $B$. We call the resulting graph the \emph{fractured graph $\fracture{H}{\rho}$};
a formal definition is given in Definition~\ref{def:mrhoGeneral}, a visualisation is given in
Figure~\ref{fig:fracture}.

\begin{definition}[Fractured Graph $\fracture{Q}{\rho}$]\label{def:mrhoGeneral}
    Given a graph $Q$, we consider the matching $M_Q$ containing one edge for each edge of $Q$;
    formally,
    \[
        V(M_Q) := \bigcup_{e = \{u,v\} \in E(Q)} \{ u_e, v_e \}\quad\text{and}\quad
        E(M_Q) := \{ \{ u_e, v_e \} \mid e = \{u,v\} \in E(Q)\}.
    \]

    For a fracture $\rho$ of $Q$, we define the graph $\fracture{Q}{\rho}$ to be the quotient
    graph of $M_Q$ under the equivalence relation on $V(M_Q)$ which identifies two
    vertices $v_e, w_f$ of $M_Q$ if and only if $v=w$ and $e,f$ are in the same block $B$
    of the partition $\rho_v$ of $E_Q(v)$. We write $v^B$ for the vertex of
    $\fracture{Q}{\rho}$ given by the equivalence class of the vertices $v_e$ (for which
    $e \in B$) of $M_Q$.
\end{definition}

\begin{definition}[Canonical $Q$-colouring $c_\rho$]\label{def:fracture_Q_cols}
Let $Q$ be a graph and let $\rho$ be a fracture of $Q$. The \emph{canonical} $Q$-\emph{colouring} of the fractured graph $\fracture{Q}{\rho}$ maps $v^B$ to $v$ for each $v\in V(Q)$ and block $B\in \rho_v$, and is denoted by $c_\rho$.
\end{definition}
Observe that $c_\rho$ is the identity in $V(Q)$ if $\rho$ is the coarsest fracture (that is, each partition $\rho_v$ only contains one block, in which case $\fracture{Q}{\rho}=Q$).

\subsection{Parameterised and Fine-grained Computation}\label{sec:param}
A \emph{parameterised computational problem} is a pair consisting of a function $P:\Sigma^\ast \to \{0,1\}$ and a computable parameterisation $\kappa: \Sigma^\ast \to \mathbb{N}$. A \emph{fixed-parameter tractable} (FPT) \emph{algorithm} for $(P,\kappa)$ is an algorithm that computes $P$ and runs, on input $x\in \Sigma^\ast$, in time $f(\kappa(x))\cdot |x|^{O(1)}$ for some computable function $f$. We call $(P,\kappa)$ \emph{fixed-parameter tractable} (FPT) if an FPT algorithm for $(P,\kappa)$ exists.

A parameterised Turing-reduction from $(P,\kappa)$ to $(P',\kappa')$ is an FPT algorithm for $(P,\kappa)$ that is equipped with oracle access to $P'$ and for which there is a computable function $g$ such that, on input $x$, each oracle query $y$ satisfies $\kappa'(y)\leq g(\kappa(x))$.  We write $(P,\kappa)\fptred (P',\kappa')$ if a parameterised Turing-reduction from $(P,\kappa)$ to $(P',\kappa')$ exists. This guarantees
that fixed-parameter tractability of $(P',\kappa')$ implies fixed-parameter tractability of $(P,\kappa)$.
For a more comprehensive introduction, we refer the reader the standard textbooks~\cite{CyganFKLMPPS15} and~\cite{FlumG06}.

\paragraph*{Counting modulo 2 and the rETH}
The lower bounds in this work will rely on the hardness of the parameterised complexity class $\oplus\W{1}$, which can be considered a parameterised equivalent of $\oplus\mathrm{P}$. Following~\cite{CurticapeanDH21}, we define $\oplus\W{1}$ via the complete problem $\oplus\textsc{Clique}$: Given as input a graph $G$ and a positive integer $k$, the goal is to compute the number of $k$-cliques in $G$ modulo $2$, i.e., to compute $\oplus\subs{K_k}{G}$. The problem is parameterised by $k$. 
A parameterised problem $(P,\kappa)$ is called $\oplus\W{1}$\emph{-hard} if $\oplus\textsc{Clique}\fptred (P,\kappa)$, and it is called $\oplus\W{1}$\emph{-complete} if, additionally, $(P,\kappa) \fptred \oplus\textsc{Clique}$.

Modifications of the classical Isolation Lemma (see e.g.\ \cite{BjorklundDH15} and~\cite{WilliamsWWY15}) yield a \emph{randomised} parameterised Turing reduction from finding a $k$-clique to computing the parity of the number of $k$-cliques. In combination with existing fine-grained lower bounds for finding a $k$-clique~\cite{Chenetal05,Chenetal06}, it can then be shown that $\oplus\textsc{Clique}$ cannot be solved in time $f(k) \cdot |G|^{o(k)}$ for any function $f$, unless the randomised Exponential Time Hypothesis fails:
\begin{definition}[rETH,~\cite{ImpagliazzoP01}]
    The \emph{randomised Exponential Time Hypothesis} (rETH) asserts that $3$-\textsc{SAT} cannot be solved by a randomised algorithm in time $\exp{o(n)}$, where $n$ is the number of variables of the input formula.
\end{definition}
As an immediate consequence, the rETH implies that $\oplus\W{1}$-hard problems are not fixed-parameter tractable. 

For the lower bounds in this work, we won't reduce from $\oplus\textsc{Clique}$ directly, but instead from the following, more general problem:
\begin{definition}[$\oplus\cphomsprob$]
Let $\mathcal{H}$ be a class of graphs. The problem $\oplus\cphomsprob(\mathcal{H})$ 
has as input   a graph $H\in\mathcal{H}$ and a surjectively $H$-coloured graph $(G,c)$. The goal is to compute $\oplus\homs{(H,\mathsf{id}_H)}{(G,c)}$. The problem is parameterised by $|H|$.
\end{definition}

The following lower bound was proved independently in~\cite{PeyerimhoffR0SV21MFCS,PeyerimhoffRSSVW22} and~\cite{CurticapeanDH21}.
\begin{theorem}\label{thm:cphom_lower_bound}
Let $\mathcal{H}$ be a recursively enumerable class of graphs. If the treewidth of $\mathcal{H}$ is unbounded then $\oplus\cphomsprob(\mathcal{H})$ is $\oplus\W{1}$-hard and, assuming the rETH, it cannot be solved in time $f(|H|)\cdot |G|^{o(\mathsf{tw}(H)/\log \mathsf{tw}(H))}$ for any function $f$.
\end{theorem}

Next is the central problem in this work.
\begin{definition}[$\oplus\subsprob$]
Let $\mathcal{H}$ be a class of graphs. The problem $\oplus\subsprob(\mathcal{H})$ 
has as input   a graph $H\in\mathcal{H}$ and a graph $G$. The goal is to compute $\oplus\subs{H}{G}$. The problem is parameterised by $|H|$.
\end{definition}

For example, writing $\mathcal{K}$ for the set of all complete graphs, the problem $\oplus\subsprob(\mathcal{K})$ is equivalent to $\oplus\textsc{Clique}$.

\paragraph*{Complexity Monotonicity and Inclusion-Exclusion}
Throughout this work, we will rely on two important tools introduced in~\cite{PeyerimhoffRSSVW22}. For the sake of being self-contained, we encapsulate them below in individual lemmas.

The first tool is an adaptation of the so-called Complexity Monotonicity principle to the realm of fractured graphs and modular counting (see \cite[Sections 4.1 and 6.3]{PeyerimhoffRSSVW22} for a detailed treatment and for a proof). Intuitively, the subsequent lemma states that evaluating, modulo~$2$, a linear combination of colour-prescribed homomorphism counts from fractured graphs, is as hard as evaluating its hardest term with an odd coefficient.
\begin{lemma}[\cite{PeyerimhoffRSSVW22}]\label{lem:complexity_monotonicty}
    There is a deterministic algorithm $\mathbb{A}$ and a computable function $f$ such that the following conditions are satisfied:
    \begin{enumerate}
        \item $\mathbb{A}$ expects as input a graph $Q$ and a $Q$-coloured graph $(G,c)$.
        \item $\mathbb{A}$ is equipped with oracle access to a function
        \[ (G',c') \mapsto \sum_\rho a(\rho) \cdot \oplus\homs{(\fracture{Q}{\rho},c_\rho)}{(G',c')} \mod 2  \,,\]
        where the sum is over all fractures of $Q$ and $a$ is a function from fractures of $Q$ to integers. 
        \item Each oracle query $(G',c')$ is of size at most $f(|Q|)\cdot |G|$.  
        \item $\mathbb{A}$ computes $\oplus\homs{(\fracture{Q}{\rho},c_\rho)}{(G,c)}$ for each fracture $\rho$ with $a(\rho)\neq 0 \mod 2$.
        \item The running time of $\mathbb{A}$ is bounded by $f(|Q|)\cdot |G|^{O(1)}$.
    \end{enumerate}
\end{lemma}

\medskip 

\noindent The second tool is a standard application of the inclusion-exclusion principle (see e.g. \cite[Sections 4.2 and 6.3]{PeyerimhoffRSSVW22}). It will be used in the final steps of our reductions to remove the colourings.
\begin{lemma}[\cite{PeyerimhoffRSSVW22}]\label{lem:remove_edgecols}
    There is a deterministic algorithm $\mathbb{A}$ that satisfies the following conditions:
    \begin{enumerate}
        \item $\mathbb{A}$ expects as input a graph $H$ with $k$ edges, a graph $G$ and a $k$-edge colouring $\gamma$ of $G$.
        \item $\mathbb{A}$ is equipped with oracle access to the function $\oplus\subs{H}{\star}$, and each oracle query $G'$ satisfies $|G'|\leq |G|$.
        \item $\mathbb{A}$ computes $\oplus\colsubs{H}{(G,\gamma)}$.
        \item The running time of $\mathbb{A}$ is bounded by $2^{|H|}\cdot |G|^{O(1)}$.
    \end{enumerate}
\end{lemma}

\section{Classification for Hereditary Graph Classes}\label{sec:hereditary}
In this section, we will completely classify the complexity of $\oplus\subsprob(\mathcal{H})$ for hereditary classes. Let us start by restating the classification theorem.

\mainhereditary*

The proof of Theorem~\ref{thm:main_hereditary} is split in four cases, which stem from a structural property of non matching splittable hereditary graph classes $\mathcal{H}$ due to Jansen and Marx~\cite{JansenM15}. For the statement, we need to consider the following classes:
\begin{itemize}
    \item $\mathcal{F}_\omega$ is the class of all complete graphs.
    \item $\mathcal{F}_\beta$ is the class of all complete bipartite graphs.
    \item $\mathcal{F}_{P_2}$ is the class of all $P_2$-packings, that is, disjoint unions of paths with two edges.\footnote{To avoid confusion, we remark that~\cite{JansenM15} uses $P_3$ to denote the path of two edges (and three vertices). In the current work, it will be more convenient to use the number of edges of a path as index.}
    \item $\mathcal{F}_{K_3}$ is the class of all triangle packings, that is, disjoint unions of the complete graph of size $3$.
\end{itemize}

\begin{theorem}[Theorem 3.5 in~\cite{JansenM15}]\label{thm:matchsplit_characterisation}
Let $\mathcal{H}$ be a hereditary class of graphs. If $\mathcal{H}$ is not matching splittable then at least one of the following are true: (1.) $\mathcal{F}_\omega \subseteq \mathcal{H}$, (2.)  $\mathcal{F}_\beta \subseteq \mathcal{H}$, (3.) $\mathcal{F}_{P_2} \subseteq \mathcal{H}$, or (4.) $\mathcal{F}_{K_3} \subseteq \mathcal{H}$.
\end{theorem}

Thus, it suffices to consider cases 1.\ - 4.\ to prove Theorem~\ref{thm:main_hereditary}. We start with the easy cases of cliques and bicliques; they follow implicitly from previous works~\cite{CurticapeanDH21,DorflerRSW22,PeyerimhoffR0SV21MFCS} and we only include a proof for completeness. Note that a tight bound under rETH is known for those cases:
\begin{lemma}\label{lem:cliques_and_bicliques}
Let $\mathcal{H}$ be a hereditary class of graphs. If $\mathcal{F}_\omega \subseteq \mathcal{H}$ or $\mathcal{F}_\beta \subseteq \mathcal{H}$ then $\oplus\subsprob(\mathcal{H})$ is $\oplus\W{1}$-hard and, assuming rETH, cannot be solved in time $f(|H|)\cdot |G|^{o(|V(H)|)}$ for any function $f$.
\end{lemma}
\begin{proof}
If $\mathcal{F}_\omega \subseteq \mathcal{H}$ then $\oplus\W{1}$-hardness follows immediately from the fact that $\oplus\textsc{Clique}$ is the canonical $\oplus\W{1}$-complete problem~\cite{CurticapeanDH21}. For the rETH lower bound, we can reduce from the problem of \emph{deciding} the existence of a $k$-clique via a (randomised) reduction using a version of the Isolation Lemma due to Williams et al.\ \cite[Lemma~2.1]{WilliamsWWY15}. This reduction does not increase $k$ or the size of the host graph and is thus tight with respect to the well-known lower bound for the clique problem due to Chen et al.\ \cite{Chenetal05,Chenetal06}: Deciding the existence of a $k$-clique in an $n$-vertex graph cannot be done in time $f(k)\cdot n^{o(k)}$ for any function $f$, unless ETH fails. Our lower bound under rETH follows since the reduction is randomised.

If $\mathcal{F}_\beta \subseteq \mathcal{H}$, then the claim holds by~\cite[Theorem 5]{DorflerRSW22}, which established the problem of counting, modulo $2$, the induced copies of a $k$-by-$k$-biclique in an $n$-vertex bipartite graph to be $\oplus\W{1}$-hard and not solvable in time $f(k)\cdot n^{o(k)}$ for any function $f$, unless rETH fails. Since a copy of a biclique (with at least one edge) in a bipartite graph must always be induced, the claim follows. This concludes the proof of Lemma~\ref{lem:cliques_and_bicliques}.
\end{proof}

The more interesting cases are $\mathcal{F}_{P_2} \subseteq \mathcal{H}$ and $\mathcal{F}_{K_3} \subseteq \mathcal{H}$. One reason for this is that, in contrast to cliques and bicliques, the decision version of those instances are fixed-parameter tractable. Hence a reduction from the decision version via e.g.\ an isolation lemma does not help. In other words, establishing hardness for those cases requires us to rely on the full power of counting modulo $2$. More precisely, we will rely on the framework of fractures graphs (see Section~\ref{sec:prelims}). Both cases can be considered simpler applications of the machinery used in the later sections, so we will present all steps in great detail. While this might seem unnecessary given the simplicity of the constructions, we hope that it enables the reader to make themselves familiar with the general reduction strategies which will be used throughout the later sections of this work.

\subsection{Triangle Packings}\label{sec:triangle}
The goal of this subsection is to establish hardness of $\oplus\subsprob(\mathcal{F}_{K_3})$. To this end, let $\Delta$ be an infinite computable class of cubic bipartite expander graphs, and let $\mathcal{Q}=\{L(H)~|~H\in\Delta\}$ where $L(H)$ is constructed as follows: Each $v\in V(H)$ becomes a triangle with vertices $v_x$, $v_y$, and $v_z$ corresponding to the three neighbours $x$, $y$, and $z$ of $v$. Finally, for every edge $\{u,v\}\in E(H)$ we identify $v_u$ and $u_v$. In fact, $L(H)$ is just the \emph{line graph} of $H$: Every edge of $H$ becomes a vertex in $L(H)$, and two vertices of $L(H)$ are made adjacent if and only if the corresponding edges in $H$ are incident. Since all $H\in\Delta$ are bipartite (and thus triangle-free), we can easily observe the following.\footnote{Observation~\ref{obs:triangle_bijection} is also an immediate consequence of Whitney's Isomorphism Theorem implying that a triangle of a line graph corresponds to either a claw or to a triangle in its primal graph.}
\begin{observation}\label{obs:triangle_bijection}
The mapping $v \mapsto (v_x,v_y,v_z)$ is a bijection from vertices of $H$ to triangles in $L(H)$.
\end{observation}

We also consider the fracture of $L(H)$ that splits $L(H)$ back into $|V(H)|$ triangles; consider Figure~\ref{fig:trianglepackings} for an illustration.
\begin{definition}[$\tau(H)$]
Let $H\in \Delta$ and recall that each vertex $w$ of $L(H)$ is obtained by identifying $v_u$ and $u_v$ for some edge $\{u,v\}\in E(H)$. Moreover, $w$ has four incident edges $e_x$, $e_y$, $e_a$, $e_b$, to $v_x$, $v_y$, $u_a$, $u_b$, respectively, where $x,y,u$ are the neighbours of $v$ in $H$ and $v,a,b$ are the neighbours of $u$ in $H$. We define $\tau(H)_w := \{ \{e_x,e_y\} , \{e_a,e_b \} \}$, and we proceed similar for all vertices of $L(H)$.
\end{definition}

Next, we use that $\mathsf{tw}(L(H)) = \Omega(\mathsf{tw}(H))$ (see e.g.\ \cite{HarveyW18}). Moreover, $\mathsf{tw}(L(H))\leq |V(L(H))|$ since the treewidth of a graph is always bounded by the number of its vertices. Additionally, $|V(L(H))|=|E(H)|$ by construction. Since the graphs in $\Delta$ are cubic, we further have that $|E(H)|= \Theta(|V(H)|)$ for $H\in \Delta$. We combine those bounds with the fact that expander graphs have treewidth linear in the number of vertices (see e.g.\ \cite{GroheM09}); therefore $\Delta$ and thus $\mathcal{Q}$ have unbounded treewidth. Putting these facts together, we obtain the following.

\begin{fact}\label{fact:linear_tw}
$\mathcal{Q}$ has unbounded treewidth and $\mathsf{tw}(L(H))= \Theta(|V(L(H))|) = \Theta(|V(H)|)$ for $H\in \Delta$.
\end{fact}

We are now able to establish hardness of $\oplus\subsprob(\mathcal{F}_{K_3})$. The proof will heavily rely on the transformation from edge-coloured subgraphs to homomorphisms established in~\cite{PeyerimhoffRSSVW22}.

\begin{figure}
    \centering
    \includegraphics[width=\textwidth]{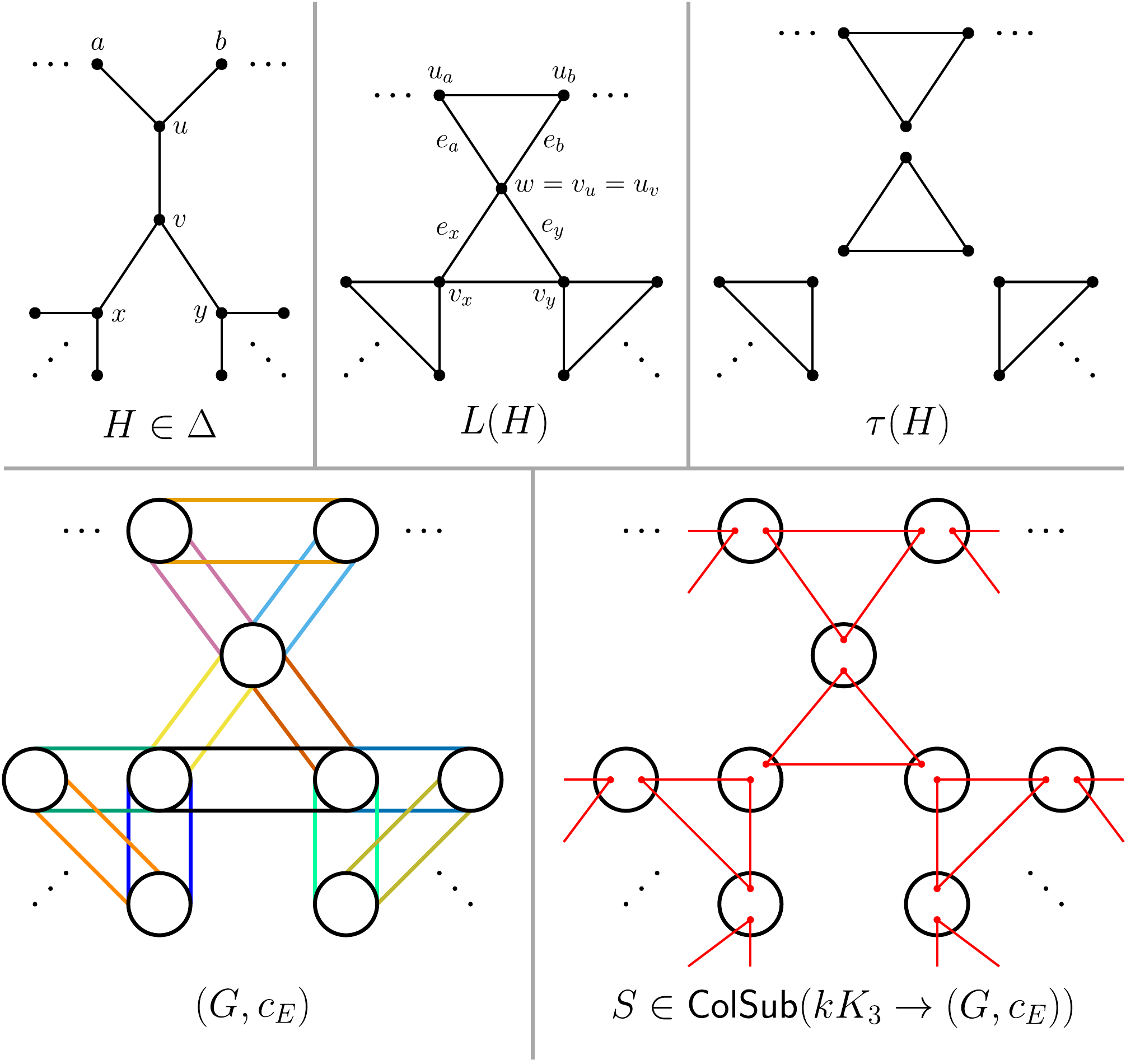}
    \caption{(\emph{Top}:) A cubic bipartite graph $H\in \Delta$, its line graph $L(H)$, and the fractured graph induced by $\tau(H)$. (\emph{Below}:) An $L(H)$-coloured graph $(G,c)$; emphasised in distinct colours is the edge-colouring $c_E$ of $G$ induced by the mapping $\{u,v\}\mapsto \{c(u),c(v)\}$. Additionally we depict an element $S\in \colsubs{kK_3}{(G,c_E)}$, that is, a subgraph of $G$ isomorphic to $kK_3$ that contains each edge colour of $G$ precisely once.} 
    \label{fig:trianglepackings}
\end{figure}

\begin{lemma}\label{lem:trianglepackings}
The problem $\oplus\subsprob(\mathcal{F}_{K_3})$ is $\oplus\W{1}$-hard. Furthermore, on input $kK_3$ and $G$, the problem cannot be solved in time $f(k)\cdot |G|^{o(k/\log k)}$ for any function $f$, unless rETH fails.
\end{lemma}

\begin{proof}
We reduce from $\oplus\cphomsprob(\mathcal{Q})$, which, by Fact~\ref{fact:linear_tw} and Theorem~\ref{thm:cphom_lower_bound}, is $\oplus\W{1}$-hard and 
for $L(H)\in \mathcal{Q}$, it
cannot be solved in time $f(|L(H)|)\cdot |G|^{o(|V(L(H))|/\log |V(L(H))|)}$, unless rETH fails.

Let $L$ and $(G,c)$ be an input instance to $\oplus\cphomsprob(\mathcal{Q})$. 
Recall that $\Delta$ is computable --- that is, there is an algorithm that takes a graph~$H$ and determines whether it is in~$\Delta$. Thus, there is an algorithm 
that takes input $L\in \mathcal{Q}$ and finds
a graph $H\in \Delta$ with $L=L(H)$.
The run time of this algorithm depends on $|L|$ but clearly not on $(G,c)$.  Let $k = |V(H)|$ and note that $|E(L(H))|=3k$, since, by construction, each vertex $v$ of $H$ becomes a triangle of $L(H)$. We consider the graph $G$ as a $3k$-edge-coloured graph, coloured by $c_E$. That is, each edge $e=\{x,y\}$ of~$G$ is assigned the colour $c_E(e) = \{c(x),c(y)\} $ which is an edge of~$L$ (see Figure~\ref{fig:trianglepackings} for an illustration).

Now, for \emph{any} $L$-coloured graph $(G',c')$ recall that $\colsubs{kK_3}{(G',c'_E)}$ is the set of subgraphs of $G'$ that are isomorphic to $kK_3$ and that include each edge colour (each edge of $L$) precisely once. 
We will see later that $\oplus\colsubs{kK_3}{(G',c'_E)}$ can be computed using our oracle for $\oplus\subsprob(\mathcal{F}_{K_3})$ using the principle of inclusion and exclusion.

It was shown in~\cite[Lemma~4.1]{PeyerimhoffRSSVW22} that there is a unique function $a$ such that for every $L$-coloured graph $(G',c')$ we have\footnote{In the language of~\cite{PeyerimhoffRSSVW22}, Equation~\eqref{eq:hombasis_triangles} is obtained by choosing $\Phi$ as the property of being isomorphic to~$kK_3$.}
\begin{equation}\label{eq:hombasis_triangles}
    \#\colsubs{kK_3}{(G',c'_E)} = \sum_{\rho} a(\rho) \cdot \homs{\fracture{L}{\rho}}{(G',c')}\,.
\end{equation}
where the sum is over all fractures of $L$. Additionally, it was shown in~\cite[Corollary~4.3]{PeyerimhoffRSSVW22} that 
\begin{equation}\label{eq:topcoeff_triangles}
    a(\top) = \sum_{\rho \in \mathsf{F}(kK_3,L)} \prod_{w\in V(L)} (-1)^{|\rho_w|-1}\cdot (|\rho_w|-1)! \,,
\end{equation}
where $\top$ is the fracture in which each partition consists only of one block (that is, $\fracture{L}{\top}=L$), and $\mathsf{F}(kK_3,L)$ is the set of all fractures $\rho$ of $L$ such that $\fracture{L}{\rho}\cong kK_3$. However, note that, by Observation~\ref{obs:triangle_bijection}, there is only way to fracture $L$ into $k$ disjoint triangles, and this fracture is given by $\tau(H)$. Thus, \eqref{eq:topcoeff_triangles} simplifies to
\begin{equation}\label{eq:topcoeff_triangles_simpl}
    a(\top) = \prod_{w\in V(L)} (-1)^{|\tau(H)_w|-1}\cdot (|\tau(H)_w|-1)! \,,
\end{equation}
which is odd since each partition of $\tau(H)$ consists of precisely two blocks (so in fact the expression in \eqref{eq:topcoeff_triangles_simpl} is $(-1)^{|V(L)|}$).

Note that the algorithm for
$\oplus\cphomsprob(\mathcal{Q})$
is supposed to compute $\oplus\homs{(L,\mathsf{id}_L)}{(G,c)}$
which is equal to  
$\oplus\homs{\fracture{L}{\top}}{(G,c_{\top})}$. Since $a(\top)$ is odd, we can invoke Lemma~\ref{lem:complexity_monotonicty} to recover this term by evaluating the entire linear combination~\eqref{eq:hombasis_triangles}, that is, by evaluating the function $\oplus\colsubs{kK_3}{\star}$. More concretely, this means that we need to compute $\oplus\colsubs{kK_3}{(G',c'_E)}$ for some $L$-coloured graphs $(G',c')$ of size at most $f(|L|)\cdot |G|$ for some computable function $f$ (see 3.\ in Lemma~\ref{lem:complexity_monotonicty}). This can easily be done using Lemma~\ref{lem:remove_edgecols} since we have oracle access to the function $\oplus\subs{kK_3}{\star}$. We emphasise that, by condition 2.\ of Lemma~\ref{lem:remove_edgecols}, each oracle query $\hat{G}$ satisfies $|\hat{G}| \leq |G'|$, where $(G',c')$ is the $L$-coloured graph for which we wish to compute $\oplus\colsubs{kK_3}{(G',c'_E)}$. Since $|(G',c')|\leq f(|L|)\cdot |G|$, we obtain that $|\hat{G}| \leq f(|L|)\cdot |G|$ as well.

Since, by Fact~\ref{fact:linear_tw},  
$k= \Theta(|kK_3|)=\Theta(|V(L)|)=\Theta(\mathsf{tw}(L))$,
our reduction yields $\oplus\W{1}$-hardness and transfers the conditional lower bound under rETH as desired.
\end{proof}

\subsection{$P_2$-packings}\label{sec:p2_pack}
Next we establish hardness for the case of $P_2$-packings. The strategy will be similar in spirit to the construction for triangle packings; however, rather then identifying a unique fracture for which the technique applies, we will encounter an \emph{odd} number of possible fractures in the current section.

Let $\Delta$ be a computable infinite class of $4$-regular expander graphs, and let $\mathcal{Q}$ be the class of all subdivisions of graphs in $\Delta$, that is $\mathcal{Q}=\{H^2~|~H \in \Delta \}$, where $H^2$ is obtained from $H$ by subdividing each edge once. 

We start by establishing an easy but convenient fact on the treewidth of the graphs in $\mathcal{Q}$.
\begin{lemma}\label{lem:Q_tw_pathpackings}
    $\mathcal{Q}$ has unbounded treewidth and $\mathsf{tw}(H^2)=\Theta(|V(H)|)$ for $H\in \Delta$.
\end{lemma}
\begin{proof}
As in  Section~\ref{sec:triangle}, $\mathsf{tw}(H) = \Theta(|V(H)|)$ for $H\in\Delta$, since expanders have treewidth linear in the number of vertices. Since $H$ is a minor of $H^2$, and since taking minors cannot increase treewidth (see~\cite[Exercise 7.7]{CyganFKLMPPS15}), we thus have that $\mathsf{tw}(H^2) = \Omega(|V(H)|))$. Finally, we have $\mathsf{tw}(H^2) \leq |V(H^2)|$ since the treewidth is at most the number of vertices, and $|V(H^2)|= O(|V(H)|)$ since $H$ is $4$-regular. In combination, we obtain $\mathsf{tw}(H^2)=\Theta(|V(H)|)$ for $H\in \Delta$. Note that this also implies that $\mathcal{Q}$ has unbounded treewidth (as $\Delta$ is infinite).
\end{proof}

For what follows, given a subdivision $H^2$ of a graph $H$, it will be convenient to assume that $V(H^2)=V(H)\cup S_E$, where $S_E=\{s_e~|~e\in E(H\})$ is the set of the subdivision vertices.

\begin{definition}[Odd Fractures]\label{def:odd_fractures}
    Let $H\in \Delta$ and let $\tau$ be a fracture of $H^2$. We say that $\tau$ is \emph{odd} if the following two conditions are satisfied:
    \begin{enumerate}
        \item For each $s \in S_E$ the partition $\tau_s$ consists of two singleton blocks.
        \item For each $v \in V(H)$ the partition $\tau_v$ consists of two blocks of size $2$.
    \end{enumerate}
    Consider Figure~\ref{fig:p2packings} for a depiction of an odd fracture.
\end{definition}

The following two lemmas are crucial for our construction.
 
\begin{figure}
    \centering
    \includegraphics[width=\textwidth]{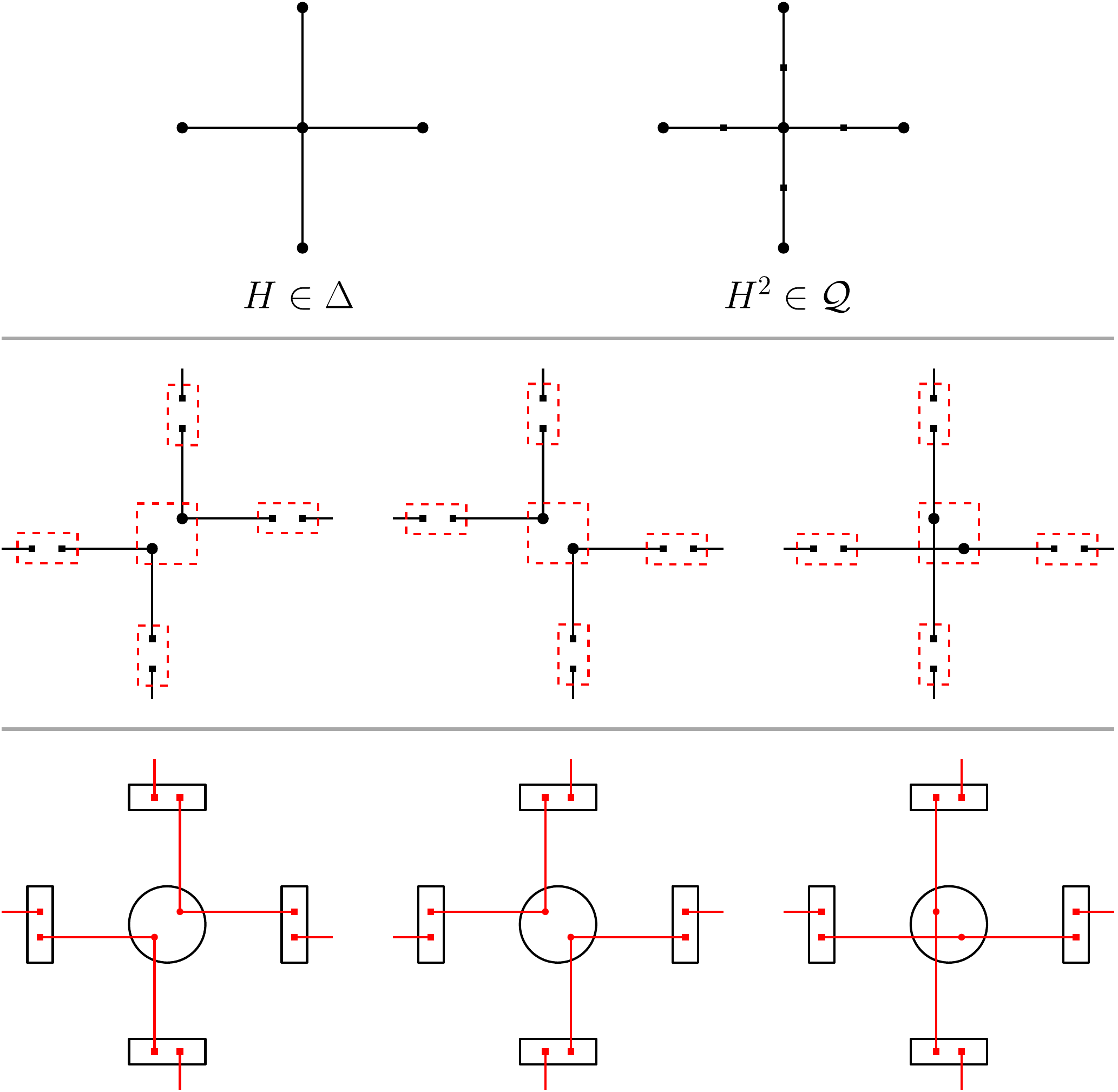}
    \caption{(\emph{Top}:) Subdividing a $4$-regular expander in $\Delta$ depicted by the neighbourhood of an individual vertex. (\emph{Centre}:) Illustrations of odd fractures (Definition~\ref{def:odd_fractures}). For each non-subdivision vertex, there are only three ways to satisfy 2.\ in Definition~\ref{def:odd_fractures}. This observation is used in Lemma~\ref{lem:odd_fractures} to show that the number of odd fractures is a power of~$3$. (\emph{Bottom}:) Elements of $\colsubs{kP_2}{(G,c_E)}$ inducing fractures of $H^2$ such that each partition has at most two blocks. Lemma~\ref{lem:isolate_P2} shows that those are precisely the odd fractures of~$H^2$.} 
    \label{fig:p2packings}
\end{figure}

\begin{lemma}\label{lem:odd_fractures}
    Let $H\in \Delta$. The number of odd fractures of $H^2$ is odd.
\end{lemma}
\begin{proof}
    The first condition in Definition~\ref{def:odd_fractures} leaves only one choice for subdivision vertices. Let us thus consider a vertex $v\in V(H)=V(H^2)\setminus S_E$. Since $H$ is $4$-regular, there are $4$ incident edges to $v$. Now note that there are precisely $3$ partitions of a $4$-element set with two blocks of size $2$. Thus the total number of odd fractures of $H^2$ is $3^{|V(H)|}$, which is odd.
\end{proof}

\begin{lemma}\label{lem:isolate_P2}
    Let $H \in \Delta$, let $k=2|V(H)|$ and let $\tau$ be a fracture of $H^2$ such that $\tau_v$ consists of at most $2$ blocks for each $v\in V(H^2)$. Then $\fracture{H^2}{\tau} \cong kP_2$ if and only if $\tau$ is odd.
\end{lemma}
\begin{proof}
    First observe that $|E(H^2)|=2|E(H)|=4|V(H)|=2k$. Thus the number of edges of $\fracture{H^2}{\tau}$ is equal to $2k$ (for each fracture $\tau$ of $H^2$), which is also equal to the number of edges of $kP_2$.

    Thus, $\fracture{H^2}{\tau}$ is isomorphic to $kP_2$ if and only if each connected component of $\fracture{H^2}{\tau}$ is a path of length $2$. It follows immediately by Definition~\ref{def:odd_fractures} that $\tau$ being odd implies that $\fracture{H^2}{\tau}$ consists only of disjoint $P_2$. It thus remains to show the other direction.

    Assume for contradiction that there is a subdivision vertex $s \in S_E$ of $H^2$ such that $\tau_s$ consists of only one block (recall that $s$ has degree $2$, thus $\tau_s$ either consists of two singleton blocks, or of one block of size $2$). Let $e=\{u,v\}\in E(H)$ be the edge corresponding to $s$, that is, $s$ was created by subdividing $e$. Since $\fracture{H^2}{\tau}$ is a union of $P_2$, we can infer that $\tau_v$ and $\tau_u$ contain a singleton block (otherwise we would have created a connected component which is not isomorphic to $P_2$). Now recall that both $u$ and $v$ have degree $4$, since $H$ is $4$-regular. We obtain a contradiction as follows: By assumption of the lemma, we know that $\tau_v$ and $\tau_u$ can have at most two blocks. Since we have just shown that both contain a singleton block, it follows that both $\tau_v$ and $\tau_u$ contain one further block of size $3$. However, a block of size $3$ yields a vertex of degree $3$ in the fractured graph $\fracture{H^2}{\tau}$, contradicting the fact that $\fracture{H^2}{\tau}$ consists only of disjoint $P_2$.

    Thus we have established that, for each $s\in S_E$, the partition $\tau_s$ consists of two singleton blocks. Given this fact, the only way for $\fracture{H^2}{\tau}$ being a disjoint union of $P_2$ is that each partition $\tau_v$, for $v\in V(H)=V(H^2)\setminus S_E$, consists of two blocks of size $2$. 
\end{proof}

We are now able to prove our hardness result.
\begin{lemma}\label{lem:P2_hardness}
The problem $\oplus\subsprob(\mathcal{F}_{P_2})$ is $\oplus\W{1}$-hard. Furthermore, on input $kP_2$ and $G$, the problem cannot be solved in time $f(k)\cdot |G|^{o(k/\log k)}$ for any function $f$, unless rETH fails.
\end{lemma}

\begin{proof}
We reduce from $\oplus\cphomsprob(\mathcal{Q})$, which, by Lemma~\ref{lem:Q_tw_pathpackings} and Theorem~\ref{thm:cphom_lower_bound}, is $\oplus\W{1}$-hard and 
for $H'\in \mathcal{Q}$, it
cannot be solved in time $f(|H'|)\cdot |G|^{o(|V(H')|/\log |V(H')|)}$, unless rETH fails. 

Let $H'$ and $(G,c)$ be an input instance to $\oplus\cphomsprob(\mathcal{Q})$. 
There is an algorithm that takes as input a graph $H' \in \mathcal{Q}$ and finds a graph $H \in \Delta$ with $H' = H^2$ --- this is basically 2-colouring.  
The run time of this algorithm depends on $|H'|$ but clearly not on $(G,c)$.  Let $k = 2|V(H)|$ and note that $|E(H^2)|=2|E(H)|=4|V(H)|=2k$. We consider the graph $G$ as a $2k$-edge-coloured graph, coloured by $c_E$. That is, each edge $e=\{x,y\}$ of~$G$ is assigned the colour $c_E(e) = \{c(x),c(y)\} $ which is an edge of~$H'=H^2$.

Now, for \emph{any} $H^2$-coloured graph $(G',c')$ recall that $\colsubs{kP_2}{(G',c'_E)}$ is the set of subgraphs of $G'$ that are isomorphic to $kP_2$ and that include each edge colour (each edge of $H^2$) precisely once. 
We will see later that $\oplus\colsubs{kP_2}{(G',c'_E)}$ can be computed using our oracle for $\oplus\subsprob(\mathcal{F}_{P_2})$ using the principle of inclusion and exclusion.

It was shown in~\cite[Lemma~4.1]{PeyerimhoffRSSVW22} that there is a unique function $a$ such that, for every $H^2$-coloured graph $(G',c')$, 
\begin{equation}\label{eq:hombasis_P2}
    \#\colsubs{kP_2}{(G',c'_E)} = \sum_{\rho} a(\rho) \cdot \homs{\fracture{H^2}{\rho}}{(G',c')}\,.
\end{equation}
where the sum is over all fractures of $H^2$. 
As in Section~\ref{sec:triangle} from~\cite[Corollary~4.3]{PeyerimhoffRSSVW22} we know that 
\begin{equation}\label{eq:topcoeff_P2}
    a(\top) = \sum_{\rho \in \mathsf{F}(kP_2,H^2)} \prod_{w\in V(H^2)} (-1)^{|\rho_w|-1}\cdot (|\rho_w|-1)! \,,
\end{equation}
where $\top$ is the fracture in which each partition consists only of one block 
and $\mathsf{F}(kP_2,H^2)$ is the set of all fractures $\rho$ of $H^2$ such that $\fracture{H^2}{\rho}\cong kP_2$.

Our next goal is to show that $a(\top)=1 \mod 2$. First, suppose that a fracture $\rho$ contains a partition $\rho_w$ with at least three blocks. Then $(|\rho_w|-1)!=0 \mod 2$. Thus such fractures do not contribute to $a(\top)$ if arithmetic is done modulo $2$. Next, 
note that if, for each $w$, the partition $\rho_w$ contains at most $2$ blocks, then
 
\[\prod_{w\in V(H^2)} (-1)^{|\rho_w|-1}\cdot (|\rho_w|-1)! = 1 \mod 2.\]
Let  $\mathsf{Odd}(kP_2,H^2)$ be the set of all fractures~$\rho$ of $H^2$ such that  $\fracture{H^2}{\rho}\cong kP_2$ and each partition of $\rho$ consists of at most $2$ blocks. Our analysis then yields $a(\top)= |\mathsf{Odd}(kP_2,H^2)| \mod 2$. 
Finally, Lemma~\ref{lem:isolate_P2} states that $\mathsf{Odd}(kP_2,H^2)$ is precisely the set of odd fractures, and Lemma~\ref{lem:odd_fractures} thus implies that $|\mathsf{Odd}(kP_2,H^2)| =1 \mod 2$. Consequently, $a(\top)=1 \mod 2$ as well, and we have achieved the goal.

Next we can proceed similarly to  the case of triangle packings. 
As in that case, the goal is to compute
$\oplus\homs{(H^2,\mathsf{id}_{H^2})}{(G,c)}) $
which is equal to $\oplus\homs{(\fracture{H^2}{\top},c_\top)}{(G,c)}$.
Since $a(\top)$ is odd, we can invoke Lemma~\ref{lem:complexity_monotonicty}
to recover this term by evaluating the   entire linear combination~\eqref{eq:hombasis_P2}, that is, if we can evaluate the function $\oplus\colsubs{kP_2}{\star}$. 
This can be done by using Lemma~\ref{lem:remove_edgecols}.
Each call to the oracle  
is of the form $\oplus\subs{kP_2}{\hat{G}}$ where $|\hat{G}|$ is bounded by $f(k)\cdot |G|$. 

Now recall that $k\in \Theta(|V(H)|)$. By Lemma~\ref{lem:Q_tw_pathpackings}, we thus have $k= \Theta(\mathsf{tw}(H^2))$. Hence our reduction yields $\oplus\W{1}$-hardness and transfers the conditional lower bound under rETH as desired.
\end{proof}

We can now conclude the treatment of hereditary pattern classes by proving Theorem~\ref{thm:main_hereditary}, which we restate for convenience.

\mainhereditary*
\begin{proof}
The fixed-parameter tractability result was shown in~\cite{CurticapeanDH21}.
For the hardness result, using the fact that $\mathcal{H}$ is not matching splittable and Theorem~\ref{thm:matchsplit_characterisation} we obtain four cases.
    \begin{itemize}
        \item If $\mathcal{H}$ contains all cliques or all bicliques, then hardness follows from Lemma~\ref{lem:cliques_and_bicliques}.
        \item If $\mathcal{H}$ contains all triangle packings, then hardness follows from Lemma~\ref{lem:trianglepackings}.
        \item If $\mathcal{H}$ contains all $P_2$-packings, then hardness follows from Lemma~\ref{lem:P2_hardness}.
    \end{itemize}
    Since the case distinction is exhaustive, the proof is concluded.
\end{proof}

\section{Classification for Trees}\label{sec:trees}
Our overall goal is to prove Theorem~\ref{thm:main}, which we restate for convenience:

\main*

\paragraph*{Outline of Section~\ref{sec:trees}}
We begin our analysis by investigating the structural properties of classes of trees that are not matching splittable. In Lemma~\ref{lem:non-matchsplit-trees} we prove that for each such class $\calT$ (at least) one of the following parameters are unbounded: The \emph{fork number} (Definition~\ref{def:fork}), the \emph{star number} (Definition~\ref{def:star}), or the $\hal$-\emph{number} (Definition~\ref{def:hal_number}).

The remainder of this section is then split into three, largely independent, parts: Section~\ref{sec:caseH} establishes hardness of $\oplus\subsprob(\calT)$ for classes of trees $\calT$ of unbounded $\hal$-number, Section~\ref{sec:caseStar} shows hardness for unbounded star number, and Section~\ref{sec:caseF} shows hardness for unbounded fork number.

We start by introducing some terminology for trees which will be used in the remainder of this section.
\begin{definition}[$2$-paths]\label{def:2path}
A $2$-\emph{path} of length $a$ of a tree $T$ is a path $x_0,x_1,\dots,x_a$ such that $\deg(x_0)\neq 2$, $\deg(x_1)=\dots=\deg(x_{a-1})=2$ and $\deg(x_a)\neq 2$.
\end{definition}

Next we introduce rays, which are restricted $2$-paths that will be crucial in our analysis.
\begin{definition}[source, ray, $\degla{a}$, $\degl$, and $\degnl$]\label{def:ray}
Let $T$ be a tree. 
A \emph{source} of~$T$ is any vertex with degree greater than~$2$.
A \emph{ray} of length $a$ of $T$ is a $2$-path $x_0,x_1,\dots,x_a$ such that $\deg(x_0)>2$ and $\deg(x_a)=1$. We call $x_0$ the \emph{source of the ray}.
Given a vertex $s$ of degree at least $3$, we write $\degla{a}(s)$ for the number of rays of length $a$ with source $s$. We set
\[\degl(s):=\sum_a \degla{a}(s) \,.\]
Finally, we set $\degnl(s):=\deg(s)-\degl(s)$.  
\end{definition}

Next, we introduce parameters $\fork_{a,b}$, $\starnum_c$ and $\hal_d$. Our goal is then to show that, for every non-matching-splittable class of trees, at least one of those two parameters is unbounded. 

\begin{definition}[Forks and $\fork_{a,b}$]\label{def:fork}
Let $a,b$ be positive integers. A source $s$ of a tree $T$ is called an $a$-$b$-\emph{fork} if $\degnl(s)=1$ and one of the following is true
\begin{itemize}
    \item $a\neq b$ and $\degla{a}(s),\degla{b}(s)>0$.
    \item $a=b$ and $\degla{a}(s)>1$.
\end{itemize}
The $a$-$b$-\emph{fork number} of $T$, denoted by $\fork_{a,b}(T)$ is the maximum size of an independent set containing only $a$-$b$-forks. Finally, we say that a class of trees $\calT$ has \emph{unbounded fork number} if for every positive integer $B$ there 
are positive integers~$a$ and~$b$ and  a tree $T\in\calT$ such that $\fork_{a,b}(T)\geq B$.
\end{definition}

\begin{definition}[Stars and $\starnum_c$]\label{def:star}
A \emph{star} of size $k>1$ in a tree $T$ is a collection of $k$ distinct rays that have a common source $s$.
For a positive integer $c\geq 3$, a $c$-\emph{star} of size $k$ in a tree $T$ is a collection of $k$ distinct rays of length $c$ that have a common source $s$.

The $c$-\emph{star number} of a tree $T$, denoted by $\starnum_c(T)$ is the maximum size of a $c$-star in $T$. Finally, we say that a class of trees $\calT$ has \emph{unbounded star number} if for every positive integer $B$ there exists $c\geq 3$, and a tree $T\in\calT$ such that $\starnum_c(T)\geq B$.
\end{definition}

\begin{definition}[$\hal$-gadgets and $\hal_d$]\label{def:hal_number}
A $\hal$-gadget\footnote{$\hal$ stands for \emph{caterpillar}, the shape of which resembles the structure of a $\hal$-gadget.} of \emph{order} $d$ and \emph{length} $k$ in a tree $T$ is a path $x_0,\dots,x_k$ such that one of the following is true for each inner vertex $x_i\in\{1,\dots,k-1\}$:
\begin{itemize}
    \item[(i)] $\deg(x_i)=2$, that is $N(x_i)=\{x_{i-1},x_{i+1}\}$, or
    \item[(ii)] $x_i$ is a source and every neighbour $v\in N(x_i)\setminus\{x_{i-1},x_{i+1}\}$ is contained in a ray of length at most $d$ from $x_i$ to a leaf.
\end{itemize}
The $\hal_d$-number of a tree $T$, denoted by $\hal_d(T)$ is the length of the longest $\hal$-gadget of order $d$. Finally, we say that a class of trees $\calT$ has \emph{unbounded $\hal$-number} if there exists $d>0$ such that for every positive integer $B$, and a tree $T\in\calT$ such that $\hal_d(T)\geq B$.
\end{definition}
Note that the ordering of the quantifiers in the definition of the $\hal_d$-number is different from the ordering in the definition of the $c$-star-number. This is due to technical reasons which are important for the proof of Lemma~\ref{lem:non-matchsplit-trees}.

\begin{lemma}\label{lem:non-matchsplit-trees}\label{lem:ten}
Let $\calT$ be a class of trees. If $\calT$ is not matching splittable, then $\calT$ has either unbounded fork number, unbounded star number, or unbounded $\hal$-number.
\end{lemma}
\begin{proof}
We can assume that there is an overall bound $d$ on the length of $2$-paths in trees in $\calT$: Otherwise, $\calT$ already has unbounded $\hal$-number (see (i) in Definition~\ref{def:hal_number})). Hence the length of every ray in any tree in $\calT$ is bounded by $d$ as well. Thus
\begin{itemize}
    \item   $\calT$ has unbounded fork number if and only if for every positive integer $B$ there are $a,b\in\{1,\dots,d\}$ and a tree $T\in\calT$ such that $\fork_{a,b}(T)\geq B$.
    \item $\calT$ has unbounded $\hal$-number if and only if $\hal_d$ is unbounded in $\calT$ (see Definition~\ref{def:hal_number})).  
 \item  $\calT$ has unbounded star number if and only if for every positive integer $s$ there is a $c\in\{3,\dots,d\}$ and a tree $T\in\calT$ such that $\starnum_c(T)\geq s$.

\end{itemize}
We split the proof into two cases.

\noindent {\bf Case 1.  $\calT$ has unbounded diameter. }

In Case~1, we show that $\calT$ has unbounded fork number or unbounded $\hal$-number.
If $\hal_d$ is unbounded in $\calT$ then $\calT$ has unbounded $\hal$-number and we are done so assume that there is a constant $h$ such that $\hal_d(T)\leq h$ for every $T\in\calT$.

Now let $B$ be a positive integer. We show that there are $a,b\in\{1,\dots,d\}$ and $T\in\calT$ with $\fork_{a,b}(T)\geq B$. To this end, we use the premise that $\calT$ has unbounded diameter. Let $k>(h+2)(B d^2+1)$ be a positive integer, and let $T\in\calT$ be such that there is a path $P=s,p_0,\dots,p_k,t$ in $T$. Observe that the deletion of all edges in $P$ decomposes $T$ into a family of
disjoint
subtrees. We write $T_i$ for the subtree that contains $p_i$. 
Now decompose $P$ into segments $P_1,P_2,\dots$ of length $h+2$. Note that a segment $P_j=p_{j_0},\dots,p_{j_{h+2}}$ yields a $\hal$-gadget of order $d$ and length $> h$ if and only if $T_{j_i}$ is either a star or an isolated vertex for each $i\in\{1,\dots,h+1\}$.

Since no such $\hal$-gadgets exist by assumption, we obtain that each segment $P_j$ of the path $P$ contains a vertex $p_{i_j}$ such that $T_{i_j}$ is neither a star nor an isolated vertex. 

Assume that $T_{i_j}$ is rooted at $p_{i_j}$. Since $T_{i_j}$ is neither an isolated vertex nor a star, there must be a (proper) descendant $v_{i_j}$ of $p_{i_j}$ (in $T_{i_j}$) such that $v_{i_j}$ is an $(a_{i_j},b_{i_j})$-fork for some $a_{i_j},b_{i_j}\in\{1,\dots,d\}$. Now note that there are at most $d^2$ pairs of integers in $\{1,\dots,d\}$. Since we have at least one fork for every segment and since there are at least $\lfloor k/(h+2) \rfloor > Bd^2 +1$  segments, we thus obtain by the pigeon-hole principle that there is a pair $a,b\in\{1,\dots,d\}$ such that, for at least $B$ segments $P_{i_j}$, the node $v_{i_j}$ is an $(a,b)$-fork in $T_{i_j}$ and thus also in $T$. Since those forks are pairwise non-adjacent, we obtain, as desired, that the $(a,b)$-fork number of $T$ is at least $B$, concluding Case~1.  

\noindent {\bf Case 2.   $\calT$ has bounded diameter.}

 Let $D$ be the assumed upper bound on the diameter of trees in $\calT$. 
  If $\calT$ has unbounded star number then we are finished. Assume instead that $\calT$ has
bounded star number. 
Then   there is a positive integer~$s$ such that for all $c\in\{3,\ldots,d\}$ and every tree $T\in \calT$, 
$\starnum_c(T) < s$.
We will show that $\calT$ has unbounded fork number.
Consider any positive integer~$B$.
We will show that there are $a,b\in\{1,\dots,d\}$ and $T\in\calT$ with $\fork_{a,b}(T)\geq B$.

Let $k> (D+1)(Bd^2+1)  (d^2s+1)$ be a positive integer. Since $\calT$ is not matching splittable, there is a tree $T\in\calT$ whose matching-split number is at least~$k$. Note that $T$ is not  a path
since every path with matching-split number at least~$k$
has length greater than $k>D$, contradicting the bound
on the diameter.

Now fix any vertex $r$ of $T$ as the root. Given a vertex $v$ of $T$, we write $T_v$ for the subtree rooted at $v$ (assuming that $r$ is the overall root). We call $v$ a \emph{rooted fork} if $T_v$ is a star --- observe that each rooted fork must indeed be a fork. Let $f$ be the number of rooted forks.   Similar to the argument in  Case~1, if $f>Bd^2+1$, then 
by the pigeon-hole principle
there are $a,b\in\{1,\dots,d\}$ such that $F_{a,b}(T)\geq B$.

Hence assume for contradiction that $f\leq Bd^2+1$.
Let $\mathcal{R}$ be the set of all rays of $T$ and recall that each ray in $\mathcal{R}$ is, by definition, a $2$-path of the form
$R = x_0,\dots,x_{d'} $
for $d'\leq d$, where $\deg(x_0) >2$ and $x_{d'}$ is a leaf.  We call a ray $R$ \emph{long} if $d'\geq 3$.
Note that the source of every ray must either be a rooted fork, or it must lie on a path from the root $r$ to one of the rooted forks. 

Let $T'$ be the subtree of $T$ induced by all vertices that lie on paths between $r$ and a rooted fork (including $r$ and all rooted forks). Since there are $f$ rooted forks and the depth of $T$ is bounded by $D$,    $|V(T')| \leq (D+1)f\leq (D+1)(Bd^2+1)$.

Consider a vertex $v$ of $T'$. Assume for contradiction that $v$ is the source of $>ds$ long rays (in $T$). Recall that for all $c\in \{3,\dots,d\}$ we have that $\starnum_c(T)<s$. Recall further that each long ray has length $d'$ for some $3\leq d'\leq d$. Thus we obtain a contradiction by the pigeon-hole principle. 

Now let $S$ be the set containing all vertices of $T'$ and all vertices of long rays. Noting that each long ray has length at most $d$, and that the source of each long ray must be a vertex of $T'$ by construction, we can use the observation that each vertex of $T'$ is the source of at most $ds$ long rays  to (generously) bound
\[|S|\leq |V(T')| + |V(T')|\cdot d\cdot ds\,.\]
Note further that $T[V(T)\setminus S]$ consists only of isolated edges and vertices: The only vertices in $V(T)\setminus S$ are non-source vertices of rays of length $<3$, the sources of which are in $T'$. Thus, $S$ is a splitting set. Finally, recalling that $|V(T')| \leq (D+1)f\leq (D+1)(Bd^2+1)$, we have
\[|S|\leq |V(T')| + |V(T')|\cdot d\cdot ds \leq (D+1)(Bd^2 +1)(d^2s+1) \,,\]
contradicting the fact that the matching-split number of $T$ is strictly larger than $(D+1)(Bd^2 +1)(d^2s+1)$. This concludes Case~2, and hence the proof.
\end{proof}

In the next three subsections, we will prove hardness of $\oplus\subsprob(\calT)$ for non-matching-splittable $\calT$ in each of the three cases given by  Lemma~\ref{lem:ten}.

\subsection{Unbounded $\hal$-number}\label{sec:caseH}
For our hardness proof,  it will be useful to find a proper sub-gadget of 
a 
$\hal$-gadget in a tree.
\begin{definition}[Strong $\hal$-gadgets, junctions, and closedness]
Let $C=x_0,\dots,x_L$ be a $\hal$-gadget of order $d$ and length $L$ in a tree $T$. We call $C$ a \emph{strong} $\hal$-gadget with $k$ \emph{junctions} if there are integers $0=i_0<i_1<\dots<i_k<i_{k+1}=L$ such that
\begin{itemize}
    \item[(I)] for all $j\in\{0,\dots,k\}$, $i_{j+1} - i_{j} >2d$, and 
    \item[(II)] for all $j\in\{1,\dots,k\}$, $x_{i_j}$ is the source of a ray $R_j$ of length
    $d$ that does not contain one of the neighbours $x_{i_j-1}$ and $x_{i_j+1}$ of $x_{i_j}$.
The vertices $x_{i_1},\ldots,x_{i_k}$ are called 
  the \emph{junctions}.
\end{itemize}
Finally, a strong $\hal$-gadget is called \emph{closed} if  neither $x_{i_1}$ nor $x_{i_k}$ are forks.\footnote{The condition of being closed rules out the special case in which $x_0$ 
    or $x_L$   are leaves of $T$.
More generally it rules out the case where there is a ray from $x_1$ including $x_0$ or from $x_{k}$ including $x_L$.    }

\end{definition}
Consider the bottom part of Figure~\ref{fig:HatGDef} for a visualisation.
We start with the following lemma which establishes the existence of a strong $\hal$-gadget with many junctions inside a long enough $\hal$-gadget.
\begin{lemma}\label{lem:theprevone}
Let $T$ be a tree such that the longest $2$-path in $T$ has length $d\geq 1$, and let $k$ be a positive integer. Then there exists $L>0$ (only depending on $k$ and $d$) such that the following is true: If $T$ contains an $\hal$-gadget of order $d$ and length $L$, then there exists $1\leq d'\leq d$ such that $T$ contains a strong $\hal$-gadget of order $d'$ with at least $k$ junctions.
\end{lemma}
\begin{proof}
Let $f(x)=x/(k+1) -2d-1$ and let $L$ be large enough such that $f^d(L) >d$. Let $H^d=x_0,\dots,x_L$ be a $\hal$-gadget of order $d$ and length $L$ in $T$.

Let $d'=d$. 
Note that $H^{d'}$ is a $\hal$-gadget of order $d'$ and length at least $L= f^{d-d'}(L)$ in $T$.
For each graph $H^{d'}$ with $d'\geq 1$ we will either
\begin{enumerate}[(1)]
\item    construct a strong $\hal$-gadget with $k$ junctions
with order $d'$, or 
\item  
find a subsequence $H^{d'-1}$ of $H^{d'}$ that
is an $\hal$-gadget of order $d'-1$ of length at least
$f^{d-(d'-1)}(L)$.
\end{enumerate}
If we ever do (1) we are finished. If from $d'=1$ we do (2) then we find a 2-path of length at least 
$f^{d}(L) > d$, which is a contradiction.

Here is how we proceed from $H^{d'} = y_0,\ldots,y_{\ell}$.
We set $i_0=0$. Then iteratively,
for each $j\in \{1,\ldots,k\}$
we will either construct $H^{d'-1}$ as in (2) or
we find $i_j\in \{i_{j-1} + 2d+1, \ldots,\ell\}$ 
such that $y_{i_j}$ is the source of a length-$d'$ ray
that does not contain $y_{i_j}-1$ or $y_{i_j}+1$.
If we succeed in defining $i_1,\ldots,i_k,i_{k+1}$ in this way 
then $y_0,\ldots,y_{i_{k+1}}$ 
is a strong $\hal$-gadget with $k$ 
junctions of order $d'$ so (1) is satisfied.

Let us now make this argument rigorous; again, assume that $H^{d'}=y_0,\dots, y_\ell$ is a $\hal$-gadget of order $d'$ and length $\ell \geq f^{d-d'}(L)$. Set $i_0 =0$ and, starting with $j=0$, proceed iteratively as follows:
\begin{enumerate}
    \item Let $S_j$ be the set of all indices $i\in\{i_{j-1}+2d+1,\dots,\ell\}$ such that $y_i$ is the source of a length-$d'$ ray that does not contain $y_{i-1}$ and $y_{i+1}$.
    \item If $S_j = \emptyset$ then set $\mathsf{stop}=j$ and terminate. Otherwise, set $i_j=\min S_j$ and $j\leftarrow j+1$, and go back to 1.
\end{enumerate}

We now distinguish two cases: If $\mathsf{stop}\geq k+1$, then we found indices $i_0,\dots,i_{k+1}$ such that $\hat{H}^{d'}:=y_0,\dots,y_{i_{k+1}}$ is a strong hardness gadget of order $d'$ with $k$ junctions; hence we achieved (1) and we are done. Otherwise we have $\mathsf{stop}< k+1$. Let $I_j:=\{i_j,\dots,i_{j+1} -1\}$ for all $0\leq j<\mathsf{stop}$, and let $I_\mathsf{stop} =\{i_{\mathsf{stop}},\dots,\ell\}$. By the pigeon-hole principle, at least one of those intervals, say $I_{j'}$, has size at least $\ell/(\mathsf{stop}+1)\geq \ell/(k+1)$. Now, by construction of our iterative procedure above, we find that the sub-interval $\{i_{j'}+2d+1,\dots, i_{j'+1}-1\}\subseteq I_{j'}$ contains no index $i$ such that $y_i$ is the source of a length-$d'$ ray that does not contain $y_{i-1}$ and $y_{i+1}$. Thus, the subsequence $H^{d'-1}:=y_{i_{j'}+2d+1},\dots, y_{i_{j'+1}-1}$ constitutes a $\hal$-gadget of order $d'-1$. Furthermore, $H^{d'-1}$ has length at least $\ell/(k+1)-2d-1=f(\ell)$. Since $\ell\geq f^{d-d'}(L)$, and since $f$ is monotonically increasing, we find that 
$f(\ell)\geq f^{d-(d'-1)}(L)$. Hence we achieved (2) and we can conclude this case as well.
\end{proof}

Now, by removing the first and the last junction, we can also ensure the existence of a closed strong $\hal$-gadget
\begin{corollary}\label{cor:strongHgadget}
Let $T$ be a tree such that the longest $2$-path in $T$ has length $d\geq 1$, and let $k$ be a positive integer. Then there exists $L>0$ (only depending on $k$ and $d$) such that the following is true: If $T$ contains an $\hal$-gadget of order $d$ and length $L$, then there exists $1\leq d'\leq d$ such that $T$ contains a \emph{closed} strong $\hal$-gadget of order $d'$ with at least $k$ junctions.
\end{corollary}
\begin{proof}
Use  Lemma~\ref{lem:theprevone} with $k+2$ rather than $k$ and observe that every strong $\hal$-gadget with $k+2$ junctions also yields a closed strong $\hal$-gadget with $k$ junctions by removing $i_1$ and $i_{k+2}$ from the list of indices. Since $x_{i_1}$ and $x_{i_{k+2}}$ must have degree at least $3$ (they are inner vertices of a $\hal$-gadget and they are junctions), we obtain that neither $x_{i_2}$ and $x_{i_{k+1}}$ can be forks of $T$.
\end{proof}

\subsubsection{Constructions of $Q$ and $\hat{G}$}
For the scope of this subsection, to avoid notational clutter, we assume the following are given:
\begin{itemize}
    \item Positive integers $k$ and $d$.
    \item A tree $T$ that contains a closed strong $\hal$-gadget $H=x_0,\dots,x_\ell$ of order $d$ with $k$ junctions $x_{i_1},\dots,x_{i_k}$. Additionally, for each $j\in[k]$, we fix a ray $R_j=x_{i_j},r_j^1,\dots,r_j^d$ of length $d$, the source of which is $x_{i_j}$ and which does not contain one of the neighbours $x_{i_j-1}$ and $x_{i_j+1}$ --- note that the $R_j$ must exist as the $x_{i_j}$ are junctions. 
    \item A $k$-vertex cubic graph $\Delta$ containing a Hamiltonian cycle $v_1,\dots,v_k,v_1$.
\end{itemize}

We emphasise that the set of edges of $\Delta$ not contained in the Hamilton cycle must constitute a perfect matching, that is, a set of $k/2$ pairwise non-incident edges. This must be satisfied since $\Delta$ is cubic.

\begin{definition}
The \emph{core} of $H$, denoted by $C(H)$, contains the subsequence $x_{i_1},x_{i_1+1},\dots,x_{i_k-1},x_{i_k}$ and the vertices of the rays $R_j$, that is
\[C(H):=\{x_{i_1},x_{i_1+1},\dots,x_{i_k-1},x_{i_k}\} \cup \bigcup_{j=1}^k V(R_j) \,.\]
\end{definition}

\begin{definition}[$Q(\Delta,T,H)$ and $\tau_Q$]\label{def:Q_tau_caseH}
Set $\ell_j := i_{j+1}-i_j$.
The graph $Q=Q(\Delta,T,H)$ is obtained from $\Delta$ as follows:
\begin{enumerate}
    \item The edge $\{v_k,v_1\}$ is deleted.
    \item For each $j\in\{1,\dots,k-1\}$
    the edge $\{v_j,v_{j+1}\}$ is replaced by 
    a path of length $\ell_j$: \[P_j = v_j,u_j^1,\dots,u_j^{\ell_j-1},v_{j+1}\,,\]
    where the $u_j^t$ are fresh vertices.
    \item Each edge $e=\{v_i,v_j\}$ not contained on the Hamilton cycle, i.e., $j\notin\{i-1,i+1\}$, is replaced by a path $P_{i,j}$ of length~$2d$:
    \[P_{i,j} = v_i,w_i^1,\dots,w_i^{d-1},m(e),w_j^{d-1},\dots,w_j^1,v_j \,,\] 
    where the $w_i^t$ and $w_j^t$ are fresh vertices.
\end{enumerate}
Finally $\tau=\tau(\Delta,T,H)$ is a fracture of $Q$ defined as follows: For each $m(e)$, the partition $\tau_{m(e)}$ contains two singleton blocks, and for all remaining vertices $v$ of $Q$ the partition $\tau_v$ only contains one block.
\end{definition}
Since $\Delta$, $T$ and $H$ are fixed in this subsection, to avoid notational clutter, we just write $Q$ and $\tau$, rather than $Q(\Delta,T,H)$ and $\tau(\Delta,T,H)$. 

It turns out that $Q$ is isomorphic to a quotient graph of $T[C(H)]$ obtained by identifying the endpoints of the rays $R_i$ and $R_j$ for every $\{v_i,v_j\}\in E(\Delta)$ with $j\notin\{i-1,i+1\}$. This induces a homomorphism from $T[C(H)]$ to $Q$ that will be useful in the construction of $\hat{G}$; hence we explicitly define this mapping below:
\begin{definition}[$\gamma$]\label{def:gamma_caseH}
We define a function $\gamma:C(H)\rightarrow V(Q)$ as follows.
\begin{enumerate}
    \item We map the sequence $x_{i_1},x_{i_1+1},\dots,x_{i_k-1},x_{i_k}$ in $C(H)$ to the sequence $v_1,\dots,v_k$ in $Q$. More precisely, for each $j\in\{1,\dots,k-1\}$ and $t\in\{1,\dots,\ell_j-1\}$, we set $\gamma(x_{i_j}):=v_j$ and $\gamma(x_{i_j+t}):=u_j^t$.
    \item For each edge $e=\{v_i,v_j\}$ of $\Delta$ with $j\notin\{i-1,i+1\}$, we map $V(R_i)$ and $V(R_j)$ to the path $P_{i,j}$. More precisely, for each $t\in\{1,\dots,d-1\}$ we set $\gamma(r_i^t):=w_i^t$ and $\gamma(r_j^t)=w_j^t$. Furthermore, we set $\gamma(r_i^d):=m(e) =: \gamma(r_j^d)$. (Note that the images of the sources of the rays $R_i$ and $R_j$ are already set in 1.)
\end{enumerate}
\end{definition}

\begin{observation}
The function $\gamma$ is an edge-bijective homomorphism from $T[C(H)]$ to $Q$. 
\end{observation}

Let us provide the induced egde-bijection explicitly:
\begin{definition}\label{def:edge_subset_caseH}\label{def:gamE}
\label{def:Eprime} 
($E'$, $\gamma_{\text{E}}$)
Define $E':= E(T[C(H)])$, that is, $E'\subseteq E(T)$ contains all edges on the sub-path $x_{i_1},\dots,x_{i_k}$ of $H$ and all edges of the rays $R_1,\dots, R_k$. We write $\gamma_{\text{E}}: E' \to E(Q)$ for the edge-bijection from $E'$ to $E(Q)$ induced by the homomorphism $\gamma$. 
\end{definition}

Now let $(G,c)$ be a $Q$-coloured graph. We state the following fact explicitly, since it will be crucial in our construction:
\begin{observation}\label{obs:Q_coloured_E'}\label{obs:yes}
Let $(G,c)$ be a $Q$-coloured graph. 
The mapping $c_{\text E}\circ \gamma_{\text{E}}^{-1}$ 
is a map from $E(G)$ to $E'$.  
\end{observation}
Our goal is to
construct a graph $\hat{G}=\hat{G}(G,c,T,H)$ from $G$, and an edge-colouring $\hat{\gamma}: E(\hat{G}) \mapsto E(T)$
whose range is $E(T)$  such that
\[\oplus\embs{(\fracture{Q}{\tau},c_\tau)}{(G,c)} = \oplus\colsubs{T}{(\hat{G},\hat{\gamma})},\]
that is, the number of colour-preserving embeddings from the fractured graph $\fracture{Q}{\tau}$ to $(G,c)$ is equal, modulo $2$, to the number of subgraphs of $\hat{G}$ that are isomorphic to $T$ and that contain each edge-colour in $E(T)$ precisely once.

For what follows, let $V(\calR):=\cup_{j=1}^k V(R_j)$ be the set of all vertices of the rays $R_1,\dots,R_k$. We are now able to define $\hat{G}=\hat{G}(G,c,T,H)$; the construction is illustrated in Figure~\ref{fig:HatGDef}.
The definition uses the function $c_{\text{E}}$ introduced in Definition~\ref{def:cE} and the functions
 $\gamma$ and $\gamma_{\text{E}}$  introduced in Definitions~\ref{def:gamma_caseH}
 and \ref{def:gamE}, respectively. 
 It also uses the mapping $c_{\text E}\circ \gamma_{\text{E}}^{-1}$ 
 from $E(G)$ to $E'$ (see Observation~\ref{obs:yes}).

\begin{figure}[t!]
\centering
    \includegraphics[width=\textwidth]{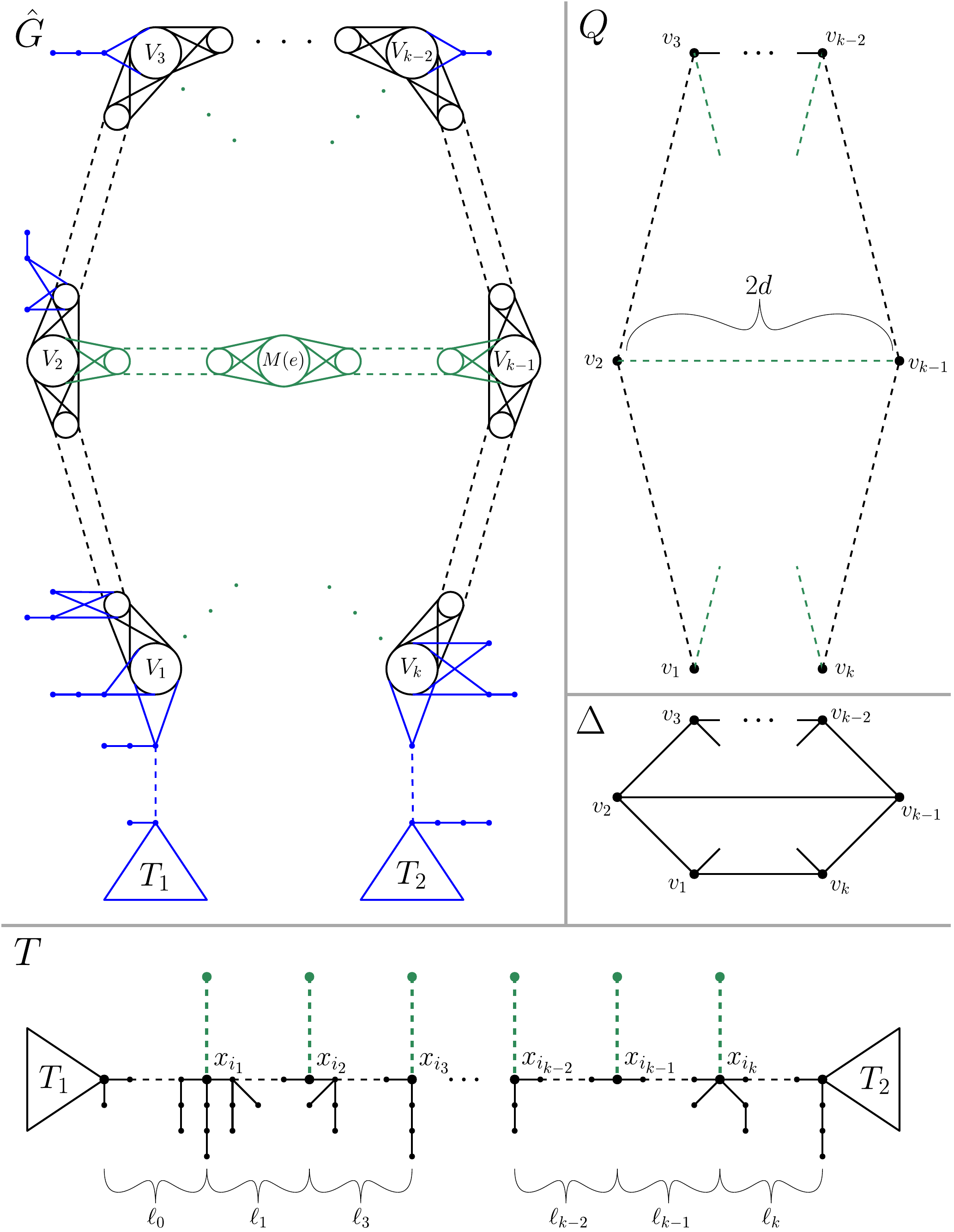}
    \caption{(\emph{Below}): The tree $T$ containing a closed strong $\hal$-gadget of order $d$; the green dashed lines are rays of length $d$. (\emph{Left}): The construction of $\hat{G}=\hat{G}(G,c,T,H)$; note that the removal of the vertices and edges coloured blue yields $G$ (see Definition~\ref{def:hatG_caseH}), and note that $G$ is $Q$-coloured as depicted. (\emph{Right}): The graphs $\Delta$ and $Q=Q(\Delta,T,H)$; we assume in the picture that $\{v_2,v_{k-1}\}$ is an edge of $\Delta$.}
    \label{fig:HatGDef}
\end{figure}

 \vbox{
\begin{definition}[$\hat{G}(G,c,T,H),\hat{\gamma}(G,c,T,H)$]\label{def:hatG_caseH} 
Let $(G,c)$ be a $Q$-coloured graph. The  pair $(\hat{G},\hat{\gamma}) = (\hat{G}(G,c,T,H),\hat{\gamma}(G,c,T,H))$ is an edge-coloured graph constructed as follows, where the co-domain of~$\hat{\gamma}$ is $E(T)$: 
\begin{itemize}
\item[(A)] The graph $\hat{G}$ contains $G$ as a subgraph. 
For each $e\in E(G)$, define $\hat{\gamma}(e) = \gamma^{-1}_{\text{E}}(c_E(e))$.
\item[(B)] The vertex set of $\hat{G}$ is the union of $V(G)$
and  $V(T)\setminus C(H)$.
\item[(C)] Pairs of vertices in $V(T)\setminus C(H)$ 
are connected by an edge in $\hat{G}$  if and only if they are adjacent in $T$. 
For each such edge~$e$, $\hat{\gamma}(e) = e$.

\item[(D)] The remaining edges of $\hat{G}$ are defined as follows.  For each edge $e\in E(T)$ that connects a vertex $z \in V(T)\setminus C(H)$   to a vertex $y\in C(H)$ there are corresponding edges in $\hat{G}$. These edges 
connect~$z$ to  all vertices $g\in V(G)$ such that $c(g)=\gamma(y)$ 
For each such edge $e'$ in~$\hat{G}$,
$\hat{\gamma}(e') = e $.

\end{itemize}
\end{definition}}

\noindent Observe that for each element $\colT\in\colsubs{T}{(\hat{G},\hat{\gamma})}$ 
the induced subgraph \[\colT[G]:=\colT[V(\colT)\cap V(G)]\] of $\colT$ is an edge-colourful subgraph in $G$, that is, $\colT[G]$ contains precisely one edge per edge-colour of $G$ 
under the edge colouring $\hat{\gamma}$
hence it contains precisely one edge 
per edge-colour of $G$ under
$\cE$.
  As shown in Section~3 in the full version~\cite{RothSW20arxiv} of~\cite{RothSW21}, $\colT[G]$ thus induces a fracture $\rho=\rho(\colT)$ of~$Q$: Two edges $\{v,w\}$ and $\{v,y\}$ of $Q$ are in 
the same block 
in the partition $\rho_v$ corresponding to   vertex $v$ of $Q$ 
if and only if the edges of $\colT[G]$
that are coloured $\gamE^{-1}(\{v,w\})$ and 
$\gamE^{-1}(\{v,y\})$ are adjacent.
In what follows, we show that $\rho$ must always be equal to $\tau(\Delta,T,H)$ (see Definition~\ref{def:Q_tau_caseH}).

\begin{lemma}\label{lem:unique_fracture_caseH}\label{lem:prev}
For every $\colT\in\colsubs{T}{(\hat{G},\hat{\gamma})}$ we have that $\rho(\colT) = \tau(\Delta,T,H)$.
\end{lemma}
\begin{proof}
To avoid notational clutter, we set $\rho:=\rho(\colT)$ and $\tau := \tau(\Delta,T,H)$.
Let $T_1$ and $T_2$ be the subtrees of $T$ attached to the ends of the $\hal$-gadget $H$ as shown in the bottom part of Figure~\ref{fig:HatGDef}. 

We  first give an overall intuition of the proof; consider Figure~\ref{fig:HinHatG} for an illustration. Since $\colT$ is isomorphic to $T$, there must be a (unique) path connecting $T_1$ and $T_2$ in $\hat{G}$ (recall that, since $\colT$ is \emph{edge-colourful} and since every edge in $T_1$ and $T_2$ has a different colour --- see (C) in Definition~\ref{def:hatG_caseH} --- $\colT$ must contain all edges in $T_1$ and $T_2$). We claim that this path must follow the outer cycle in $\hat{G}$, in which case the designated rays in $\calR$ of length $d$ at the junctions must follow the inwards direction and thus induce $\tau$. To see why the path connecting $T_1$ and $T_2$ must follow the outer cycle,
first recall that $V_j$ is the subset of $V(G)$ 
coloured by $c$ with $v_j$. 
Then recall 
that the path between $V_j$ and $V_{j+1}$  along the outer cycle in $\hat{G}$ has length $\ell_j \geq 2d+1$. Hence the designated rays in $\calR$ cannot be used to cover all edge colours in the path between $V_j$ and $V_{j+1}$.

\begin{figure}[t!]
    \centering
    \includegraphics[scale=0.9]{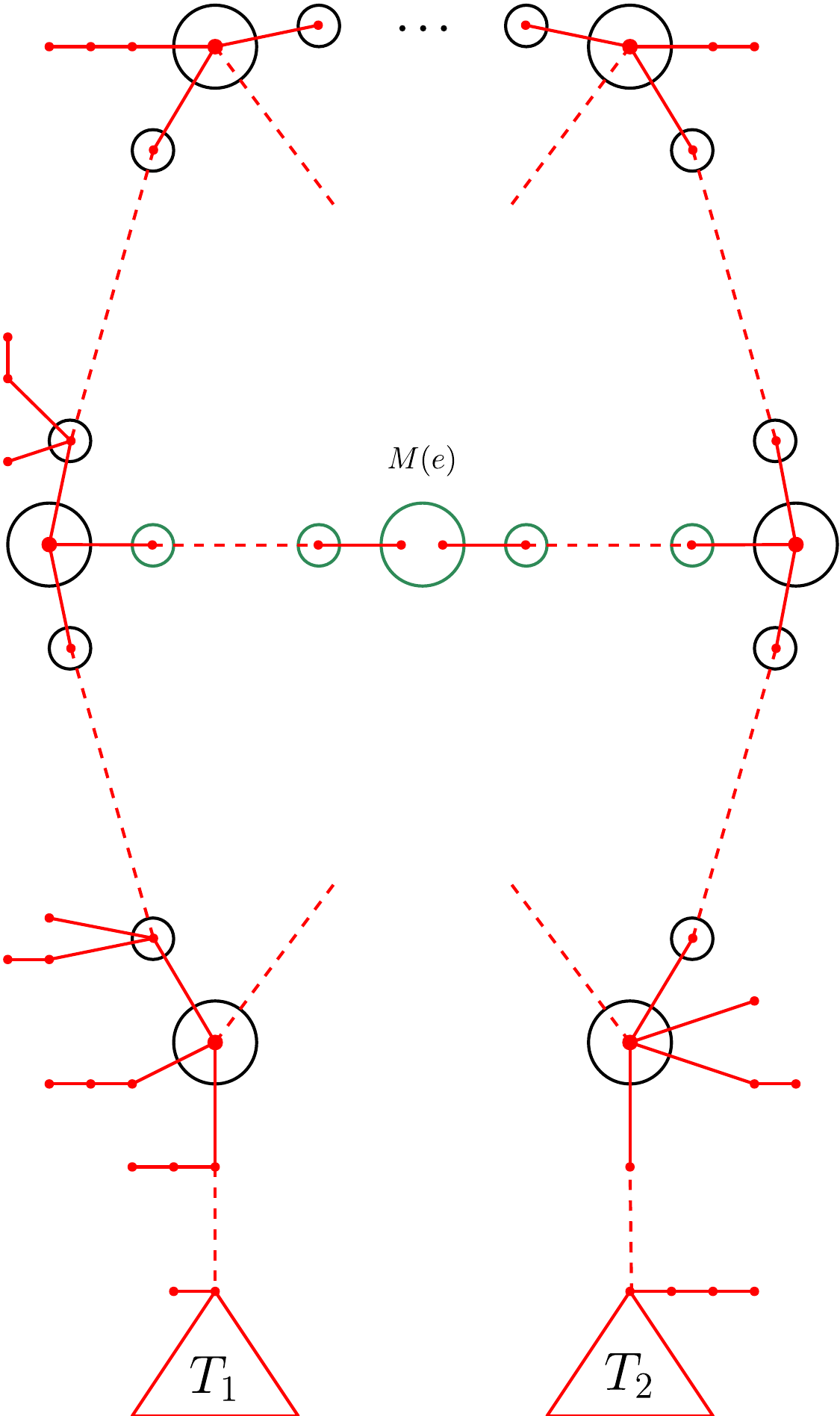}
    \caption{Illustration of Lemma~\ref{lem:unique_fracture_caseH}: The only possibility for an edge-colourful copy of $T$ to be embedded in $\hat{G}$ is depicted in red.  }
    \label{fig:HinHatG}
\end{figure}

We next provide a rigorous argument. 
Let $$S:= V(T_1) \cup V(T_2) \cup \{x_0,\dots,x_{i_1-1}\} \cup \{x_{i_k+1},\dots, x_{k+1}\}.$$ Note that $S$ is a subset of $V(T)\setminus V(H)$ hence it is a subset of $V(T)$
and of   $V(\hat{G})$.

We first claim that  every fork and every ray of length $>d$ of $T$ must be fully contained in the subgraph of $T$ induced by $S$.
This claim follows from the definition of closed strong $\hal$-gadgets. 
In particular, the condition of being closed implies that neither $x_{i_1}$ nor $x_{i_k}$ is a fork.

As a consequence, every fork and every ray of length greater than $d$ of $\colT$ must be contained in the subgraph of~$\hat{G}$ induced by $S$ as well. Additionally, this implies that none of the vertices in $\colT[G]$ can be a fork or the source of a ray of length $>d$ in $\colT$ --- otherwise, $\colT$ would have either more forks or more rays of length $>d$ than $T$, contradicting the fact that $\colT$ and $T$ are isomorphic.
 
Recall that $V_1,\dots,V_k$ denote the subsets of vertices of $G$ that are coloured by $c$ with $v_1,\dots,v_k$.
Now let $P$ be the (unique) path $P$ in $\colT$ that connects $T_1$ with $T_2$. Then, starting with $V_1$ and ending with $V_k$, the path $P$ must pass through a sequence of colour classes $V_1=V_{j_1},V_{j_2},\dots,V_{j_t}=V_k$ of $G$. 
The following claim formalises the idea that
  this sequence must correspond to the Hamilton cycle $v_1,\dots,v_k$ in $\Delta$.
 
{\bf Claim:\quad}
We have $t=k$ and $V_{j_i}=V_i$ for each $i\in[k]$.
 
Before proving the claim, we show that it implies the lemma.
Since, from the claim, $P$ must   follow the outer cycle, the fracture $\rho=\rho(\colT)$ induced by $\colT$  must split the inner paths of length $2d$ (otherwise $\colT$ would contain a cycle).  However, since there are no sources or rays of length greater than $d$ outside of $S$
   in $\colT$, $\rho$ must split all of the inner length-$2d$ paths
    at the central vertex $m(e)$. Furthermore, it cannot split additional vertices since this would disconnect $\colT$.  Thus, $\rho$ is  the fracture $\tau$, concluding the proof. $\blacksquare$

To conclude the proof, we now prove the claim.
Note first that $P$ cannot pass through any of the colour classes $V_i$ more than once as this would cause $\colT$ to use an edge-colour multiple times. 
Next assume for contradiction that $P$ misses some colour class $V_a$ for some $a\in[2,k-1]$ (i.e., we assume that $t<k$). Since $\colT$ is a connected tree 
containing all of the edge colours in~$Q$  
there must be an index  
$j_i\neq a$ and a vertex $u\in V_{j_i}\cap P$ such that $\colT$ contains a (unique) path $P_u$ from $u$ to a vertex  $w\in V_a$. 
In order to get the contradiction, root $\colT$ at $u$.  Construct a subtree $\colT(u)$ of $\colT$ as follows: For each neighbour $x$ of $u$ except the ancestor of $w$ on the path from $u$,
we delete $x$ and all of its descendants. Observe that the edge colours of $\colT(u)$ are disjoint from the edge-colours of $P$
and that $V(\colT(u))$ is disjoint from $S$. Now, if $\colT(u)$ is a path, then (using that $\ell_i>2d$), we obtain that $u$ is the source of a ray in $\colT$ of length greater than $d$,  contradicting the fact that every ray of length $>d$ of $\colT$ is in the subgraph of $\hat{G}$ induced by~$S$.   
Otherwise, $\colT(u)$ contains a fork, contradicting the fact that all forks of $\colT$ are in the subgraph of $\hat{G}$ induced by~$S$.

Having established that $t=k$ and that no $V_i$ is visited more than once, it remains to show that $P$ visits the colour classes in the correct order, that is $V_{j_i}=V_i$ for each $i\in[k]$. Assume for contradiction that this is not the case, which allows us to set
\[m := \min\{i \in[k]~|~V_{j_i} \neq V_i \} -1 \,.\]
Note that $m\geq 1$ since $j_1=1$. Let $z_{m}\in V_{m}\cap P$ and $z_{m+1}\in V_{m+1}\cap P$ and recall that $G$ contains colour classes $U^1_{m},\dots,U^{\ell_{m}-1}_{m}$ corresponding to the path \[P_{m}=v_{m},u^1_{m},\dots,u^{\ell_{m}-1}_{m},v_{m+1}\] in $Q$ (see Definition~\ref{def:Q_tau_caseH}). Let us now define the subtrees $\colT(m)$ and $\colT(m+1)$:
\begin{itemize}
    \item For $\colT(m)$ we root $\colT$ at $z_m$ and for each neighbour $x$ of $z_m$ in $\colT$, we delete $x$ and all of its descendants unless $x\in U^1_{m}$. 
    \item For $\colT(m+1)$ we root $\colT$ at $z_{m+1}$ and for each neighbour $x$ of $z_{m+1}$ in $\colT$, we delete $x$ and all of its descendants unless $x\in U^{\ell_m-1}_{m}$. 
\end{itemize}
Note that at least one of $\colT(m)$ and $\colT(m+1)$ must have depth greater than $d$ (if rooted at $z_m$ and $z_{m+1}$, respectively), since $\ell_m>2d$ and $\colT$ is edge-colourful with respect to $\hat{\gamma}$, that is, we have to make sure that we cover all of the edge colours  
\[\{v_{m},u^1_{m}\}, \{u^1_{m},u^2_{m} \},\dots,\{u^{\ell_{m}-1}_{m},v_{m+1}\} \]
Finally, regardless of which one of the two subtrees has depth greater than $d$, we will find either a fork, or the source of a ray of length greater than $d$ outside of the set $S$, yielding the desired contradiction and concluding the proof of the claim, and hence the proof of the lemma.

\end{proof}  

We are now able to prove the main lemma of this subsection.
\begin{lemma}\label{lem:success_Hgadgets} 
$\oplus\embs{(\fracture{Q}{\tau},c_\tau)}{(G,c)} = \oplus\colsubs{T}{(\hat{G},\hat{\gamma})}$.
\end{lemma}

\begin{proof}
We start with the following claim from~\cite{RothSW20arxiv}.

{\bf Claim:} 
A colour-preserving embedding $\varphi \in \embs{(\fracture{Q}{\tau},c_\tau)}{(G,c)}$ is uniquely defined by its image (which is a subgraph of $(G,c)$).

For convenience, we provide a proof of the claim:
Consider in image $(G',c')$ of~$\varphi$ where $G'$ is 
  a subgraph
 of $G$ and  
$c' = c \mid_{V(G')}$. 
Let $e=\{u,v\}$ be an edge of $G'$ Then $c'(e)=\{c(u),c(v)\}$ is an edge of $Q$ since $c$ is a $Q$-colouring. Recall that $\fracture{Q}{\tau}$ is $Q$-coloured by the function $c_\tau$ that maps $w^B$ to $w$ for each $w\in V(Q)$ and block $B\in \tau_w$. Now 
recall the definition of fractured graphs (Definition~\ref{def:mrhoGeneral}) and  
let $B_1$ and $B_2$ be the blocks of $\tau_{c(u)}$ and $\tau_{c(v)}$ that contain $c(e)$. Then, 
since $\varphi$ is an embedding, it maps $c(u)^{B_1}$ to $u$ and $c(v)^{B_2}$ to $v$. Since $Q$ does not have isolated vertices, continuing this process over all edges of $G'$ defines~$\varphi$. This concludes the proof of the claim. $\blacksquare$

By the claim, it is sufficient to construct a bijection $b$ from elements in $\colsubs{T}{(\hat{G},\hat{\gamma})}$ to subgraphs $(G',c')$ that are
images of embeddings in $\embs{(\fracture{Q}{\tau},c_\tau)}{(G,c)}$. Given $\colT\in \colsubs{T}{(\hat{G},\hat{\gamma})}$ we set $b(\colT):=(\colT[G],c({\colT}))$
where $c({\colT})$ is the colouring 
of vertices of $\colT[G]$ which agrees with $\hat{\gamma}$ on the edges of $\colT[G]$.  In the rest of the proof,  we show that $b$ is the desired bijection.

First, we have to show that
for all $\colT$,
$(\colT[G],c(\colT))$ is the image of an embedding in $\embs{(\fracture{Q}{\tau},c_\tau)}{(G,c)}$. To this end, recall that $\colT[G]$ induces a fracture $\rho=\rho(\colT)$ of $Q$. By the definition of $\rho$,
$\colT[G]$ and $\fracture{Q}{\rho}$ are isomorphic and this isomorphism preserves the colours so 
$c_\rho$ agrees with $\hat{\gamma}$ on
the edges of $\fracture{Q}{\rho}$.
This implies that  $c_\rho$ and $c(\colT)$ are the same. 
So
$(\colT[G],c(\colT))$ is the image of an embedding in $\embs{(\fracture{Q}{\rho},c_\rho)}{(G,c)}$. 
Finally, 
Lemma~\ref{lem:prev}
guarantees that $\rho=\tau$.

Second, we will show that $b$ is injective. To this end, let $\colT_1 \neq \colT_2 \in \colsubs{T}{(\hat{G},\hat{\gamma})}$. Since $\colT_1$ and $\colT_2$ must both fully contain $V(T)\setminus C(H)$, and since both are edge-colourful (see Definition~\ref{def:hatG_caseH}), the only possibility for $\colT_1$ and $\colT_2$ not being equal is that they disagree on $G$, that is, $\colT_1[G]\neq \colT_2[G]$. This proves $b$ to be injective.

Finally, we will show that $b$ is surjective: Given any $(G',c')$
that is the image of an embedding $\varphi\in\embs{(\fracture{Q}{\tau},c_\tau)}{(G,c)}$,
we construct $\colT(G',c') \in \colsubs{T}{(\hat{G},\hat{\gamma})}$ with $b(\colT(G',c'))=(G',c')$   as follows. Observe first that  $G'$ is isomorphic to $T[C(H)]$ since  $\fracture{Q}{\tau}$ is, by definition of $\tau$, isomorphic to $T[C(H)]$: Splitting the inner paths of length $2d$ in $Q$ at their central vertices yields precisely $T[C(H)]$. 
Then $\colT(G',c')$ is obtained by adding the remainder of $T$ to $(G',c')$:  
\begin{enumerate}
    \item We add to $(G',c')$ all vertices in $V(T)\setminus C(H)$ (see (B) in Definition~\ref{def:hatG_caseH}).
    \item We add all edges between vertices in $V(T)\setminus C(H)$ that are present in $\hat{G}$ (see (C) in Definition~\ref{def:hatG_caseH}).
    \item Finally, we connect a vertex in $z$ in $V(T)\setminus C(H)$ with a vertex $w$ in $G'$ if and only if $z$ and $w$ are connected in $\hat{G}$ (see (D) in Definition~\ref{def:hatG_caseH}).
\end{enumerate}
The resulting subgraph $\colT(G',c')$ of $\hat{G}$ is clearly edge-colourful and isomorphic to $T$, concluding the proof.
\end{proof}

We are now able to establish hardness of $\oplus\subsprob(\mathcal{T})$ in case of unbounded $\hal$-number.
\begin{lemma}\label{lem:TreeHardnessCaseH}
Let $\mathcal{T}$ be a recursively enumerable class of trees of unbounded $\hal$-number. Then $\oplus\subsprob(\mathcal{T})$ is $\oplus\W{1}$-hard.
\end{lemma}
\begin{proof}
Assume first that $\calT$ contains trees with $2$-paths of unbounded length. 
In this case we reduce from the problem of counting $k$-cycles, modulo $2$, which was shown $\oplus\W{1}$-hard in~\cite{CurticapeanDH21}. In the first step,  this problem reduces to the problem of counting $s$-$t$-paths of length $k$, modulo $2$ as shown in Lemma 5.2 in the full version~\cite{PeyerimhoffR0SV21Arxiv} of~\cite{PeyerimhoffR0SV21MFCS}. In the second and final step, we can easily reduce from the problem of counting $s$-$t$-paths of length $k$, modulo $2$, to $\oplus\subsprob(\mathcal{T})$, as shown in Figure~\ref{fig:F_long2paths}: Concretely, let $(G,s,t,k)$ be a problem instance. Since $\calT$ contains trees with $2$-paths of unbounded length, we can find, in time only depending on $k$, a tree $T$ in $\calT$ containing a $2$-path $x_0,x_1,\dots,x_{k+1},x_{k+2}$ of length $k+2$. Let furthermore $T_1$ and $T_2$ be the subtrees of $T$ as depicted in Figure~\ref{fig:F_long2paths}. We construct a graph $G'$ from $G$ in two steps as follows: First, we add fresh vertices $x_0$ and $x_{k+2}$ and edges $\{x_0,s\}$ and $\{t,x_{k+2}\}$. Second, we add $T_1$ and $T_2$ and identify their roots with $x_0$ and $x_{k+2}$, respectively. The construction is depicted in Figure~\ref{fig:F_long2paths} as well. Now let $A$ be the set of subgraphs of $G'$ that are isomorphic to $T$ and that contain all edges of $T_1$ and $T_2$. It is easy to see that the cardinality of $A$ is equal to the number of $s$-$t$-paths of length $k$ in $G$. Thus it suffices to compute $|A|\mod 2$, using an oracle for $\oplus\subsprob(\mathcal{T})$. This can be achieved by a simple application of the inclusion-exclusion principle: Setting $S= E(T_1) \cup E(T_2)$, we have
\begin{equation}\label{eq:incl_excl_T1T2}
    |A| = \sum_{J\subseteq S} (-1)^{|J|} \cdot \#\subs{T}{G'\setminus J}\,,
\end{equation}
where $G'\setminus J$ is the graph obtained from $G'$ by deleting all edges in $J$. We can conclude the reduction by observing that the number of terms in~\eqref{eq:incl_excl_T1T2} only depends on $T$ and thus on $k$, and that our oracle to $\oplus\subsprob(\mathcal{T})$ allows us to evaluate~\eqref{eq:incl_excl_T1T2} modulo $2$.

For the remainder of the proof we can thus assume that the length of any $2$-path in any tree in $\calT$ is bounded by a constant $d$. Since $\calT$ has unbounded $\hal$-number, we obtain that the trees in $\calT$ contain $\hal$-gadgets of order $d$ of unbounded length. By Corollary~\ref{cor:strongHgadget} we obtain that for any positive integer $k$, there is
a value
$d'$ in the range $1\leq d' \leq d$ such that there is
a tree $T_k$ in $\calT$ which contains a strong $\hal$-gadget of order $d'$ with $k$ junctions.

Let $\mathcal{C}$ be a class of cubic Hamiltonian graphs of unbounded treewidth. Assume w.l.g.\ that, for each $k$, the class $\mathcal{C}$ contains at most one graph with $k$ vertices; otherwise we just keep one $k$-vertex graph with the largest treewidth among all $k$-vertex graphs in $\mathcal{C}$. For each $\Delta\in\mathcal{C}$ set $T_\Delta:=T_{|V(\Delta)|}$,  
that is $T_\Delta$ is contained in $\calT$ and contains a strong $\hal$-gadget $H_\Delta$ with at least $|V(\Delta)|$ junctions. Recall Definition~\ref{def:Q_tau_caseH} and set 
\[\mathcal{Q}:= \{Q(\Delta,T_\Delta,H_\Delta) ~|~\Delta \in \mathcal{C}\} \,.\]
Observe that $Q(\Delta,T_\Delta,H_\Delta)$ contains as minor the graph obtained from $\Delta$ by removing one edge. Since the removal of a single edge can decrease the treewidth only by a constant, and since treewidth is minor-monotone, we have that $\mathcal{Q}$ has unbounded treewidth.

By Theorem~\ref{thm:cphom_lower_bound} the problem $\oplus\cphomsprob(\mathcal{Q})$ is therefore $\oplus\W{1}$-hard. Thus it suffices to show that 

\[\oplus\cphomsprob(\mathcal{Q})\fptred \oplus\subsprob(\mathcal{T})  \,.\]

In the first step, we reduce the computation of $\oplus\homs{(Q,\mathsf{id}_Q)}{\star}$ to the computation of $\oplus\embs{(\fracture{Q}{\tau},c_\tau)}{\star}$; here, $\tau$ is the fracture defined in Definition~\ref{def:Q_tau_caseH}. To this end, it was shown in~\cite{PeyerimhoffRSSVW22} that 
\begin{equation}
    \oplus\embs{(\fracture{Q}{\tau},c_\tau)}{\star} = \sum_{\rho\geq \tau} \mu(\tau,\rho) \cdot \oplus\homs{(\fracture{Q}{\rho},c_\rho)}{\star}\,,
\end{equation}
where the relation ``$\geq$'' and the M\"obius function $\mu$ are over the lattice of fractures. We omit introducing these objects in detail, since we only require that the coefficient of the term $\oplus\homs{(\fracture{Q}{\top},c_\top)}{\star}$ (which is equal to $\oplus\homs{(Q,\mathsf{id}_Q)}{\star}$) in the above linear combination was shown in~\cite{PeyerimhoffRSSVW22} to be equal to
\[\prod_{v\in V(Q)} (-1)^{|\tau_v|-1}\cdot (|\tau_v|-1)!\,. \]
Since each partition $\tau_v$ has at most two blocks, the above term is odd. Thus, by Lemma~\ref{lem:complexity_monotonicty}, we can evaluate the term $\oplus\homs{(\fracture{Q}{\top},c_\top)}{\star}$ if we can evaluate the entire linear combination, that is, if we can evaluate $\oplus\embs{(\fracture{Q}{\tau},c_\tau)}{\star}$. 
It thus remains to show how we can evaluate $\oplus\embs{(\fracture{Q}{\tau},c_\tau)}{\star}$ using our oracle for $\oplus\subsprob(\mathcal{T})$.

To this end, we use Lemma~\ref{lem:success_Hgadgets}: Given any $Q=Q(\Delta,T_\Delta,H_\Delta)$-coloured graph $(G,c)$ for which we want to compute $\oplus\embs{(\fracture{Q}{\tau},c_\tau)}{(G,c)}$, we first construct $(\hat{G},\hat{\gamma})$ as in Definition~\ref{def:hatG_caseH}. Then Lemma~\ref{lem:success_Hgadgets} yields that
\[\oplus\embs{(\fracture{Q}{\tau},c_\tau)}{(G,c)} = \oplus\colsubs{T_\Delta}{(\hat{G},\hat{\gamma})}.\]

Finally, by Lemma~\ref{lem:remove_edgecols} we can compute $\oplus\colsubs{T_\Delta}{(\hat{G},\hat{\gamma})}$ in FPT time using an oracle for $\oplus\subs{T_\Delta}{\star}$. 
Since the size of $T_\Delta$ only depends on $Q$, and since, with input $Q$ we can find $T_\Delta$ (recall that $\mathcal{T}$ is recursively enumerable) this yields indeed a parameterised Turing-reduction and the proof is concluded.  
\end{proof}

\begin{figure}[t!]
    \centering
    \includegraphics[width=\textwidth]{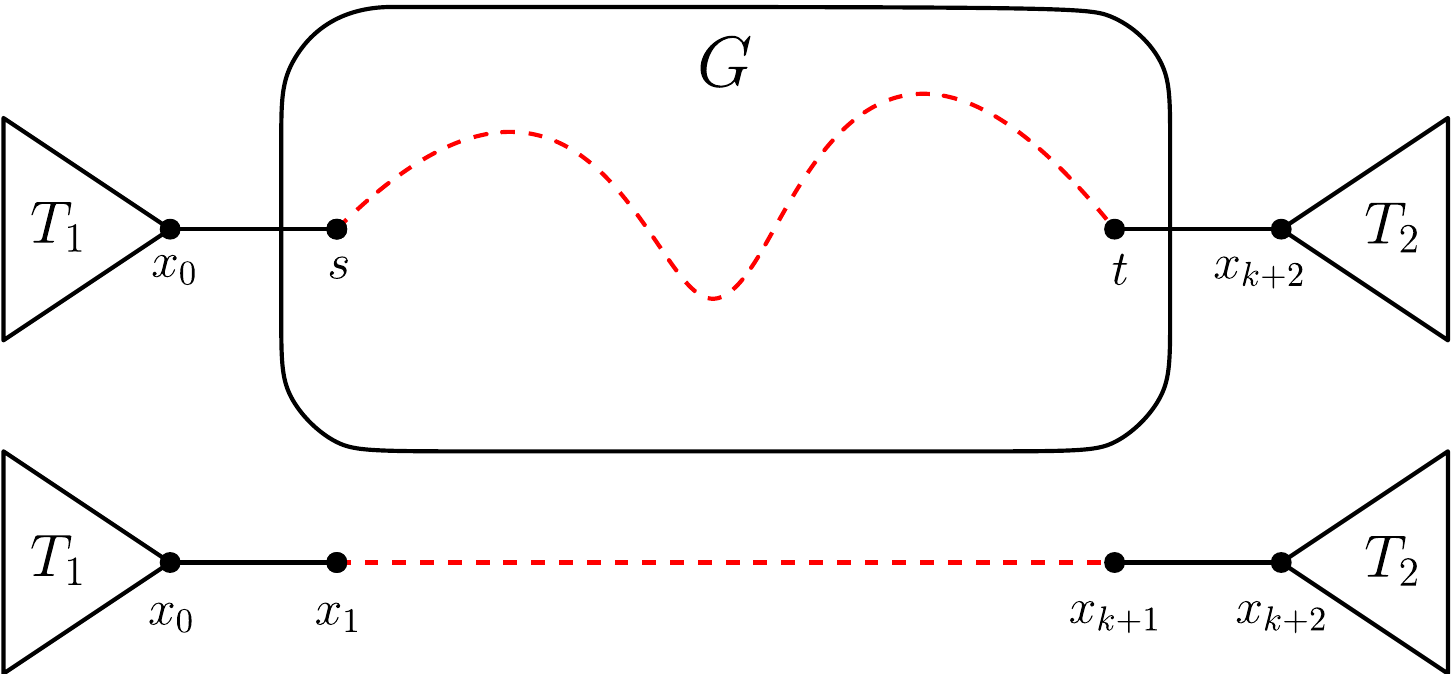}
    \caption{Reduction from counting $s$-$t$-paths of length $k$, modulo $2$, in a graph $G$ to counting copies of a tree $T$ with a $2$-path of length at least $k+2$.}
    \label{fig:F_long2paths}
\end{figure}

\subsection{Unbounded Star Number}\label{sec:caseStar}
We will use the same strategy as in  Subsection~\ref{sec:caseH}: Given a tree $T$ with large star number, we start with a properly chosen cubic graph $\Delta$, and we construct a graph $Q$ depending on $\Delta$ and $T$ which contains $\Delta$ as a minor. Then we show that for any $Q$-coloured graph $(G,c)$, we can construct an edge-coloured graph $(\hat{G},\hat{\gamma})$ such that $\oplus\colsubs{T}{(\hat{G},\hat{\gamma})}$ is equal to $\oplus\embs{(\fracture{Q}{\tau},c_\tau)}{(G,c)}$ for a particular fracture $\tau$.

To this end, let $T$ be a tree with star number (at least) $6k$ for some positive integer $k$. By definition of the star number, there is a $d\geq 3$ such that $T$ contains a vertex $s$ which is the source of $6k$ rays $R_1,\dots,R_{6k}$ of length precisely $d$. For each $i\in[6k]$, let $R_i=s,r_i^1,\dots,r_i^d$. Furthermore, let $T_s$ be the subtree of $T$ obtained by deleting the vertices $r_i^1,\dots,r_i^d$ for each $i\in[6k]$; consider Figure~\ref{fig:A_treewithstar} for an illustration.
\begin{figure}
    \centering
    \includegraphics[scale=0.8]{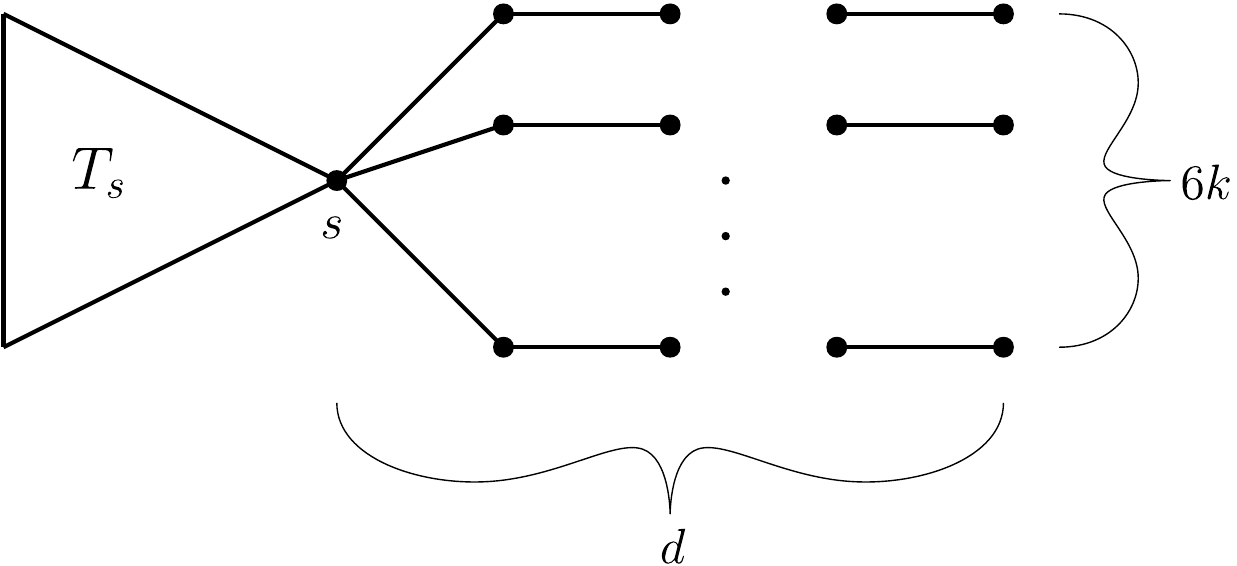}
    \caption{A tree with $\starnum_d(T)\geq 6k$. }
    \label{fig:A_treewithstar}
\end{figure}

\begin{definition}[$Q$]\label{def:Qstar}
Let $\Delta$ be cubic graph on $k$ vertices. We obtain $Q$ from $\Delta$ by substituting each vertex $v$ by a gadget depicted in Figure~\ref{fig:B_QCaseStar}. Afterwards, we connect the gadgets as follows: If $\{v,x\}$ is an edge of $\Delta$, then we identify the vertex $v_x$ in the gadget of $v$ and the vertex $x_v$ in the gadget of $x$. 
\end{definition}

\begin{figure}
    \centering
    \includegraphics[width=\textwidth]{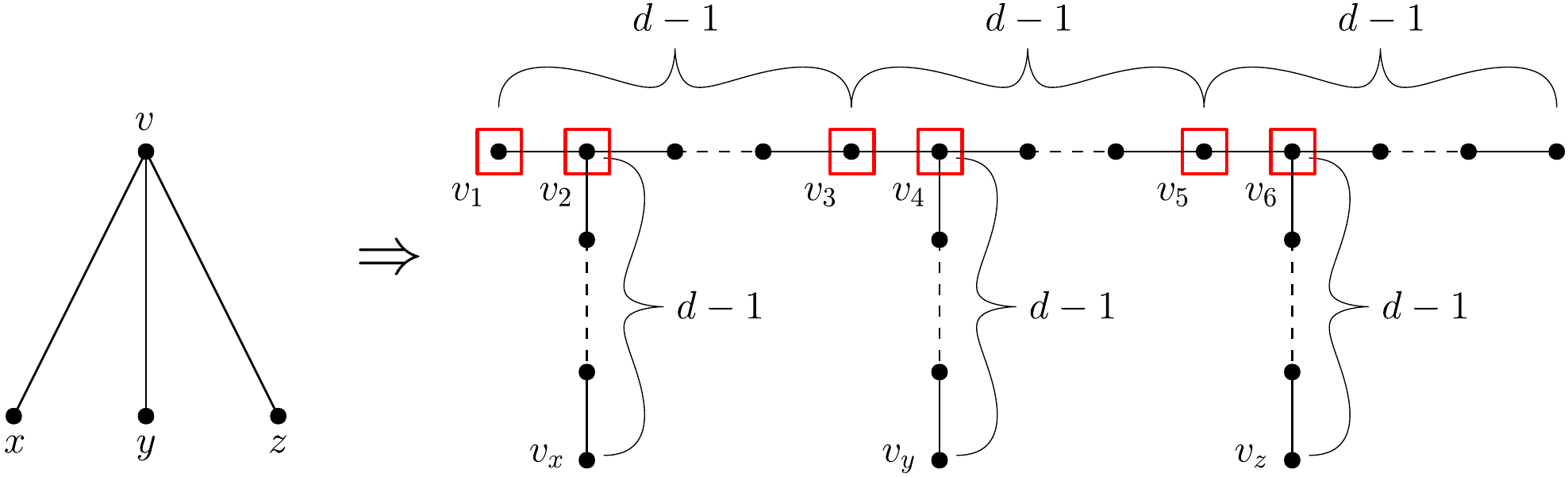}
    \caption{The construction of $Q$; the vertices $v_1,\dots,v_6$ on the gadget of $v$ are emphasized.}
    \label{fig:B_QCaseStar}
\end{figure}

\begin{observation}
$\Delta$ is a minor of $Q$.
\end{observation}

The fracture $\tau$ of $Q$ that we will be interested in is defined as follows; Figure~\ref{fig:D_tauCaseStar} depicts the fractured graph $\fracture{Q}{\tau}$.

\begin{definition}[$\tau$]\label{def:taustar}
Let $Q$ be the graph defined in Definition~\ref{def:Qstar}.
\begin{itemize}
    \item For each edge $\{v,x\}$ of $\Delta$, the graph $Q$ contains a vertex $v_x(=x_v)$, which has degree $2$. We let $\tau_{v_x}$ be the partition consisting of $2$ singleton blocks.
    \item For each vertex $v$ of $\Delta$, the vertices $v_3$ and $v_5$ have degree $2$ in $Q$. We let $\tau_{v_3}$ and $\tau_{v_5}$ be the partitions consisting of $2$ singleton blocks.
    \item For each vertex $v$ of $\Delta$, the vertices $v_2$, $v_4$ and $v_6$ have degree $3$ in $Q$. For each $i\in\{2,4,5\}$ we let $\tau_{v_i}$ be the partition consisting of one block of size $2$ corresponding to the edges incident to $v_i$ from the left and the right, and one block of size $1$ corresponding to the edge incident to $v_i$ from below.
\end{itemize}
For all other vertices $u$ of $Q$, we let $\tau_u$ be the partition consisting only of one block.
\end{definition}

\begin{figure}
    \centering
    \includegraphics[width=\textwidth]{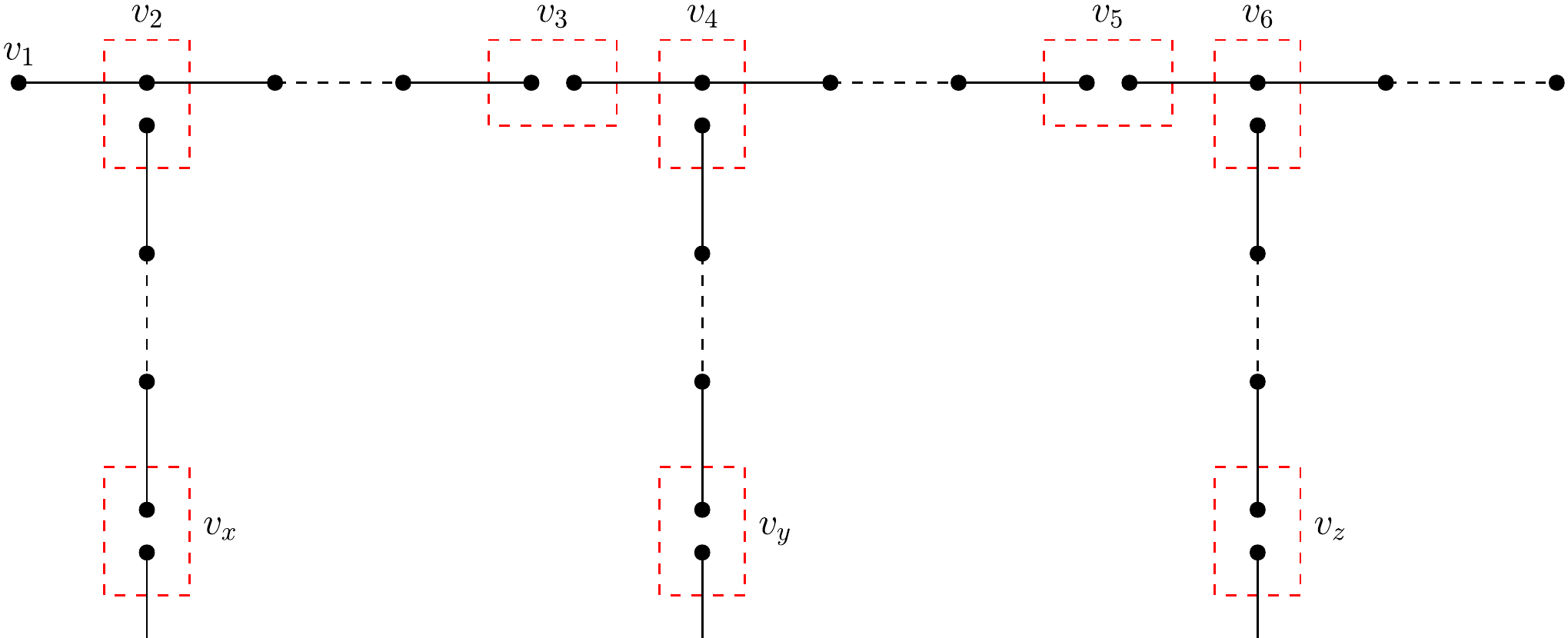}
    \caption{Illustration of the fractured graph $\fracture{Q}{\tau}$ via fracturing the vertex gadgets.}
    \label{fig:D_tauCaseStar}
\end{figure}

Analogously to the notion of a core in the case of unbounded $\hal$-number, we will identify a specific subgraph of the tree $T$  and we will use it to define the graph $\hat{G}$ later. 
\begin{definition}[$V'$]\label{def:coreStarnum}
Let $V'$ be the vertex subset of $T$ defined as follows:
\[V' := \left(\bigcup_{i\in[6k]} V(R_i)\right)\setminus \{s\} \,.\]
Furthermore, we set $E':=E(T[V'])$.
\end{definition}

Observe that $T[V']$ is a (disjoint) union of $6k$ paths of length $d-1$, where the vertices of the $i$-th path are $r_i^1,\dots,r_i^d$. Observe further that $V(T)= V(T_s) \dot \cup V'$ and that 
\begin{equation}\label{eq:edge_partition_stars}
    E(T)= E' ~\dot\cup~ E(T_s) ~\dot\cup~ \{ \{s,r_i^1\}~|~i \in[6k]\}\,.
\end{equation}

Next, note that the edges of $Q$ can be decomposed into $6k$ paths, each of length $d-1$: There are $k$ vertices of $\Delta$. For each vertex $v\in V(\Delta)$ the graph $Q$ contains, by definition, a gadget corresponding to $v$, the edges of which can be decomposed into $6$ paths $P_v^1,\dots,P_v^6$ of length $d-1$ (formally, the fractured graph $\fracture{Q}{\tau}$ yields precisely this decomposition; see Figure~\ref{fig:D_tauCaseStar}). Additionally, for each $v\in V(\Delta)$ and $i\in[6]$, the first vertex of $P_v^i$ is chosen to be $v_i$ as depicted in Figure~\ref{fig:B_QCaseStar}.

\begin{definition}[$\gamma,\gamma_E$]\label{def:gamma_gammaE_stars}
We define a function $\gamma:T[V'] \to V(Q)$ as follows. Recall that $T[V']$ is the union $6k$ paths $P'_j:= r_j^1,\dots,r_j^d$ for $j\in[6k]$. Fix any bijection $b:[6k] \to V(\Delta) \times [6]$. Then $\gamma$ maps $P'_j$ to $P_{v}^{i}$, where $b(j)=(v,i)$. In particular, we enforce that the first vertices of the paths are mapped onto each other, that is, $\gamma(r_j^1):=v^{i}$. Additionally, we define $\gamma_E:E' \to E(Q)$ by mapping $e$ to $\gamma(e)$.
\end{definition}

\begin{observation}
The function $\gamma$ is an edge-bijective homomorphism from $T[V']$ to $Q$. Specifically, $\gamma_E$ is a bijection.
\end{observation}

Now let $(G,c)$ be a $Q$-coloured graph. We state the following explicitly, since it will be crucial in our reduction.
\begin{observation}
Let $(G,c)$ be a $Q$-coloured graph. The mapping $c_E \circ \gamma_E^{-1}$ is a map from $E(G)$ to $E'$.
\end{observation}

Let us now construct a graph $\hat{G}$ from a $Q$-coloured graph $G$; an illustration is provided in Figure~\ref{fig:C_hatGCaseStar}.

\begin{definition}[$(\hat{G},\hat{\gamma})$]\label{def:hatG_caseStar}
Let $(G,c)$ be a $Q$-coloured graph. The graph $\hat{G}$ is an edge-coloured graph, with colouring $\hat{\gamma}\colon E(\hat{G}) \to E(T)$, constructed as follows: 
\begin{itemize}
\item[(A)] The graph $\hat{G}$ contains $G$ as a subgraph. For each $e\in E(G)$ we set $\hat{\gamma}(e)=\gamma_E^{-1}(c_E(e))$.
\item[(B)] The vertex set of $\hat{G}$ is the union of $V(G)$
and  $V(T_s)$, and pairs of vertices in $V(T_s)$ 
are connected by an edge in $\hat{G}$  if and only they are adjacent in $T$. 
For each such edge~$e$, $\hat{\gamma}(e) = e$. 
\item[(C)] The remaining edges of $\hat{G}$ are defined as follows. For each edge $e=\{s,r_j^1\}\in E(T)$, we connect $s$ to all vertices in $G$ that are coloured (by $c$) with $\gamma(r_j^1)$ (see Definition~\ref{def:gamma_gammaE_stars}), and for each of those newly added edges $e'$ we set $\hat{\gamma}(e'):=e$
\end{itemize}
\end{definition}
Observe that $\hat{\gamma}$ colours the edges of $\hat{G}$ with $E(T)$; the cases (A), (B), and (C) correspond, respectively, to the sets $E'$, $E(T_s)$ and $\{ \{s,r_i^1\}~|~i \in[6k]\}$ (see Equation~(\ref{eq:edge_partition_stars})).
Similarly to the case of unbounded $\hal$-gadgets,  for each element $\colT\in\colsubs{T}{(\hat{G},\hat{\gamma})}$ 
the induced subgraph \[\colT[G]:=\colT[V(\colT)\cap V(G)]\] of $\colT$ is an edge-colourful subgraph in $G$, that is, $\colT[G]$ contains precisely one edge per edge-colour of $G$ 
under the edge colouring $\hat{\gamma}$
hence it contains precisely one edge 
per edge-colour of $G$ under
$\cE$.
  As shown in Section~3 in the full version~\cite{RothSW20arxiv} of~\cite{RothSW21}, $\colT[G]$ thus induces a fracture $\rho=\rho(\colT)$ of~$Q$: Two edges $\{v,w\}$ and $\{v,y\}$ of $Q$ are in 
the same block 
in the partition $\rho_v$ corresponding to   vertex $v$ of $Q$ 
if and only if the edges of $\colT[G]$
that are coloured $\gamE^{-1}(\{v,w\})$ and 
$\gamE^{-1}(\{v,y\})$ are adjacent.
In what follows, we show that $\rho$ must always be equal to $\tau(\Delta,T,H)$ (see Definition~\ref{def:taustar}).

\begin{figure}
    \centering
    \includegraphics[scale=0.75]{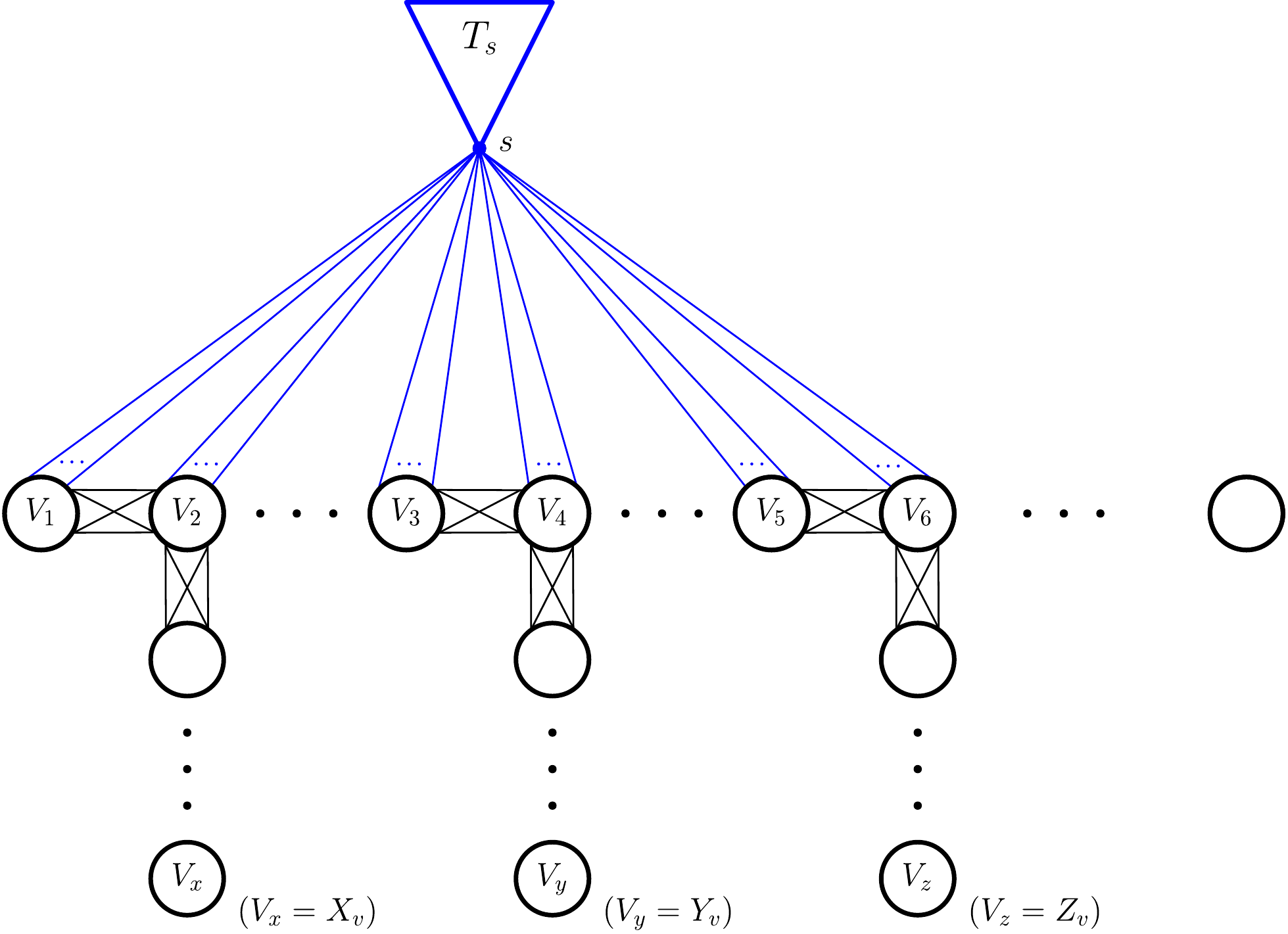}
    \caption{The construction of $\hat{G}$. The graph $G$ within $\hat{G}$ is depicted in black.}
    \label{fig:C_hatGCaseStar}
\end{figure}

\begin{lemma}\label{lem:unique_fracture_caseStar}
For every $\colT\in\colsubs{T}{(\hat{G},\hat{\gamma})}$ we have that $\rho(\colT) = \tau$.
\end{lemma}
\begin{proof}
Let $\colT\in\colsubs{T}{\hat{G},\hat{\gamma}}$. Since $\colT$ must include each of the edge colours given by $\hat{\gamma}$ (precisely) once, we have that $\colT$ must fully contain $T_s$. Note that $T_s$ fully contains $T$ except for $6k$ rays of length $d$, and the only way to attach those rays in $\hat{G}$ is via the vertex $s$.
Now consider the subgraph $\colT[G+s]$ of $\colT$ defined as follows:
\[\colT[G+s]:=\colT[(V(\colT)\cap V(G))\cup \{s\} ]\,.\]
Since $\colT$ includes all edge colours given by $\hat{\gamma}$, we have that $s$ must have degree $6k$ in $\colT[G+s]$: By (C) in Definition~\ref{def:hatG_caseStar}, the vertex $s$ must be connected (within $\colT[G+s]$) to one vertex in each of the colour classes $V_i=c^{-1}(v_i)$ for $v\in V(\Delta)$ and $i\in[6]$. Additionally, this implies the following:
\begin{observation}\label{obs:colTisoSTretchStar}
    $\colT[G+s]$ is isomorphic to the $d$-stretch of $K_{1,6k}$ with $s$ at the centre.
\end{observation}

In the remainder of the proof, we will show that the only way for $\colT$ to (colourfully) embed the $6k$ rays of length $d$ is as depicted in Figure~\ref{fig:E_EmbedTCaseStar}. Note that this will conclude the proof since the induced fracture of the depicted embedding is $\tau$.

Hence we proceed with proving the 
claim. 
We first consider, for each edge $\{v,x\}\in E(\Delta)$, the vertex $v_x =(x_v)$ of $Q$ (see Definition~\ref{def:Qstar} and Figure~\ref{fig:B_QCaseStar}). The vertex $v_x$ has two neighbours $n_v$ and $n_x$ in $Q$, where $n_v$ denotes the neighbour in the gadget of $v$ and $n_x$ denotes the neighbour in the gadget of $x$.
Recall that we write $V_x =c^{-1}(v_x), N_v =c^{-1}(n_v), N_x =c^{-1}(n_x) \subseteq V(G)$ for their colour class within $G$ (and thus within $\hat{G}$). Since $\colT$ is edge-colourful, it must contain precisely one edge $e_v$ between $V_x$ and $N_v$ and one edge $e_x$ between $V_x$ and $N_x$ (see (A) in Definition~\ref{def:hatG_caseStar}). Now observe that every vertex in $V_x$ has distance (at least) $d$ to $s$ within $\hat{G}$. This has two crucial consequences:
\begin{itemize}
    \item First, the endpoints of $e_v$ and $e_x$ inside $V_x$ cannot be equal: Otherwise, they could not be part of a ray of length precisely $d$ with source $s$, and this would contradict the previous observation that  $\colT[G+s]$ is isomorphic to the $d$-stretch of $K_{1,6k}$ with $s$ at the centre (Observation~\ref{obs:colTisoSTretchStar}).
    \item Hence, second, the endpoints of $e_v$ and $e_x$ inside $V_x$ both have degree $1$. Consequently, they must be the endpoints of two of the rays of length $d$. However, the only way for this to be true is them each being connected to $s$ as depicted in Figure~\ref{fig:E_EmbedTCaseStar};
    in all other cases, $\colT[G+s]$ cannot be isomorphic to the $d$-stretch of $K_{1,6k}$ with $s$ at the centre.
\end{itemize}
The second consequence  implies that the edge colours corresponding to the edges in the paths $P_v^2$, $P_v^4$, and $P_v^6$ are covered for each $v$ (recall that $\colT$ must include each edge colour precisely once).
Thus, the only possibility to include the remaining edge colours corresponding to the paths $P_v^1$, $P_v^3$, and $P_v^5$ while keeping $\colT[G+s]$ being isomorphic to the $d$-stretch of $K_{1,6k}$, is to embed, for each gadget, the remaining $3$ rays of length $d$ as depicted in Figure~\ref{fig:E_EmbedTCaseStar}. This concludes the proof.
\end{proof}

\begin{figure}
    \centering
    \includegraphics[scale=0.75]{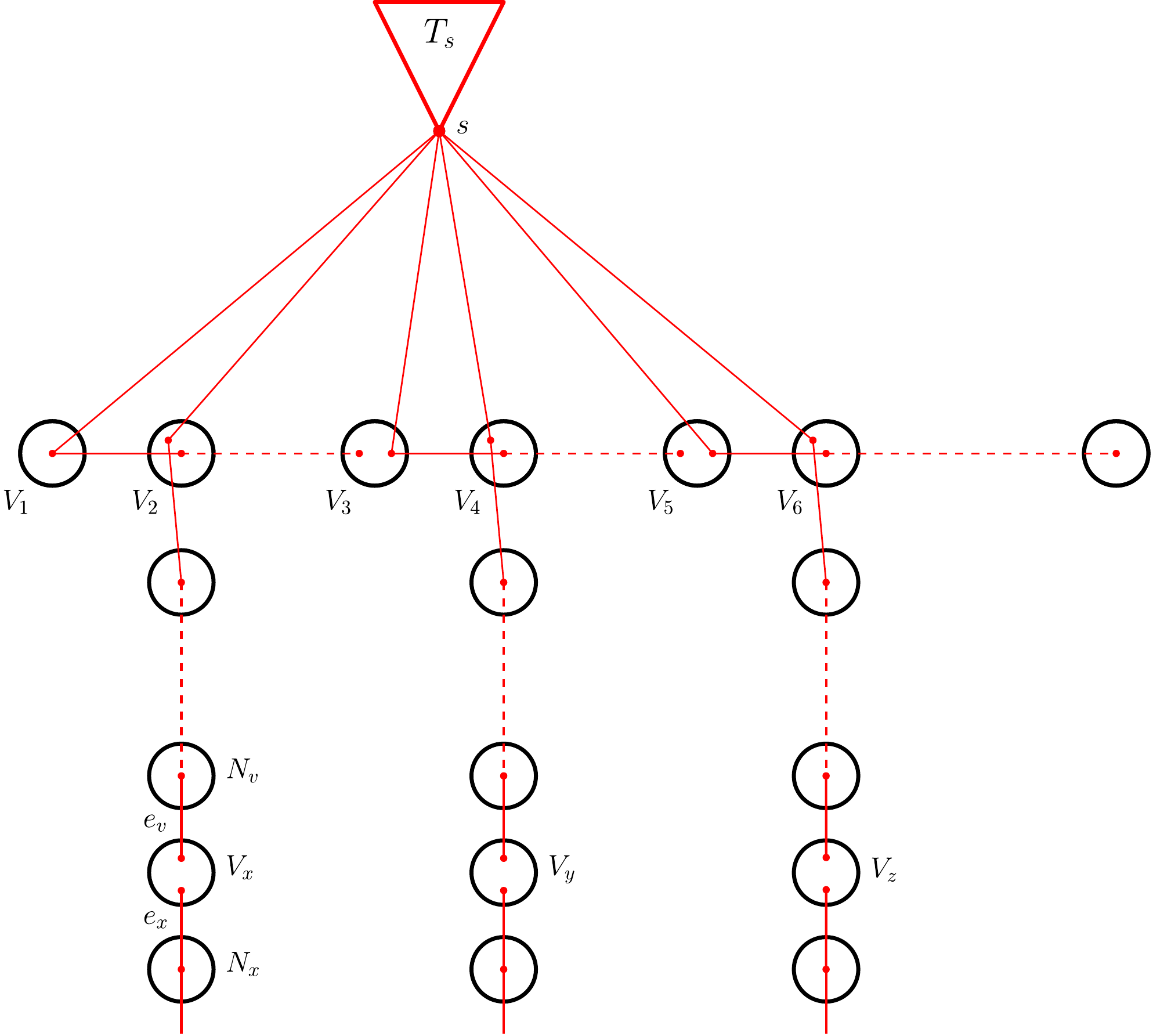}
    \caption{Illustration of the unique way to colourfully embed $T$ into $\hat{G}$. The induced fracture is~$\tau$.}
    \label{fig:E_EmbedTCaseStar}
\end{figure}

We are now able to prove the main lemma of this section.
\begin{lemma}\label{lem:success_stars}
$\oplus\embs{(\fracture{Q}{\tau},c_\tau)}{(G,c)} = \oplus\colsubs{T}{(\hat{G},\hat{\gamma})}$.
\end{lemma}
\begin{proof}
Thanks to Lemma~\ref{lem:unique_fracture_caseStar}, the proof is  similar to the proof of Lemma~\ref{lem:success_Hgadgets}: Colour-preserving embeddings in $\embs{(\fracture{Q}{\tau},c_\tau)}{(G,c)}$ are uniquely identified by their image, and a bijection $b$ from $\colsubs{T}{(\hat{G},\hat{\gamma})}$ to images of colour-preserving embeddings in $\embs{(\fracture{Q}{\tau},c_\tau)}{(G,c)}$ is given by $b:\colT \mapsto \colT[G]$. 
\end{proof}

Similarly to the proof in Section~\ref{sec:caseH},  Lemma~\ref{lem:success_stars}  is sufficient for hardness.

\begin{lemma}\label{lem:TreeHardnessCaseStar}
Let $\mathcal{T}$ be a recursively class of trees of unbounded star number. Then $\oplus\subsprob(\mathcal{T})$ is $\oplus\W{1}$-hard.
\end{lemma}
\begin{proof}
The proof is almost identical to the proof of Lemma~\ref{lem:TreeHardnessCaseH}, with the exception that we use $Q$, $\tau$, $\hat{G}$, and $\hat{\gamma}$ as defined in the current section, and that we rely on Lemma~\ref{lem:success_stars} for the identity
\[\oplus\embs{(\fracture{Q}{\tau},c_\tau)}{(G,c)} = \oplus\colsubs{T}{(\hat{G},\hat{\gamma})}.\]
The remainder of the proof transfers verbatim.
\end{proof}

\subsection{Unbounded Fork number}\label{sec:caseF}
We will rely on the same high-level strategy as the one that we used when the $C$-number or star number was unbounded: Given a tree $T$ with large $a$-$b$-fork number, we start with a properly chosen cubic graph $\Delta$, and we construct a graph $Q$ which depends on $T$ and $\Delta$, and which contains $\Delta$ as a minor. Afterwards, we show that for any $Q$-coloured graph $(G,c)$ we can construct an edge-coloured graph $(\hat{G},\hat{\gamma})$ 
where the co-domain of $\hat{\gamma}$ is $E(T)$ 
such that $\#\colsubs{T}{(\hat{G},\hat{\gamma})}$ is equal (modulo~$2$) to $\#\embs{(\fracture{Q}{\tau},c_\tau)}{(G,c)}$ for a particular fracture $\tau$ of $Q$. However, proving this equality will be more involved than it was in the previous cases: In Sections~\ref{sec:caseH} and~\ref{sec:caseStar}, we were able to prove, implicitly, that $\#\colsubs{T}{(\hat{G},\hat{\gamma})}=\#\embs{(\fracture{Q}{\tau},c_\tau)}{(G,c)}$, that is, we were able to establish equality, rather than equality modulo $2$. In the current case, we are not able to prove equality and must therefore rely on parity arguments, which makes the case slightly more involved. We start by fixing the following:
\begin{itemize}
    \item Positive integers $k$, $a$ and $b$ with $a\leq b$ and $k\geq 2$. 
    \item A tree $T$ with $\fork_{a,b}(T)\geq 2k$. By definition of forks (Definition~\ref{def:fork}), $T$ contains \emph{designated sources} $s_1^1,s_1^2,\dots,s_k^1,s_k^2$ such that for each $(i,j)\in[k]\times [2]$, the source $s_i^j$ is the source of two (distinct) rays $F_a(i,j)$ of length $a$ and $F_b(i,j)$ of length $b$. Additionally $\degnl(s_i^j)=1$. We assume w.l.o.g.\ that the designated sources are ordered by their leaf-degrees, that is
\begin{equation}\label{eq:order_sources}
    \degl(s_1^1)\geq \degl(s_1^2) \geq \dots \geq \degl(s_k^1) \geq \degl(s_k^2)\,.
\end{equation}
Consider Figure~\ref{fig:G} for an illustration of $T$, its designated sources, and the rays $F_a(i,j)$ and $F_b(i,j)$.
\item A $k$-vertex \emph{bipartite} cubic graph $\Delta$ with vertices $V(\Delta)=\{v_1,\dots,v_k\}$.
\item A proper $3$-edge-colouring $C:E(\Delta)\rightarrow \{s,m,\ell\}$ of $\Delta$.\footnote{That is, $C(e_1)\neq C(e_2)$ whenever $e_1\neq e_2$ share a vertex. Note that every cubic bipartite graph has a $3$-edge-colouring by Hall's Theorem.}
\end{itemize}

\begin{figure}
    \centering
    \includegraphics[scale=1]{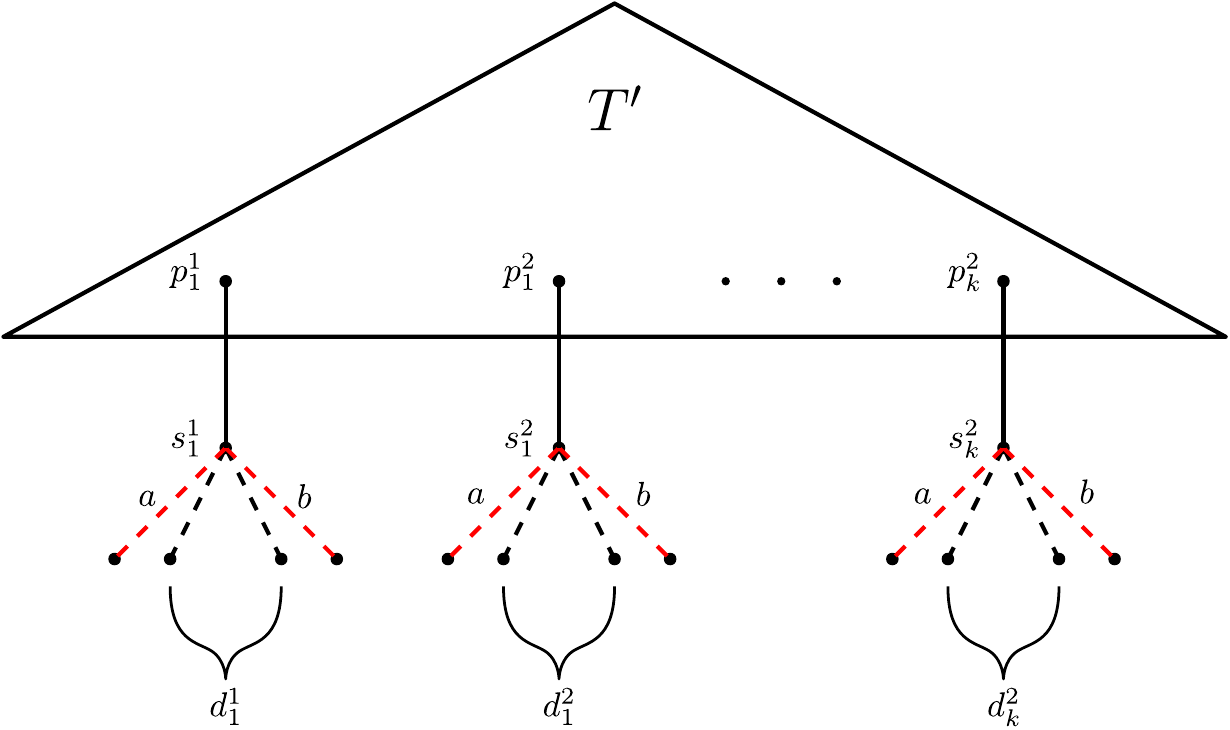}
    \caption{A tree $T$ with $\fork_{a,b}(T)\geq 2k$. Note that the parents of the $s_i^j$ are not necessarily distinct. The rays $F_a(i,j)$ and $F_b(i,j)$ are depicted in red.}
    \label{fig:G}
\end{figure}

We first note that, since there are at least $2k\geq 4$ sources in $T$, any pair of distinct sources must not be adjacent: Otherwise, the tree $T$ would either be disconnected, or one of the sources would have $\degnl$ at least $2$, both of which is a contradiction.
\begin{observation}\label{obs:sources_nonadjacent}
For any distinct pair $(i,j)\neq (i',j')$ we have that $s_i^j$ and $s_{i'}^{j'}$ are not adjacent in $T$.
\end{observation}

Next, we define the graph $Q$.
\begin{definition}[$Q$]\label{def:Q_caseF}\label{aaa}
The graph $Q$ is obtained from~$\Delta$ and~$C$ via substituting $v_i$ by the gadget depicted in Figure~\ref{fig:def_Q_fork} for each $i\in[k]$. Afterwards, for every edge $e=\{v_i,v_j\}$ of $\Delta$ we identify the vertex coloured with $C(e)$ in the gadget of $v_i$ with the vertex coloured with $C(e)$ in the gadget of $v_j$.
\end{definition}

\begin{figure}
    \centering
    \begin{tikzpicture}[scale=0.6]
    \node[circle,draw,fill] (1) at (0,0) {};
    \node[circle] (100) at (1,0) {$v_i^1$};
    \node[circle,draw,fill] (2) at (0,4) {};
    \node[circle] (200) at (1,4) {$v_i^2$};
    \draw[ultra thick,dashed] (1) -- node[right] {$a$} (2);
    \node[circle,draw, fill] (3) at (0,8) {};
    \node[circle] (500) at (0,9) {$\ell$};
    \draw[ultra thick, dashed] (3) --  node[right] {$b$}(2);
    \node[circle,draw, fill] (7) at (-3,-3) {};
    \node[circle] (300) at (-3,-4) {$s$};
    \draw[ultra thick, dashed] (1) --  node[below right] {$a$} (7);
    \node[circle,draw, fill] (9) at (3,-3) {};
    \node[circle] (400) at (3,-4) {$m$};
    \draw[ultra thick, dashed] (1) -- node[below left] {$b$} (9);
    \end{tikzpicture}
    \caption{A vertex gadget in the construction of $Q$ in Definition~\ref{def:Q_caseF}. A dashed line labelled with $a$ (resp.\ $b$) depicts a path of length $a$ (resp.\ $b$).}
    \label{fig:def_Q_fork}
\end{figure}
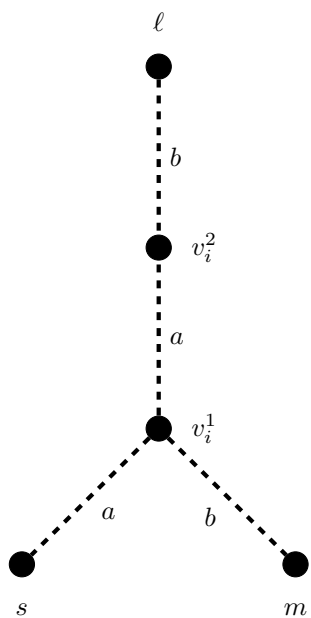

While Definition~\ref{aaa} will be useful in our proofs, we note the following easier equivalent way to define $Q$. 

\begin{observation}
The graph $Q$ is obtained from~$\Delta$ and~$C$ by substituting each edge of colour $s$ (of~$\Delta$) with a path of length $2a$, each edge of colour $m$ with a path of length $2b$, and each edge of colour $\ell$ with a path of length $2(a+b)$.
Consequently, $\Delta$ is a minor of $Q$.
\end{observation}

The fracture $\tau$ of $Q$ that we will be interested in is defined as follows; Figure~\ref{fig:H} depicts the fractured graph $\fracture{Q}{\tau}$.
\begin{definition}[$\tau$]\label{def:tau_caseF}
Let $Q$ be the graph defined in Definition~\ref{def:Q_caseF}.
\begin{itemize}
    \item For each edge $e=\{v_i,v_j\}$ of $\Delta$, there is a vertex $C(e)\in\{s,m,\ell\}$ of degree $2$ that connects the gadgets of $v_i$ and $v_j$. We let $\tau_{C(e)}$ be the partition consisting of two singleton blocks.

    \item For each vertex $v_i$ of $\Delta$, the gadget of $v_i$ in $Q$ contains the vertex $v_i^1$ of degree $3$ which is connected to $s$ via a path of length $a$, to $m$ via a path of length $b$, and to $\ell$ via a path of length $a+b$. Let $e_s$, $e_m$, and $e_\ell$ be the first edges on those paths. We set
    \[ \tau_{v_i} = \{\{e_s,e_m\},\{e_\ell\}\} \,.\]
\end{itemize}
For all other vertices $u$ of $Q$, we let $\tau_u$ be the partition consisting only of one block.
\end{definition}

\begin{figure}
    \centering
    \includegraphics[scale=1]{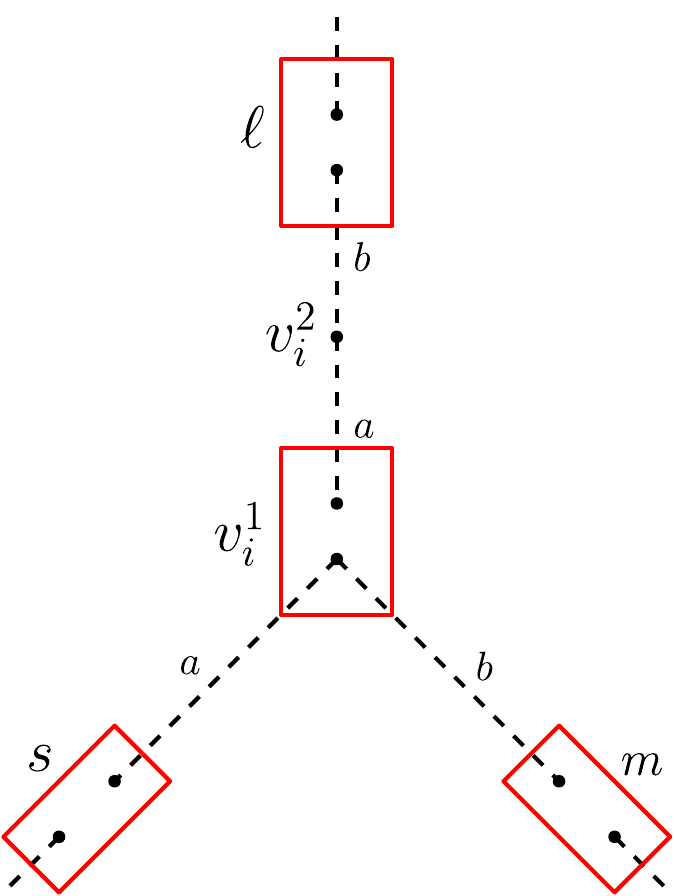}
    \caption{The fractured graph $\fracture{Q}{\tau}$. Note that the illustration only depicts the fracturing of a single vertex gadget.}
    \label{fig:H}
\end{figure}

\noindent Next we identify specific substructures of $T$ that will be necessary in the construction of $\hat{G}$.
\begin{definition}\label{def:substructs_forks}
Recall that $s_i^j$ with $(i,j)\in[k]\times[2]$ are the designated sources of $T$.
\begin{itemize}
    \item  $T'$ is the graph obtained from $T$ by deleting, for each $(i,j)\in[k]\times[2]$, the designated source $s_i^j$ as well as all rays with source $s_i^j$.
    \item For each $(i,j)\in[k]\times[2]$,  $p_i^j$ is the neighbour of $s_i^j$ which is not contained in a ray. Note that $p_i^j$ is unique by definition of forks. Note that $p_i^j\in V(T')$ and that the $p_i^j$ are not necessarily pairwise distinct.
    \item For each $(i,j)\in[k]\times[2]$,  $d_i^j = \degl(s_i^j)-2$, that is, $d_i^j$ is
    the number of rays with source $s_i^j$ minus $2$. Note that $d_i^j\geq 0$ since each $s_i^j$ is the source of $F_a(i,j)$ and $F_b(i,j)$.
    \item 
    $\displaystyle F:= \bigcup_{(i,j)\in[k]\times[2]} \left( F_a(i,j) \cup F_b(i,j) \right),$
    that is, $F$ is the subset of $V(T)$ that contains the vertices of the rays $F_a(i,j)$ and $F_b(i,j)$ (which includes $s_i^j$) for each $(i,j)\in[k]\times[2]$.
    \item  $E':= E(T[F])$.
\end{itemize}
An illustration of these notions is given in Figure~\ref{fig:G}. 
\end{definition}

Observe that $T[F]$ is a disjoint union of $2k$ paths of length $a+b$. Specifically, for each $(i,j)\in[k]\times[2]$ it contains the path
\[F_i^j := T[F_a(i,j) \cup F_b(i,j)] \,.\]

It turns out that $Q$ is isomorphic to a quotient graph of $T[F]$, since for each vertex $v_i$ of $\Delta$, the vertex gadget of $v_i$ decomposes into two paths of length $a+b$. In fact, this decomposition is given by the fractured graph $\fracture{Q}{\tau}$ (see Figure~\ref{fig:H}). Formally, we have the following:
\begin{observation}\label{obs:iso_forks}\label{abq}
$T[F]\cong \fracture{Q}{\tau} \cong 2kP_{a+b}$.
\end{observation}

Similarly to the previous two cases, we introduce functions $\gamma$ and $\gamE$ which we will need for defining the edge-colours of $\hat{G}$.

\begin{definition}[$\gamma,\gamE$]\label{def:gam_gamE_forks}
We define a function $\gamma:F \to V(Q)$ as follows:
\begin{enumerate}
    \item For each $i\in[k]$, $\gamma$ maps $F_i^1$ to the $(a+b)$-path in the gadget of $v_i$ from $s$ to $m$, such that $\gamma(s_i^1)=v_i^1$.
    \item For each $i\in[k]$, $\gamma$ maps $F_i^2$ to the $(a+b)$-path in the gadget of $v_i$ from $v_i^1$ to $\ell$, such that $\gamma(s_i^2)=v_i^2$.
\end{enumerate}
Furthermore, we write $\gamE:E'\to E(Q)$ by setting $\gamE(\{x,y\}):=\{\gamma(x),\gamma(y)\}$.
\end{definition}
Note that the definition of $\gamE$ is well-defined since $\gamma$ is a homomorphism by Observation~\ref{abq}. Concretely, $\gamma$ can be viewed as the composition of an isomorphism from $T[F]$ to $\fracture{Q}{\tau}$ and the $Q$-colouring $c_\tau$ of $\fracture{Q}{\tau}$ (see Definition~\ref{def:fracture_Q_cols}). Furthermore, $\gamE$ is clearly a bijection. Hence, similarly to the previous sections, we point out the following:
\begin{observation}
Let $(G,c)$ be a $Q$-coloured graph. The mapping $c_E \circ \gamE^{-1}$ is a map from $E(G)$ to $E'$.
\end{observation}

We are now able construct a graph $\hat{G}$ from a $Q$-coloured graph $G$; an illustration is provided in Figure~\ref{fig:I}.

\begin{definition}[$(\hat{G},\hat{\gamma})$]\label{def:hatG_caseF}
Let $(G,c)$ be a $Q$-coloured graph. The pair $(\hat{G},\hat{\gamma})$ is an edge-coloured graph constructed as follows, where the co-domain of $\hat{\gamma}$ is $E(T)$.
\begin{itemize}
    \item[(A)] The graph $\hat{G}$ contains $G$ as a subgraph. For each $e\in E(G)$, define $\hat{\gamma}(e)=\gamE^{-1}(c_E(e))$.
    \item[(B)] The vertex set of $\hat{G}$ is the union of $V(G)$ and $V(T)\setminus F$.
    \item[(C)] Pairs of vertices in $V(T)\setminus F$ are connected by an edge in $\hat{G}$ if and only if they are adjacent in $T$. For each such edge $e$, we set $\hat{\gamma}(e)=e$.
    \item[(D)] The remaining edges of $\hat{G}$ are defined as follows.  For each edge $e\in E(T)$ that connects a vertex $z \in V(T)\setminus F$   to a vertex $y\in F$ there are corresponding edges in $\hat{G}$. These edges 
connect~$z$ to  all vertices $g\in V(G)$ such that $c(g)=\gamma(y)$ 
For each such edge $e'$ in~$\hat{G}$,
$\hat{\gamma}(e') = e $. 
\end{itemize}
\end{definition}
In (D), the only edges in $T$ connecting $z \in V(T)\setminus F$ to a vertex $y\in F$ satisfy that $y$ is one of the designated sources $s_i^j$, and $z$ is either $p_i^j\in V(T')$ or $z$ is contained in one of the $d_i^j$ rays with source $s_i^j$ that are not $F_a(i,j)$ or $F_b(i,j)$ (see Definition~\ref{def:substructs_forks}). 

\begin{figure}
    \centering
    \includegraphics[scale=0.8]{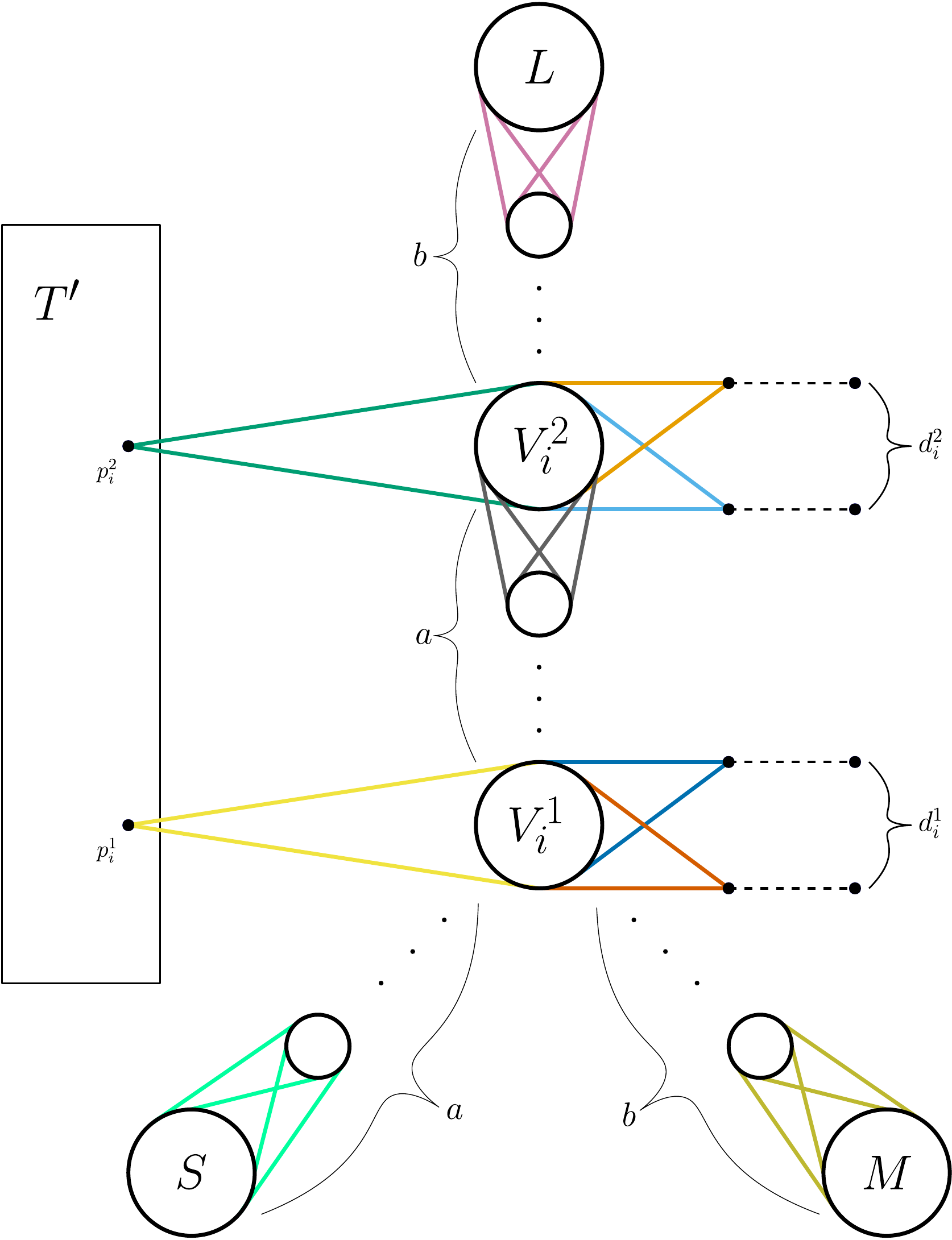}
    \caption{\small The graph $\hat{G}$. Depicted in the centre is the part of $G$ (within $\hat{G}$) that is coloured with the vertices of the $i$-th vertex gadget of $Q$.
    Depicted in black are the subtree $T'$ of $T$ (left), and, as dashed lines, the inner edges of the $d_i^1+d_i^2$ rays incident to $s_i^1$ and $s_i^2$ (right) --- here, the \emph{inner} edges are those that are not incident to the sources $s_i^1$ and $s_i^2$. Each edge of $\hat{G}$ fully contained in the black part has a unique colour w.r.t.\ $\hat{\gamma}$ (see Definition~\ref{def:hatG_caseF} (C)). Pairs consisting of remaining edges have the same colour (w.r.t.\ $\hat{\gamma}$) if and only if they are depicted with the same colour.}
    \label{fig:I}
\end{figure}

Similarly to the  other cases,  for each element $\colT\in\colsubs{T}{(\hat{G},\hat{\gamma})}$ 
the induced subgraph $\colT[G]:=\colT[V(\colT)\cap V(G)]$ of $\colT$ is an edge-colourful subgraph in $G$.
Also, $\colT[G]$ induces a fracture $\rho=\rho(\colT)$ of~$Q$ as follows. First, recall that $G$ is $Q$-coloured by $c$, and that $G$ is contained in $\hat{G}$ (see (A) in Definition~\ref{def:hatG_caseF}). Next note that $\colT[G]$ is a subgraph of $G$ that contains each edge colour in the image of $c_E \circ \gamE^{-1}$ precisely once. Since $\gamE$ is a bijection from $E'$ to $E(Q)$, we can thus equivalently view $\colT[G]$ as a subgraph of $G$ that contains each edge colour in the image of $c_E$ precisely once. This
fact allows us to define $\rho(\calT)$ in terms of the function~$c_E$ as follows. 

\begin{definition}[$\rho(\colT)$] \label{def:rho_caseF}
Let $\colT$ be an element of $ \colsubs{T}{(\hat{G},\hat{\gamma})}$.
The fracture $\rho = \rho(\colT)$ of~$Q$ is defined as follows.
Two edges $\{v,w\}$ and $\{v,y\}$ of $Q$ are in 
the same block 
in the partition $\rho_v$ corresponding to   vertex $v$ of $Q$ 
if and only if the edges of $\colT[G]$
that are coloured by $c_E$ with $\{v,w\}$ and 
$\{v,y\}$ are incident.
\end{definition}

With $(\hat{G},\hat{\gamma})$ defined, we can finally state formally the goal of this section. Recall that $(G,c)$ is a $Q$-coloured graph.
\begin{restatable}{lemma}{mainforks}\label{lem:main_forks}
Suppose that $|c^{-1}(v)|$ is odd for each $v\in V(Q)$. Then $\oplus\colsubs{T}{(\hat{G},\hat{\gamma})} = \oplus \embs{(\fracture{Q}{\tau},c_\tau)}{(G,c)}$.
\end{restatable}

The proof requires some additional set-up. In particular, we need the condition that $|c^{-1}(v)|$ is odd to deal with the case in which what we call ``invalid trees'' arise. To this end, recall that $V_i^j = c^{-1}(v_i^j)$ denotes the set of vertices in $G$ that are coloured by $c$ with $v_i^j$. Since $G$ is a subgraph of $\hat{G}$ (see Definition~\ref{def:hatG_caseF}), we slightly abuse notation and write $V_i^j$ also for the subset of vertices in $\hat{G}$ corresponding to $V_i^j$ in $G$.

\begin{definition}\label{def:invalidTrees}
Let $\colT\in\colsubs{T}{(\hat{G},\hat{\gamma})}$ and let $(i,j)\in[k]\times [2]$. We call $\colT$ \emph{invalid at} $(i,j)$ if the following two conditions are met:
\begin{itemize}
    \item[(I)] $\colT$ contains precisely two vertices $x$ and $y$ in $V_i^j$.
    \item[(II)] $x$ is adjacent to $p_i^j$ and not incident in $\colT$ to any edge coloured with a colour in $E'$ (see Definition~\ref{def:hatG_caseF} (A)). 
\end{itemize}
Otherwise $\colT$ is called \emph{valid} at $(i,j)$. We call $\colT$ an \emph{invalid tree} if there exists a pair $(i,j)\in[k]\times[2]$ such that $\colT$ is invalid at $(i,j)$. Otherwise, we call $\colT$ a \emph{valid tree}. We write $\colsubsval{T}{(\widehat{G},\hat{\gamma})}$ for the set of all valid $\colT$ in $\colsubs{T}{(\widehat{G},\widehat{\gamma})}$.
\end{definition}
Consider Figure~\ref{fig:J} for an illustration of Definition~\ref{def:invalidTrees}.

\begin{figure}
    \centering
    \includegraphics[width=\textwidth]{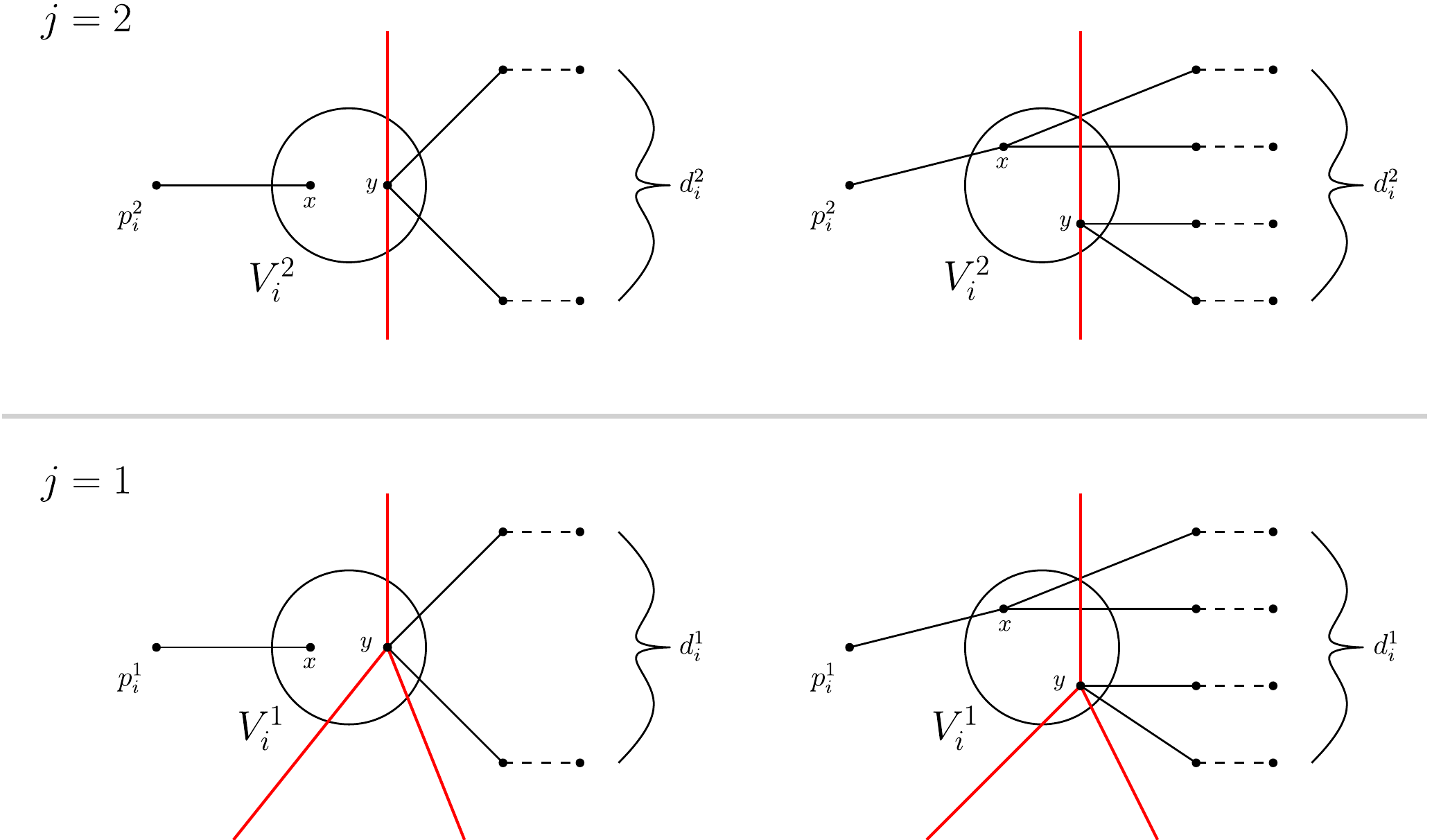}
    \caption{Illustration of the condition that yields invalid trees at $(i,1)$ (below) and $(i,2)$ (above). Edges contained in $E'$ are coloured red.}
    \label{fig:J}
\end{figure}

\begin{lemma}\label{lem:blah}
Suppose that $|c^{-1}(v)|$ is odd for each $v\in V(Q)$.
Then the number of invalid trees  $\colT\in\colsubs{T}{(\hat{G},\hat{\gamma})}$ is even.
\end{lemma}
\begin{proof}
For the proof, given two tuples $(i,j)$ and $(i',j')$ in $[k]\times[2]$ we write $(i',j')<(i,j)$ if $(i',j')$ is lexicographically smaller than $(i,j)$. Write $\mathcal{T}(i,j)$ for the set of all $\colT\in\colsubs{T}{(\hat{G},\hat{\gamma})}$ that are invalid at $(i,j)$ but valid on all pairs $(i',j')<(i,j)$. We will prove that $\mathcal{T}(i,j)$ is even for all $(i,j)\in[k]\times[2]$; this is sufficient for the lemma to hold.

Hence fix $(i,j)$, let $\colT\in\mathcal{T}(i,j)$, and let $x$ and $y$ be as in Definition~\ref{def:invalidTrees}. 
Since $V_i^j = c^{-1}(v_i^j)$ and for $j\in [2]$, $v_i^j$ is a vertex of~$Q$, the assumption in the statement of the lemma implies that $|V_i^j|$ is odd. 
Since $x$ and $y$ are distinct vertices in $V_i^j$, $V_i^j$ contains
additional vertices other than $x$ and $y$.
Fix a vertex $x'\in V_i^j\setminus \{x,y\}$. 
Obtain $\colTprime$ 
from $\colT$ by deleting $x$ (including edges incident to $x$) and by adding $x'$ and the edge $\{x',u\}$ for every $u$ that was adjacent to $x$ --- this is well-defined since $x$ is not incident to any edge coloured with a colour in $E'$, and by construction of $\hat{G}$ (see Definition~\ref{def:hatG_caseF} (C) and (D)) whenever $\{x,u\}\in E(\hat{G})$ is an edge not coloured with a colour in $E'$, then $\{x',u\}\in E(\hat{G})$ for every $x'\in V_i^j$. Additionally, $\{x,u\}$ and $\{x',u\}$ have the same edge-colour.
Hence, clearly, $\colTprime$ an edge-colourful subgraph of $\hat{G}$ that is isomorphic to $\colT$ (and thus to $T$). For this reason, we obtain that $\colTprime\in \mathcal{T}(i,j)$.

More generally, the   observation 
that $\colTprime\in \mathcal{T}(i,j)$
allows us to define an equivalence relation on $\mathcal{T}(i,j)$: Let $\colT$ and $\colTprime$ be elements of $\mathcal{T}(i,j)$, and let $x$ and $x'$ be the vertices in $\colT$ and $\colTprime$ that satisfy (II) in Definition~\ref{def:invalidTrees}. We set $\colT$ and $\colTprime$ to be equivalent if and only if one can obtained from the other by switching $x$ with $x'$ as defined above. The size of one equivalence class is precisely $|V_i^j|-1=|c^{-1}(v_i^j)|-1$, which is even by the premise of the lemma.
\end{proof}

For the proof of Lemma~\ref{lem:main_forks}, we need to establish some facts about rays and $2$-paths of elements $\colT\in \colsubsval{T}{(\hat{G},\hat{\gamma})}$,
which are
those    $\colT \in \colsubs{T}{(\hat{G},\hat{\gamma})}$ that are valid.
We encapsulate these facts in the next section.

\subsubsection{The Proof of Lemma~\ref{lem:main_forks}}
We first note that, thanks to Lemma~\ref{lem:blah}, it suffices to prove that
\[\#\colsubsval{T}{(\hat{G},\hat{\gamma})} = \#\embs{(\fracture{Q}{\tau},c_\tau)}{(G,c)} \,.\]
This requires some preparation. We
first fix the following objects (recall the definitions of $2$-path, Definition~\ref{def:2path} and ray, Definition~\ref{def:ray}).
\begin{itemize}
    \item $\colT$ is an element of $\colsubsval{T}{(\hat{G},\hat{\gamma})}$
    \item $\colT[G]$ is the graph obtained from $\colT[V(\colT)\cap V(G)]$ with isolated vertices removed. (In fact, our proof will show that, for valid trees $\colT\in \colsubsval{T}{(\hat{G},\hat{\gamma})}$, the induced subgraph $\colT[V(\colT)\cap V(G)]$ cannot have isolated vertices. However, at the current point of the proof, it is easiest to just remove them.)
\item For any positive integer $t$,   $\calR^t$ is the set of length-$t$ rays   in $T$.    $\calP^t$  is the set of length-$t$ $2$-paths in $T$ that  are not rays. 
\item  For any positive integer $t$,   $\colR^t$ is the set of length-$t$ rays in $\colT$ and $\colP^t$ is the set of $2$-paths in $\colT$ that are not rays.  Note that $|\calR^t|=|\colR^t|$ and $|\calP^t|=|\colP^t|$ for all $t$ since $T$ and $\colT$ are isomorphic.
\end{itemize}

We will also rely on the following notion of external rays and $2$-paths.
\begin{definition}\label{def:hi}
A $2$-path $P$   of $\colT$ is called \emph{external} if the following two conditions are satisfied.  
\begin{itemize}
    \item Except for the endpoints, none of the vertices of $P$ is contained in $V(G)$. 
    \item $P$ does not contain an edge of $G$.  
\end{itemize} 
\end{definition}

Definition~\ref{def:hi} applies whether or not $P$ is a ray.
The following lemmas establish that all $2$-paths of $\colT$ of length greater than $b$ must be external.

\begin{lemma}\label{lem:anotherlongRays}
Suppose that $t$ is an integer that is greater than $b$. Suppose that,
for all $t'>t$, every $2$-path in $\colR^{t'} \cup \colP^{t'}$ is external. Then every $2$-path in $\colR^{t} \cup \colP^{t}$ is external. 
\end{lemma}
\begin{proof}
We first construct a bijection $f$ from $\calR^t$ to $\colR^t$. 
We will use this bijection to argue that every ray in $\colR^t$ is external.
In order to define the bijection,  consider a ray $R=r_0,r_1,\dots,r_t$  in $\calR^t$. Since $t>b\geq a$, $R$ is not one of the designated rays $F_a(i,j)$ and $F_b(i,j)$. If $r_0$ is not among the designated sources $s_i^j$, then, by the construction of~$\hat{G}$, $R$ is contained in $T'$ and thus 
$R\in \colR^t$. In this case $R$ is external and we set $f(R):=R$. Alternatively, suppose that $r_0 = s_i^j$ for some $i$ and $j$. Then $R$ must be one of the $d_i^j$ black rays in Figure~\ref{fig:G} (see Definition~\ref{def:substructs_forks}). 
By the construction of $\hat{G}$ and the fact that $\colT$ is edge-colourful, 
there is a vertex $x\in V_i^j$ such that $\colT$ contains the path $x,r_1,\ldots,r_t$.
In $\colT$, as in~$T$, the vertices $r_1,\ldots,r_{t-1}$ have degree~$2$ and the vertex~$r_t$ has degree~$1$.
Vertex $x$ cannot have degree $1$ in $\colT$ since this would disconnect $\colT$. 
Also, vertex $x$  $x$ cannot have degree $2$: To see this, assume for contradiction that $x$ has degree $2$. Then there is an integer $t'>t$ and a ray $R'\in \colR^{t'}$ the last vertices of which are $x,r_1,\dots,r_t$. Since $x$ is not an endpoint of the ray and since $x\in V(G)$, the ray $R'$ is not external, contradicting the premise of the lemma. Hence $x$ has degree at least $3$ and therefore $f(R):=x,r_1,\dots,r_t$ is an external ray of $\colT$.  The function~$f$ is   injective by construction. Since $\colT$ and $T$ are isomorphic,   $|\calR^t|=|\colR^t|$ and thus $f$~is a bijection. Since the image of~$f$ only contains external rays, we have shown that every element of~$\colR^t$ is external. 

Every ray in the image of $f$ has the property that its   degree-$1$ endpoint is not contained in $V(G)$. Since the image of $f$ is $\colR^t$, we obtain

\medskip

\noindent $(\ast)$ Every ray in $\colR^t$ has the property that its degree-1 endpoint is not contained in $V(G)$.

\medskip

To complete the proof, we show that every $2$-path in $\colP^t$ is external.  Following the same strategy that we used before, we construct a bijection $g$ from  $\calP^t$ to $\colP^t$. Every 2-path in the range of~$g$ is external, so we will conclude that every element of $\calP^t$ is external. In order to define the bijection, consider a $2$-path  $P=p_0,\dots,p_t$  in $\calP^t$. If neither of the endpoints of~$P$ is
among the designated sources $s_i^j$, then $P$ is contained in $T'$ and thus $P\in \calP^t$. In this case, $P$ is external and we set   $g(P):=P$. If exactly one endpoint of~$P$ is among the designated sources, say $p_0=s_i^j$, then there  is a vertex $x\in V_i^j$ such that $x,p_1,\dots,p_t$ is a path in $\colT$.
The vertices $p_1,\ldots, p_{t-1}$ have degree~$2$ in $\colT$ (as in~$T$) and the vertex $p_t$ has degree at least~$3$.

If $x$ has degree $1$ in~$\colT$, the ray $R=p_t,\dots,p_1,x$ is in $\colT$, and 
its degree-1 endpoint~$x$ is   in $V(G)$, contradicting $(\ast)$. Hence $x$ cannot have degree~$1$ in~$\colT$. Similarly, $x$ cannot have degree $2$, since this would create a 2-path longer than $t$ in $\colT$ that is not external, which contradicts   the premise of the lemma. Hence $x$ has degree at least $3$, and thus $g(P):=x,p_1,\dots,p_t$ is an external $2$-path in $\colP^t$.

For the last case, suppose that both endpoints of~$P$ are among the designated sources, say $p_0=s_i^j$ and $p_t=s_{i'}^{j'}$. Then there are $x$ and $y$ in, respectively, $V_i^j$ and $V_{i'}^{j'}$ such that $x,p_1,\dots,p_{t-1},y$ is a path in $\colT$. Again, $p_1,\dots,p_{t-1}$ must all have degree $2$ in $\colT$ as well. We show that both $x$ and $y$ have degree at least $3$ in $\colT$: If both have degree $1$, then $\colT$  is disconnected. If one of them has degree $1$ and the other one has degree at least $3$, then we created a ray of length $t$ whose degree-1 endpoint in in $V(G)$, contradicting $(\ast)$. If one has degree 
$1$ and the other one has degree $2$, then we found a ray longer than $t$ which is not external, contradicting the premise of the lemma. If one has degree~$2$ and the other has degree at least~$2$, then there is a non-external 2-path longer than $t$, again contradicting the premise of the lemma.
Thus, as desired, both must have degree at least $3$. Therefore, $g(P):=x,p_1,\dots,p_{t-1},y$ is an external $2$-path in $\colP^t$.  The function~$g$ is  injective by construction. Since $\colT$ and $T$ are isomorphic,   $|\calP^t|=|\colP^t|$ and thus $g$ is a bijection. Since the image of $g$ only contains external $2$-paths, we have shown that every element of $\colP^t$  is external, concluding the proof.
\end{proof}
 
\begin{lemma}\label{lem:longRays}
Suppose that $t$ is an integer that is greater than $b$. Then
 every $2$-path in $\colR^{t} \cup \colP^{t}$ is external.
\end{lemma}
\begin{proof}

Let $\tmax$ be the maximum integer for which $\calR^\tmax \cup \calP^\tmax$ is nonempty. 
Let $\Phi_t$ be the proposition  
``$t\leq b$ or
every $2$-path in $\colR^{t} \cup \colP^{t}$ is external''.

We will show by induction on $\tmax - t$ that $\Phi_t$ holds.
The base case arises when $\tmax-t=0$, so $t=\tmax$. 
If  $\tmax \leq b$ then  $\Phi_{t}$ is satisfied.
Otherwise, for each $t'>t$, the set $\colR^{t'} \cup \colP^{t'}$ is empty and we can invoke Lemma~\ref{lem:anotherlongRays} to conclude  that $\Phi_{t}$ holds. 

For the induction step,  consider $t$ such that   $\tmax - t\geq 1$.
By the induction hypothesis, $\Phi_{t'}$ holds for all $t'\in \{t+1,\ldots, \tmax\}$.
If $t\leq b$ then $\Phi_t$ holds. Otherwise, 
for all $t' > t > b$, we know from $\Phi_{t'}$ that every 2-path in 
$\colR^{t'} \cup \colP^{t'}$ is external. We can then apply  Lemma~\ref{lem:anotherlongRays} to conclude that every $2$-path in $\colR^{t} \cup \colP^{t}$ is external.

\end{proof}

Before proceeding with the proof of Lemma~\ref{lem:main_forks}, we provide an overview of the central steps of the proof. Recall that
it suffices to prove that
\[\#\colsubsval{T}{(\hat{G},\hat{\gamma})} = \#\embs{(\fracture{Q}{\tau},c_\tau)}{(G,c)}\]
and that we have a fixed an element
$\colT$   of $\colsubsval{T}{(\hat{G},\hat{\gamma})}$ and proved various properties about it.

\begin{enumerate}[(1)]
\item \label{firstgoal}
Our goal is to show that $\colT$ is embedded in $\hat{G}$ in the following manner (see Figure~\ref{fig:K}). For each $(i,j)\in[k]\times[2]$, $\colT$ contains a ray $R_a(i,j)$ of length $a$ and a ray $R_b(i,j)$ of length $b$; those rays correspond to the designated rays $F_a(i,j)$ and $F_b(i,j)$ in $T$ (recall that $T$ and $\colT$ are isomorphic.)
\begin{enumerate}
\item \label{bul1}
$T'$ is part of $\colT$. 
\item \label{bul2} For every $i\in [k]$ and $j\in[2]$,
the  vertices $p_i^j$ in $T'$  is connected to a vertex $w_i^j$ of~$G$ with $c(w_i^j) = v_i^j = \gamma(s_i^j)$.  In $\colT$, the vertex $w_i^j$ is the source of $d_i^j$ rays 
other than $R_a(i,j)$ and $R_b(i,j)$. The vertices of these $d_i^j$ rays
  are not in~$T'$ and are not in~$G$. The edge colours of the edges in these rays in $\hat{\gamma}$   are the same as the edge-names in~$T$ (see Definition~\ref{def:hatG_caseF} (C)).
\item  \label{bul3}
The length-$a$ ray  $R_a(i,1)$ is a path in $\colT$ from $w_i^1$ to the vertex~$u_a(i,1)$ of~$G$
 with some colour~$c(u_a(i,1))$ (a vertex of $Q$).
 This colour $c(u_a(i,1))$ corresponds to the vertex ``$s$'' in the gadget of the vertex $v_i$ of $\Delta$ (see Definition~\ref{def:Q_caseF} and Figure~\ref{fig:def_Q_fork}). 
 \item \label{bul4} The length-$b$ ray $R_b(i,1)$ is a path in $\colT$ from $w_i^1$ to the vertex~$u_b(i,1)$ of~$G$
 with some colour~$c(u_b(i,1))$ (a vertex of $Q$).
 This colour $c(u_b(i,1))$ corresponds to the vertex ``$m$'' in the gadget of the vertex $v_i$ of $\Delta$ (see Definition~\ref{def:Q_caseF} and Figure~\ref{fig:def_Q_fork}). 
 
  \item \label{bul5} The length-$b$ ray $R_b(i,2)$ is a path in $\colT$ from $w_i^2$ to the vertex~$u_b(i,2)$ of~$G$
 with some colour~$c(u_b(i,2))$ (a vertex of $Q$).
 This colour $c(u_b(i,2))$ corresponds to the vertex ``$\ell$'' in the gadget of the vertex $v_i$ of $\Delta$ (see Definition~\ref{def:Q_caseF} and Figure~\ref{fig:def_Q_fork}).

 \item \label{bul6} The length-$a$ ray $R_a(i,2)$ is a path in $\colT$ from $w_i^2$ to the vertex~$u_a(i,2)\neq w_i^1$ of~$G$
 with some colour~$c(u_a(i,2))=\gamma(s_i^1)=v_i^1$ (recall that the colour is a vertex of $Q$).

  \item \label{bul7}
  For every edge $e=\{v_i, v_{i'}\}$ in~$\Delta$, $u_a(i,1) \neq u_a(i',1)$, $u_b(i,1) \neq u_b(i',1)$ and $u_b(i,2) \neq u_b(i',2)$.

\end{enumerate}

\item \label{secondgoal}
We now make some observations about the fracture $\rho= \rho(\colT)$ from Definition~\ref{def:rho_caseF}, given that $\colT$ is embedded in $\hat{G}$ as described in Item~\eqref{firstgoal}.

\begin{itemize}
\item The definition of $Q$ (Definition~\ref{def:Q_caseF}) tells us that, for every edge $e=\{v_i, v_{i'}\}$ in~$\Delta$, 
there is a degree-$2$ vertex $y$ of $Q$ that connects the gadgets of $v_i$ and $v_{i'}$.  
Vertex $y$ corresponds to the vertex
 $C(e)\in\{s,m,\ell\}$ in the two gadgets. Suppose without loss of generality that $C(e)=s$.
 The other cases are similar.
From~\eqref{bul3} the colour $C(e)=s$ is the same as $c(u_a(i,1))$ and $c(u_a(i',1))$. From~\eqref{bul2}  $c(w_i^1)=v_i^1$ and $c(w_{i'}^1)=v_{i'}^1$.
Since  $\colT$ is colourful and the embedding is as in \eqref{firstgoal}, the edges of the ray from $w_i^1$ to $u_a(i,1)$   have different edge colours to the ray from $w_{i'}^1$ to $u_a(i',1)$.
Thus, the edge in $G$ in the first ray that is adjacent to $u_a(i,1)$ (call it $e_i$) has a different colour from the edge n $G$ in the second ray that is adjacent to $u_a(i',1)$ (call it $e_{i'}$). Concretely, we have $c_E(e_i)=\{s,x\}$ and $c_E(e_{i'})=\{s,x'\}$ where $x$ and $x'$ are the neighbours of $s$ (in $Q$) in the gadgets of $v_i$ and $v_{i'}$, respectively. By~\eqref{bul7} we have $u_a(i,1) \neq u_a(i',1)$ and thus, by definition of $\rho$ (Definition~\ref{def:rho_caseF}), $\rho_y$ consists of two singleton blocks.
Similar arguments show that $\rho$ coincides with $\tau$ (see Definition~\ref{def:tau_caseF})
at every vertex of $Q$ that corresponds to vertex ``$s$'', ``$\ell$'' or ``$m$'' 
in any gadget corresponding to any vertex~$v_i$ of~$\Delta$.
\item We now continue with the vertices $v_i^1$ for $i\in[k]$ of $Q$. 
See Figure~\ref{fig:def_Q_fork} for the gadget containing $v_i^1$ in $Q$ and Figure~\ref{fig:K} for the graph $\hat{G}$. We will use ``$s$'', ``$\ell$'' and ``$m$''
as the names of these vertices in the gadget containing~$v_i^1$.
The vertex~$v_i^1$ has degree $3$ and is connected to $s$ via a path of length $a$, to $m$ via a path of length $b$ and to $\ell$ via a path of length $a+b$. Let $y_s$, $y_m$, and $y_\ell$ be the successors of $v_i^1$ on those paths, that is, the   edges incident to~$v_i^1$ in~$Q$ are $e_s:=\{v_i^1,y_s\}$, $e_m:=\{v_i^1,y_m\}$, an $e_\ell:=\{v_i^1,y_\ell\}$. Now, by \eqref{bul3} and \eqref{bul4}, the edges of $\colT$ that are coloured (by $c_E$) with $e_s$ and $e_m$ are $\{w_i^1,r_a\}$ and $\{w_i^1,r_b\}$, where $r_a$ and $r_b$ are the successors of $w_i^1$ on the rays $R_a(i,1)$ and $R_b(i,1)$, respectively. Furthermore, by \eqref{bul6}, the edge of $\colT$ that is coloured (by $c_E$) with $e_\ell$ is $\{u_a(i,2),\hat{r}\}$ where $\hat{r}$ is 
the vertex in the ray $R_a(i,2)$ that is adjacent to $u_a(i,2)$. Since $u_a(i,2)\neq w_i^1$ (by \eqref{bul6}), 
the edge $\{u_a(i,2),\hat{r}\}$ is not incident to either 
 $\{w_i^1,r_a\}$  or $\{w_i^1,r_b\}$. Thus $\rho_{v_i^1}=\{\{e_s,e_m\},\{e_\ell\}\}$ which coincides with $\tau_{v_i^1}$ 
 by Definition~\ref{def:tau_caseF}.
 So $\tau$ and $\rho$ coincide at vertex $v_i^1$. 
\item Next are the vertices $v_i^2$ for $i\in[k]$ (see Figure~\ref{fig:def_Q_fork}).  This case is easy. If $\colT$ is embedded as described in \eqref{firstgoal} (see Figure~\ref{fig:K}), then, for each $i\in[k]$, there is only one vertex of $\colT$ 
which is coloured by~$c$ with colour~$v_i^2$. This vertex is $w_i^2$. Thus every edge of $\colT$ whose edge colour includes $v_i^2$ is  incident to $w_i^2$. Hence $\rho_{v_i^2}$ only consists of one block, which coincides with $\tau_{v_i^2}$ by Definition~\ref{def:tau_caseF}.
\item Finally, every remaining vertex of $Q$ 
(see Figure~\ref{fig:def_Q_fork})
has degree $2$. Let $y$ be such a vertex and let $y_1$ and $y_2$ be the neighbours of $y$. Then the edges of $\colT$ coloured by $c_E$ with $\{y,y_1\}$ and $\{y,y_2\}$ must be successive edges on one of the rays $R_a(i,1)$, $R_b(i,1)$, $R_a(i,2)$, or $R_b(i,2)$.  So
these successive edges are both incident to the vertex of the ray that is coloured~$y$ by~$c$.   Thus $\rho_y$ only consists of one block, which coincides with $\tau_{y}$.
\end{itemize}
Since we have shown
that the fractures~$\rho$ and~$\tau$ coincide at every vertex of~$Q$, we conclude that
$\rho=\tau$.
 
 \item We next explain why it is useful to have $\rho=\tau$.
 Recall that our goal is
 to prove that
$\#\colsubsval{T}{(\hat{G},\hat{\gamma})} = \#\embs{(\fracture{Q}{\tau},c_\tau)}{(G,c)}$
and that  
$\colT$ is an element of $\colsubsval{T}{(\hat{G},\hat{\gamma})}$.
Our method will be to   show that the function $\beta$
defined by $\beta(\colT) = \colT[G]$  is a bijection from  $\colsubsval{T}{(\hat{G},\hat{\gamma})}$ to $\embs{(\fracture{Q}{\tau},c_\tau)}{(G,c)}$. 
It will turn out that this implies that the
  embedding $\rho$ coincides with $\tau$.

\item In order to prove Item~\eqref{firstgoal} we will proceed as follows.

\begin{enumerate}[(i)]
\item  We show that all $2$-paths (including rays) of $\colT$  are external, \emph{except} for $2k$ rays of length $b$ and $2k$ rays of length $a$. Note that we already established this claim for $2$-paths of lengths greater than $b$ in Lemma~\ref{lem:longRays}.
\item  \label{ss} Then we show that $\colT$   contains two degree-$1$ vertices  in each of the vertex sets $L$ and $M$ of $G$ (within $\hat{G}$) --- see Figure~\ref{fig:K}, recalling
  that, for each vertex gadget, the sets $L$ and $M$ denote the vertex subsets of $G$ that are coloured by $c$ with $\ell$ and $m$. 
The point of this is that we will also prove that $\colT$ has two degree-$1$ vertices in~$S$  (Item \ref{blahblah}) --- this will split off the part of $\colT$ corresponding to a single gadget, so we will only have to 
study the embedding of $\colT$ within each gadget. 
We  prove the claim about $L$ and $M$ by using the fact that $\colT$ is isomorphic to $T$ and that all $2$-paths longer than $b$  are external.
This implies that 
if $v_i$ and $v_{i'}$ are the two vertices of~$\Delta$ sharing this gadget then
the 2-paths between $V_i^2$ and $V_{i'}^2$ are covered by two rays in $\colT$, both of which end in $L$.

\item \label{zzz} We  next show that the  degree-$1$ vertices in \eqref{ss} are the endpoints of $2k$ rays of length $b$.
We have already seen that for each of the $k$ gadgets the endpoints of these rays are in $L$ and $M$.  For the $i$'th gadget, the sources are in    $V^1_i$ and $V^2_i$   
If $b>a$ then we show  that all remaining $2$-paths of length $b$ and also  all $2$-paths 
with lengths in $a+1,\ldots,b-1$ are external. The proof of this claim  relies on the same arguments as the proof of Lemma~\ref{lem:longRays}.
\item  
\label{blahblah}  
 Next, we  show that 
 for each gadget,
 $\colT$  contains two degree-$1$ vertices   in~$S$  --- see Figure~\ref{fig:K}. 
 The proof uses the fact that   all $2$-paths longer than $a$ that are not covered by~\eqref{zzz} are external.  
\item \label{pointfive}

We next show that 
the   degree-$1$ vertices in 
\eqref{blahblah}  are the endpoints of   $2k$ rays of length $a$. We have already seen that for each of the $k$ gadgets the endpoints of these rays are in $S$. For the $i$'th gadget, the source  is in  
  $V^1_i$.
\item 
The remaining details of the proof rely  on the fact that the tree $\colT$ is valid.

\end{enumerate}
\end{enumerate}

We now provide the proof in detail; for convenience, we also restate the lemma.

\mainforks*
\begin{proof}
We will prove that for any $\colT \in \colsubsval{T}{(\hat{G},\hat{\gamma})}$,
Item~\eqref{firstgoal} of the proof overview holds.

Using this fact and the argument from Item~\eqref{secondgoal} of the proof overview, we conclude that
for any $\colT \in \colsubsval{T}{(\hat{G},\hat{\gamma})}$,
$\rho(\colT)=\tau$.

Recall that every edge-colourful subgraph of $G$ induces a fracture of~$Q$. 

Let $G'$ be an element of $\embs{(\fracture{Q}{\tau},c_\tau)}{(G,c)}$. This means that $G'$ is
an edge-colourful subgraph of $G$ that induces~$\tau$.
We wish to see how  $G'$ can be extended to some $\colT' \in\colsubsval{T}{(\hat{G},\hat{\gamma})}$.
We know from Item~\eqref{firstgoal} that any $\colT'' \in \colsubsval{T}{(\hat{G},\hat{\gamma})}$
can only be embedded in $\hat{G}$ in one way, so $G'$ can only be extended in one way.
The details are as follows.
We claim that there is only one possible extension because $T'$ has to be included and item (b) of \eqref{firstgoal} ensures that, for each $j\in[2]$, the vertex $p_i^j$ is connected to $w_i^j$. The rest of \eqref{firstgoal} shows the unique way to include the rays, so the extension is unique.
   
Let $\beta$ be the function from $\colsubsval{T}{(\hat{G},\hat{\gamma})}$
that maps any element $\colT$ to $\colT[G]$. 
Note that $\colT[G] \in \embs{(\fracture{Q}{\tau},c_\tau)}{(G,c)}$ since $\rho(\colT)=\tau$ and $\rho(\colT)$ is a function of $\colT[G]$.
Let
$\beta'$ be the function that maps an element of $\embs{(\fracture{Q}{\tau},c_\tau)}{(G,c)}$
to its unique extension in $\colsubsval{T}{(\hat{G},\hat{\gamma})}$.  Note that $\beta \circ \beta'$ and $\beta' \circ \beta$ are both the identity. Therefore $\beta$ is a bijection and
$| \colsubsval{T}{(\hat{G},\hat{\gamma})}| = |\embs{(\fracture{Q}{\tau},c_\tau)}{(G,c)}|$.The lemma follows from Lemma~\ref{lem:blah}.
 
 \begin{figure}
    \centering
    \includegraphics[scale=0.8]{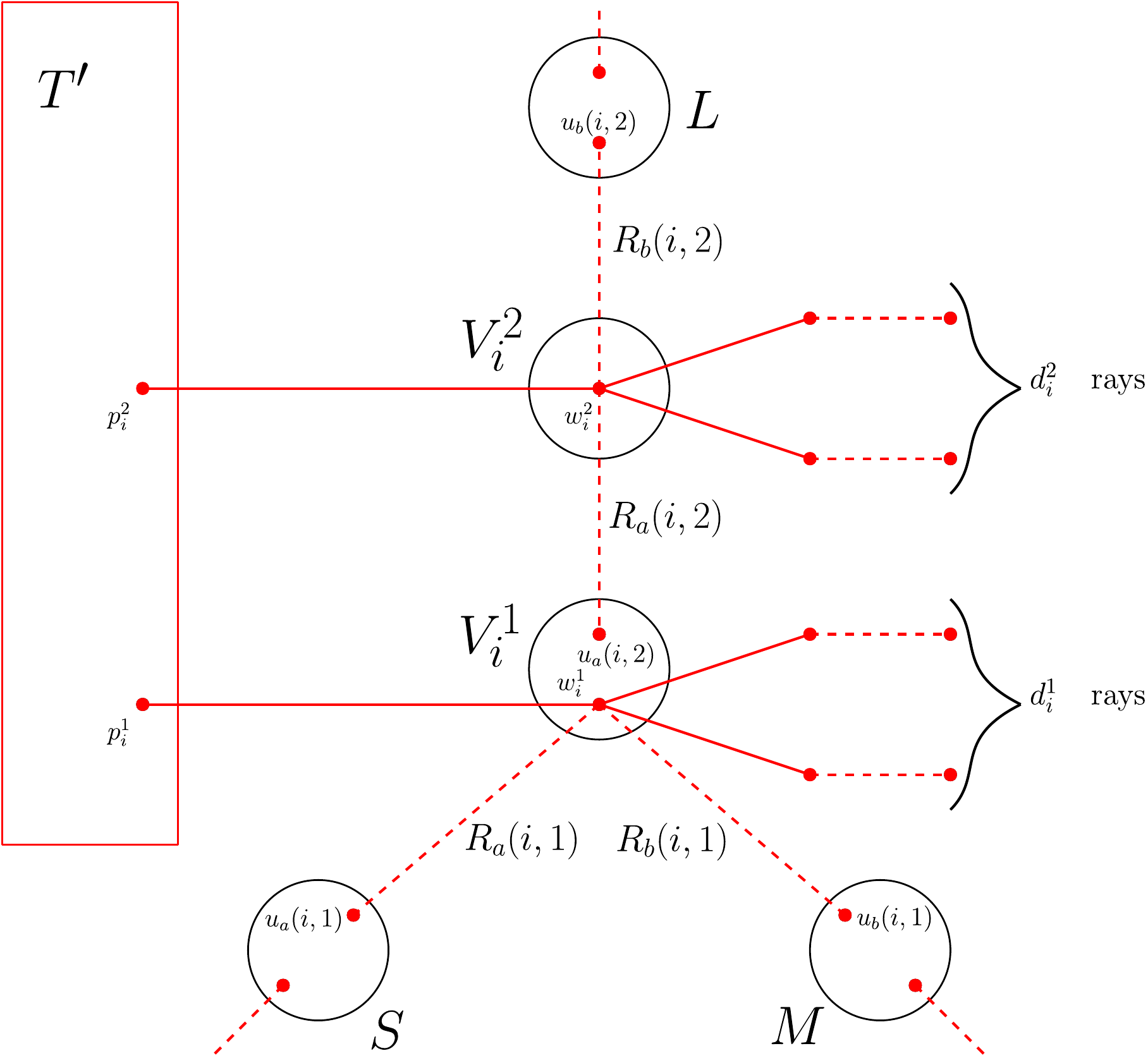}
    \caption{An embedding $\colT$ of $T$ in $\hat{G}$ that yields the fracture $\tau$. We will show that this is the only way to embed $T$ in $\hat{G}$ in such a way that each edge-colour is used precisely once. Note that dashed lines depict paths in $\colT$, and solid lines depict edges in $\colT$.}
    \label{fig:K}
\end{figure}

To finish the proof, we will 
fix $\colT \in \colsubsval{T}{(\hat{G},\hat{\gamma})}$ and we will show that
Item~\eqref{firstgoal} of the proof overview holds.  
Part (a) of \eqref{firstgoal} is trivial since  $\colT$ is edge-colourful so it contains $T'$.  
The first sentence of (b) is also trivial.
We will next focus on (c)--(g), noting along the way when the rest of (b) is proved.

Recall from Definition~\ref{def:Q_caseF} that, for each $i\in[k]$, the graph $Q$ contains 
\begin{itemize}
\item for each vertex $v_j$ such that $\Delta$ has an edge $e=\{v_i,v_j\}$ with $C(e)=m$,
a path $P_{i,j}$   of length $2b$ from $v_i^1$ to $v_j^1$, and
\item for each vertex $v_j$ such that $\Delta$ has an edge $e=\{v_i,v_j\}$ with $C(e) = \ell$,
a path $P_{i,j}$ of length $2b$ from $v_i^2$ to $v_{j}^2$.
\end{itemize}

Recall from Definition~\ref{def:cE} that $c_E$ maps edges of~$G$ to edges of~$Q$.
Furthermore, $G$ is a subgraph of~$\hat{G}$, see  Definition~\ref{def:hatG_caseF} (A).
Let  $\colT(i,j)$  be the subgraph of $\colT[G]$ induced by edges~$e$ of~$G$ such that $c_E(e)$ is in the path $P_{i,j}$   

By construction, $\colT(i,j)$ is the union of some number of paths.
We will next argue that it is the union of exactly two disjoint length-$b$ paths:
\begin{itemize}
\item If $\colT(i,j)$ has more than two components then at least one component is disconnected from $T'$ in $\colT$, contradicting the fact that $\colT$ is a tree. 
\item If $\colT(i,j)$ is a single path then it is contained in a 2-path of length at least $2b$. Since this $2$-path contains an edge in~$G$, it is not external (Definition~\ref{def:hi}). This contradicts Lemma~\ref{lem:longRays}. 
\item If $\colT(i,j)$ is the union of exactly two disjoint paths, one of which has length larger than~$b$
then this larger 2-path is contained in a 2-path that is not external contradicting Lemma~\ref{lem:longRays} 
\end{itemize}

\begin{figure}
    \centering
    \includegraphics[scale=0.85]{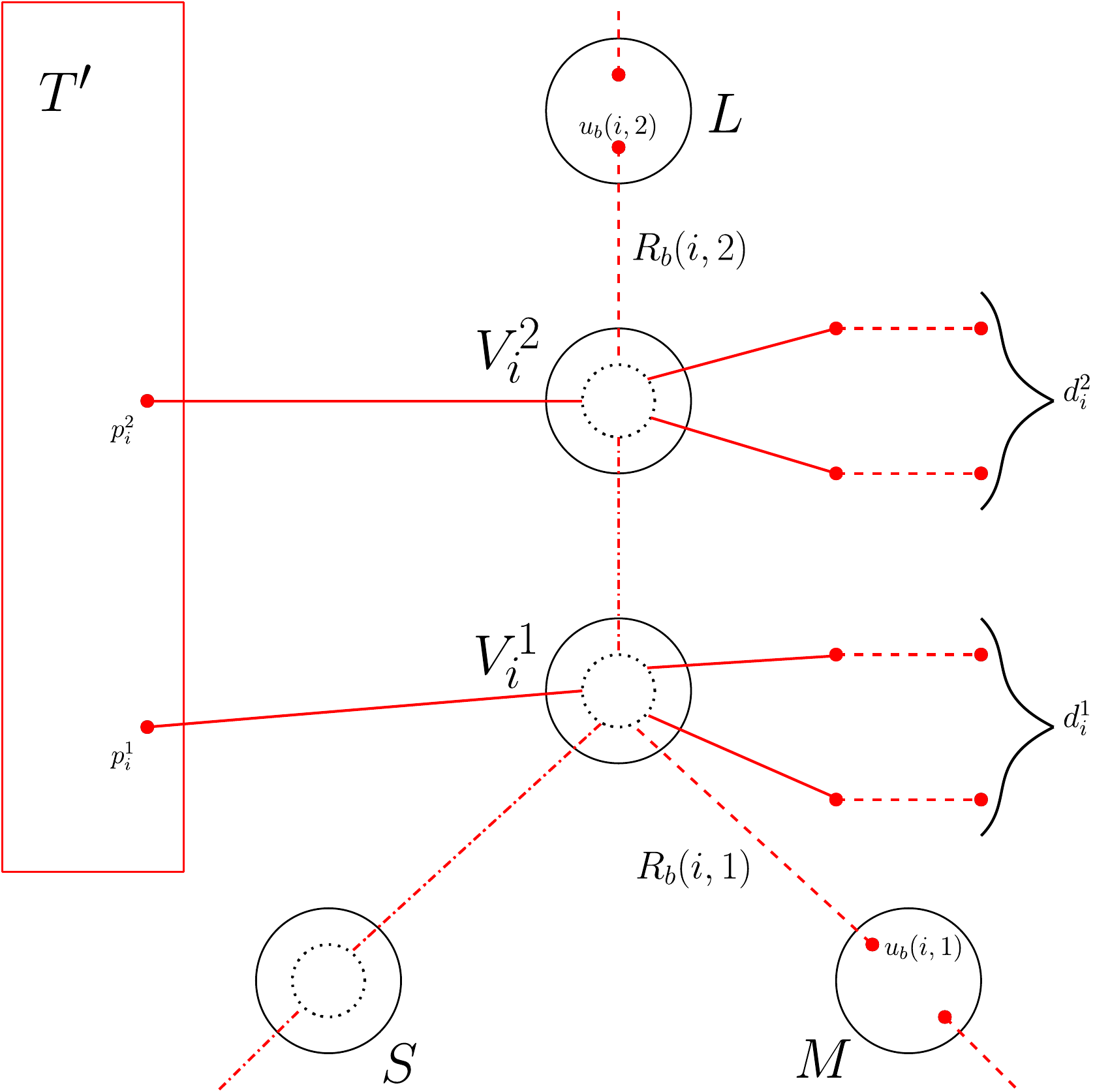}
    \caption{Illustration of the embedding of $\colT$ after the rays of length $b$ are analysed. Solid lines depict edges, dashed lines depict paths, and dash-dotted lines depict sequences of edges (the identification of the endpoints of which we have not yet been determined). Note that both $R_b(i,1)$ and $R_b(i,2)$ must be of length $b$. Except for those two rays, the identification of endpoints of the remaining edges that are incident to $G$ (within $\hat{G}$) has not been determined yet either; this is depicted by the dotted circles inside the colour classes. The fracture $\rho$ induced by $\colT$ will depend on the identification of the edges of $\colT$, both endpoints of which lie in $G$. The goal is to show that the endpoints have to be identified precisely as depicted in Figure~\ref{fig:K}.}
    \label{fig:L}
\end{figure}

What we have shown is that $T(i,j)$ consists of two length-$b$ paths.
For some $t\in \{1,2\}$, one of these paths is from $V_i^t$ and the other is from $V_j^t$.  
To be more precise and to fix the notation for $t=1$, we have now shown that,  for each $i\in[k]$, $\colT[G]$ contains a path $R_b(i,1)$ of length $b$ that starts at a vertex $w_i^1\in V_i^1$.
We refer to the other end of this path as $u_b(i,1)$. The vertex $u_b(i,1)$ has degree $1$ and is contained in $M$ (i.e., in $c^{-1}(m)$). 
We next argue 
that $w_i^1$  has degree at least $3$ in $\colT$. (See Figure~\ref{fig:L}.)
\begin{itemize}
\item If $w_i^1$ has degree $1$ in $\colT$ then $\colT$ is disconnected, contradicting the fact that it is a tree.
\item If $w_i^1$ has degree $2$ in $\colT$, then 
$\colT$ has a ray   of length at least $b+1$ that is not external, which is again a contradiction.
\end{itemize}
By the same reasoning, $\colT$ contains a ray $R_b(i,2)$ of length $b$ 
that starts at a vertex $w_i^2 \in V_i^2$ and ends at a vertex $u_b(i,2)$.
The ray $R_b(i,2)$ is contained in $\colT[G]$.

We have just finished parts (d) and (e) of~\eqref{firstgoal} and the part of (g) that concerns   length~$b$. So what we have shown corresponds to Figure~\ref{fig:L}. We would now like to prove parts (c) and (f) but unfortunately these are more difficult because we have to show where the rays with lengths between $a$ and $b$ are embedded so that we can argue about where the length-$a$ rays are embedded.

Define
 $\hat{\calR}:=\bigcup_{i=1}^k \{R_b(i,1),R_b(i,2) \}$. 
Recall that $k$, $a$, and $b$ are positive integers with $a\leq b$ and $k\geq 2$ and that $T$ has $\fork_{a,b}(T) \geq 2k$ and $\colT \cong T$.  
Also, $\colR^b$ is the set of length-$b$ rays in $\colT$ and $\calR^b$ is the set of length-$b$ rays in $T$.
(See Figure~\ref{fig:G}.) Using the notation that we have established,
we will prove the following claims.

\medskip

\noindent \textbf{Claim 1:} Let $P\in (\colR^{b} \setminus \hat{\calR}) \cup \colP^{b}$. If $a < b$   then $P$ is external.

We prove   Claim 1 for the case where $ P \in \colR^{b} \setminus \hat{\calR}$.  The other case is similar but easier.

Observe that $|\calR^{b}|\geq 2k$ since $\fork_{a,b}(T)\geq 2k$.  So $\calR^b$ can be partitioned as follows
\begin{itemize}
    \item $\calR^b[S]$ is the set of the $2k$ length-$b$ rays $F_b(i,j)$ whose sources are $s_1^1,\ldots, s_k^2$ and which   are depicted  as red dashed lines in   Figure~\ref{fig:G}.
    \item $\calR^b[T] = \calR^b\setminus \calR^b[S]$ contains the remaining rays of length $b$. 
\end{itemize}
Our goal is to show that all rays in $\colR^b \setminus \hat{\calR}$ are external. 
To do this, we first show that $|\calR^{b}[T]| = |\colR^{b}\setminus \hat{\calR}|$ and we then provide an injection from $\calR^{b}[T]$ to $ \colR^{b}\setminus \hat{\calR}$ in which all elements of the range are external rays.

To show that $|\calR^{b}[T]| = |\colR^{b}\setminus \hat{\calR}|$, first note that $|\calR^b|=|\colR^b|$ because $T$ and $\colT$ are isomorphic. We further have
$|\calR^b[S]| =|\hat{\calR}| = 2k$.

We next define the (injective) map from $\calR^{b}[T]$ to $\colR^{b}\setminus \hat{\calR}$ . 
  For any ray $R=r_0,r_1,\dots,r_b \in \calR^b[T]$ we proceed as follows. 
  \begin{itemize}
      \item  If $r_0$ is not among the designated sources $s_i^j$, then $R$  is fully contained in $T'$ (see Figure~\ref{fig:G}) and thus $R$ is a ray in $\colT$. We   map $R$ to itself. Note that $R$ is external since it is fully contained in $T'$.
      \item Otherwise, $r_0 =s_i^j$  and $R$  is one of the rays depicted as black dashed lines in Figure~\ref{fig:G}. Since $\colT$ is edge-colourful, and by construction of $\hat{G}$, $\colT$ contains a path $R'=x,r_1,\dots,r_b$ where $x\in V_i^j$. (See Figure~\ref{fig:I}.) 
      If $x$ has degree~$1$ in $\colT$ then $\colT$ is disconnected, which is not true. If $x$ has degree~$2$ in $\colT$ then $\colT$ has a non-external ray which is longer than~$b$, which is also a contradiction by Lemma~\ref{lem:longRays}. Thus, $x$ has degree at least~$3$ in $\colT$, and $R'$ is an external ray.
      We map $R$ to $R'$.
     \end{itemize}
This concludes the proof of Claim~1 for the case where 
$ P \in \colR^{b} \setminus \hat{\calR}$.  $\blacksquare$

\medskip

\noindent \textbf{Claim 2:} Suppose that
there is an integer $t'$ such that $a<t'<b  $. Suppose that $P\in \colR^{t'} \cup \colP^{t'}$. Then $P$ is external.
 
 In order to explain the proof of Claim~2, 
recall that we have established the following facts about $2$-paths in $\colT$ in Lemma~\ref{lem:longRays} and Claim~1.
\begin{itemize}
    \item Every $2$-path of length greater than $b$ is external.
    \item Every $2$-path of length $b$ is either a ray in $\hat{\calR}$ or is   external.
\end{itemize}
With those $2$-paths covered, the proof of Claim~2 is analogous to the proof of Lemma~\ref{lem:longRays}. $\blacksquare$

\medskip

Using Claims~1 and~2 we will now prove parts (c) and (f) of \eqref{firstgoal}.
For each $2$-path whose length is larger than~$a$, we have already shown that it is in $\hat{\calR}$ or we have shown that it is external. In order to prove (c) we will show that,
for each edge $\{v_i,v_{i'}\}$ of~$\Delta$ with colour $s$,
the sequence of edges in $\colT$ between $V_i^1$ and $V_{i'}^1$ is the union of two disjoint length-$a$ rays.
This is formalised as follows.
 
  Note that for each edge $\{v_i,v_j\}$ of $\Delta$ coloured by the $3$-edge-colouring $C$ with $s$, there is a path $P_{i,j}$ of length $2a$ from $v_i^1$ to $v_j^1$. Recall that $c_E$ maps edges of $G$ to edges of $Q$. We write $\colT(i,j)$ for the subgraph of $\colT[G]$ induced by edges $e$ of $G$ such that $c_E(e)$ is in the path $P_{i,j}$. By construction, $\colT(i,j)$ is the union of some number of paths.
We will next argue that it is the union of exactly two disjoint length-$a$ paths:
\begin{itemize}
\item If $\colT(i,j)$ has more than two components then at least one component is disconnected from $T'$ in $\colT$, contradicting the fact that $\colT$ is a tree. 
\item If $\colT(i,j)$ is a single path then it is contained in a 2-path of length at least $2a$. Since this $2$-path contains an edge in~$G$, it is not external (Definition~\ref{def:hi}). Additionally, it is not included in $\hat{\calR}$. This contradicts the aforementioned fact that each $2$-paths of length at least $a+1$ is external or included in the set $\hat{\mathcal{R}}$. 
\item If $\colT(i,j)$ is the union of exactly two disjoint paths, one of which has length larger than~$a$,
then this larger path yields a contradiction similarly to the previous case.
\end{itemize}

What we have shown is that $T(i,j)$ consists of two length-$a$ paths.
One of these paths is from $V_i^1$ and the other is from $V_j^1$.  
To be more precise and to fix the notation, we have now shown that,  for each $i\in[k]$, $\colT[G]$ contains a path $R_a(i,1)$ of length $a$ that starts at a vertex $\hat{w}_i^1\in V_i^1$.
We refer to the other end of this path as $u_a(i,1)$. The vertex $u_a(i,1)$ has degree $1$ and is contained in $S$ (i.e., in $c^{-1}(s)$). 
So we have established Part (c) of item~\eqref{firstgoal}.
Consider Figure~\ref{fig:M} for an illustration of all the information we gathered so far. (The vertices labelled $z_i^j$ and the edge set $E_i^a$ in the figure will be discussed below).

\begin{figure}
    \centering
    \includegraphics[scale=0.85]{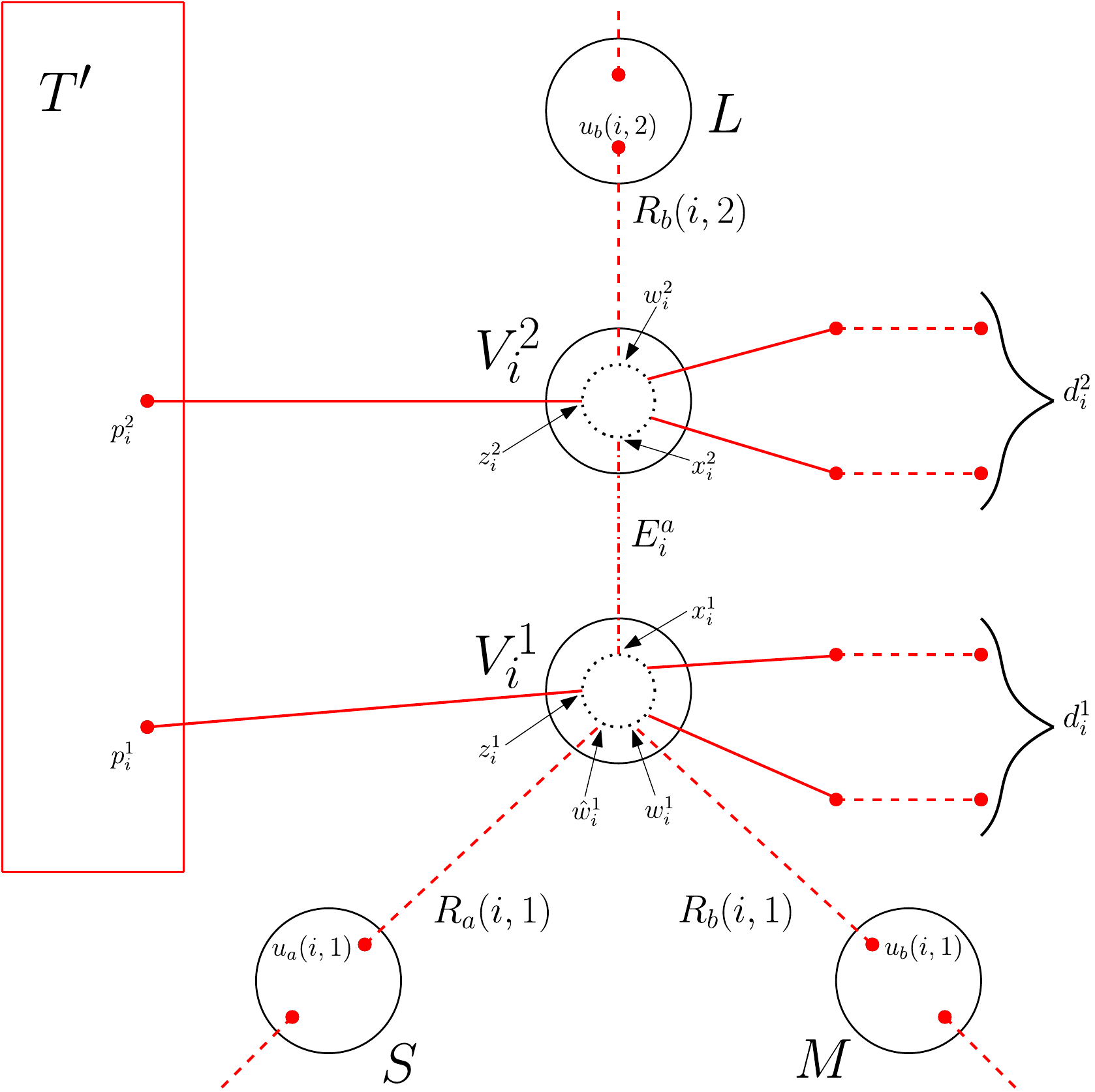}
    \caption{Depiction of the embedding of $\colT$ as established after Claim~2 (in the proof of Lemma~\ref{lem:main_forks}). Solid lines depict edges, dashed lines depict paths, and dash-dotted lines depict sequences of edges (the identification of the endpoints of which has not yet been determined). Note that we have not yet determined how the endpoints inside of the colour classes $V_i^1$ and $V_i^2$ are identified either; this is depicted by the dotted circles inside these colour classes. Proving that the embedding of $\colT$ is as depicted in Figure~\ref{fig:K} requires us to show that all endpoints in $V_i^2$ are identified, and that all endpoints in $V_i^1$, \textbf{except for} $x_i^1$, are identified.  }
    \label{fig:M}
\end{figure}

To finish the proof of item~\eqref{firstgoal} we will show part (f) and the rest of part (b). We take these together.
Recall that for every $i\in[k]$ there is a path 
$P_i^a = v_i^1, y_1,\dots,y_{a-1},v_i^2 $
of length $a$ in $Q$ from $v_i^1$ to $v_i^2$. Since $\colT$ is edge-colourful, it includes 
each of the colours of the edges on this path exactly once --- these colours are $\gamE^{-1}(\{v_i^1,y_1\})$,$\gamE^{-1}(\{y_1,y_2\})$, $\dots$,$\gamE^{-1}(\{y_{a-1},v_i^2 \})$.
Under the edge colouring~$c_E$, the same edges of $\colT$ are coloured with the colours  $\{v_i^1,y_1\}$, $\{y_1,y_2\}$, $\dots$, $\{y_{a-1},v_i^2 \}$.  

Let $e_1,\dots,e_a$ be the edges of $\colT$ with those colours; we write $E_i^a$ for this set of edges (as is depicted in Figure~\ref{fig:M}). We let $x_i^1$ be the vertex of $\colT$ which is contained in $V_i^1$ and incident to $e_1$, and we let $x_i^2$ be the vertex of $\colT$ which is contained in $V_i^2$ and incident to $e_a$. Let $z_i^1$ and $z_i^2$ be the vertices of $\colT$ in $V_i^1$ and $V_i^2$ that are adjacent to $p_i^1$ and $p_i^2$ --- those vertices are depicted in Figure~\ref{fig:M} and we point out that, a priori, $x_i^1$ might be equal to to $z_i^1$ and $x_i^2$ might be equal to $z_i^2$.

\medskip

\noindent \textbf{Claim 3:}   There are no vertices in $V(\colT) \cap V_i^1$ other than $z_i^1$, $x_i^1$, $w_i^1$, $\hat{w}_i^1$ and vertices in the $d_i^1$ rays.

To prove Claim 3, assume for contradiction that $z$ is such a vertex. 
Recall that $V_i^1$ is an independent set (because vertices in $V_i^1$ all receive the same colour under $c$.)
Since $\colT$ is connected, $z$ has a neighbour outside of $V_i^1$ but all of the edge colours incident to $V_i^1$ are already used. $\blacksquare$

The proof of the following claim is similar.

\medskip

\noindent \textbf{Claim 4:} There are no vertices in $V(\colT) \cap V_i^2$ other than $z_i^2$, $x_i^2$, $w_i^2$, and vertices in the $d_i^2$ rays. $\blacksquare$

\medskip

\noindent \textbf{Claim 5:} Both $z_i^1$ and $z_i^2$ have degree at least $3$ in $\colT$.
We prove the claim for $z_i^1$; an analogous argument applies for $z_i^2$. Assume first for contradiction that $z_i^1$ has degree $1$. Since $\colT$ is connected,  Claim 2.5 implies that $|V(\colT) \cap V_i^1| =2$
so $x_i^1 = w_i^1 = \hat{w}_i^1$ and the depicted vertices in the $d_i^1$ rays are also identified with this vertex. By Definition~\ref{def:invalidTrees}, $\colT$ is invalid, giving a
 contradiction.

Now assume for contradiction that $z_i^1$ has degree $2$. We consider two subcases: 
\begin{itemize}
    \item $z_i^1$ is identical to $x_i^1$. Then $\colT$ is disconnected, which yields a contradiction.
    \item $z_i^1$ is identical to $w_i^1$ or $\hat{w}_i^1$. This is an immediate contradiction since sources cannot have degree $2$ (recall that we already established $R_a(i,1)$ and $R_b(i,2)$ to be rays).
    \item $z_i$ is incident to the first edges of one of the additional $d_i^1$ outgoing paths. However, in this case, $\colT$ can only be connected if there is precisely one further vertex of $\colT$ in $V_i^1$ that is incident to all outgoing edges not covered by $z_i^1$. However, in this case, $\colT$ is an invalid tree, yielding the desired contradiction.
\end{itemize}
Since the three cases above are exhaustive, the proof of Claim 5 is concluded. $\blacksquare$

\medskip

Next we need the following property:

\noindent \textbf{Claim 6:} Let $t$ be a positive integer. If $t< a$ then each ray in $\colR^t$ is external.

For the proof, recall that $|\calR^t|=|\colR^t|$ since $T$ and $\colT$ are isomorphic. Note that each ray $R$ of length $t$ of $T$ is either fully contained in $T'$, or it is one of the $d_i^j$ black rays for some $(i,j)\in[k]\times[2]$. (See Figure~\ref{fig:G}) If $R$ is fully contained in $T'$, then $R$ is also contained in $\colR^t$ and it is external. 

If $R=r_0,r_1,\dots,r_t$ is one of the $d_i^j$ black rays, then $\colT$ contains a path $R'=y_0,r_1,\dots,r_t$ for some $y_0\in V_i^j$. Suppose that   $y_0$ has degree at least $3$ in $\colT$.
Then, as in Claim~1,   $R'$ is then an external ray, and we are finished.
We next consider the case where $y_0$ has degree~$1$ or~$2$ in $\colT$.  

If the degree is $1$, then $\colT$ is disconnected, leading to a contradiction. If the degree is $2$, then $y_0 \neq z_i^j$ by Claim 5. Thus, the only way for $\colT$ not being disconnected is $y_0=x_i^j$ and $\colT[E_i^a]$ is a path. However, then we obtained a ray of length at least $a+t$ which is neither external, nor in the set $\hat{\calR}$. Thus, we obtain a contradiction by either Claim~2 ($a+t<b$), or by Claim~1 ($a+t=b$), or by Lemma~\ref{lem:longRays} ($a+t>b$). This concludes the proof of Claim 6. $\blacksquare$

\medskip

Next, observe that $\colT$ cannot connect $z_i^1$ and $z_i^2$ via a path within $G$, that is, via a path containing the edges $E_i^a$: Otherwise $\colT$ would contain a cycle since $p_i^1$ and $p_i^2$ are connected by a path within $T'$. We will see that $z_i^1$ and $z_i^2$ are sources of $\colT$.

Let $\mathcal{S}$ be the set of all sources of $T$. Consider the \emph{multi-set} of leaf-degrees of $T$
\[\degl(\mathcal{S}) := \{\{\degl(s)~|~s\in\mathcal{S} \}\} \,.\]
Let $\colS$ be the set of all sources of $\colT$ and let $\degl(\colS)$ be the muti-set of leaf-degrees $\colT$.
Since $\colT$ and $T$ are isomorphic, the multi-sets $\degl(\mathcal{S})$ and $\degl(\colS)$ are equal. 

Suppose that  $s\in\mathcal{S}$ is a source of $T$ not among the designated sources $s_i^j$. Then $s$ is contained in $T'$, and it is also a source of $\colT$. 
Since all of the $z_i^j$ have degree at least $3$ in $\colT$ (by Claim 5), 
they cannot be part of further rays with source~$s$ in~$\colT$ so  $s$ has the same leaf-degree in $T$ and in $\colT$.

We next show that for each $i\in[k]$, the set $V_i^1 \cup V_i^2$ contains at least $2$ sources of $\colT$: Either $z_i^1$ is a source  or it is connected by a 2-path within $\colT[G]$ to another source. However, the only vertices reachable in $\colT[G]$ from $z_i^1$ that can have degree at least $3$ are contained in $V_i^2$. 
Similarly, either $x_i^2$ is a source or it is connected by a 2-path within $\colT[G]$ to a source in~$V_i^1$. 
We have already seen that $z_i^1$ cannot be connected to $z_i^2$ within $\colT[G]$. Thus the sources reachable from $z_i^1$ and $z_i^2$ within $\colT[G]$ must be distinct, and we have shown that for each $i\in[k]$, the set $V_i^1 \cup V_i^2$ contains at least $2$ sources of $\colT$. 

Since $\colT$ and $T$ have the same number of sources, and since $2k$ sources of $T$ are not contained in $T'$, we have thus shown that for each $i\in[k]$, the set $V_i^1 \cup V_i^2$ contains \emph{precisely} $2$ sources of $\colT$; let us denote those $2$ sources by $\hat{z}_i^1$ and $\hat{z}_i^2$. 

Now, consider the following subsets of $\mathcal{S}$ and $\colS$: 
\begin{itemize}
    \item ${\mathcal{S}'}:=\{s_1^1,s_1^2,\dots,s_k^1,s_k^2\}$ is the set of designated sources.
    \item ${\colSprime}:=\{\hat{z}_1^1,\hat{z}_1^2,\dots,\hat{z}_k^1,\hat{z}_k^2\}$ is the set of sources of $\colT$ in $G$ (within $\hat{G}$).
\end{itemize}
Since we already know that $\degl(\mathcal{S}\setminus {\mathcal{S}'}) = \degl(\colS \setminus {\colSprime})$ (those are the sources in $T'$), we require $\degl({\mathcal{S}'})=\degl(\colSprime)$ for $T$ and $\colT$ to be isomorphic.

What follows is the final claim within the proof of this lemma.

\medskip

\noindent \textbf{Claim 7:}
For all $i\in[k]$, the following five conditions are satisfied:
\begin{itemize}
    \item $\{z_i^1,z_i^2 \} = \{\hat{z}_i^1,\hat{z}_i^2\}$, that is, $z_i^1$ and $z_i^2$ are the two sources in $V_i^1\cup V_i^2$.
    \item $\colT$ contains precisely $2$ vertices in $V_i^1$: One is $z_i^1$ and one is $x_i^1$.
    \item $x_i^1$ has degree $1$. Further, $z_i^1$, $w_i^1$, $\hat{w}_i^1$ and all the endpoints of the $d_i^1$ rays are the same.
    \item $\colT$ contains precisely $1$ vertex in $V_i^2$. Further, $z_i^2$, $x_i^2$, $w_i^2$ and all endpoints of the $d_i^2$ rays are the same.
    \item $\colT[E_i^a]$ is a ray with source $z_i^2(=x_i^2=w_i^2)$.
\end{itemize}
Before proving Claim 7, we point out that~\eqref{bul2} and~\eqref{bul6} follow immediately from Claim 7; see Figure~\ref{fig:K} and observe that $\colT[E_i^a]$ becomes the ray $R_a(i,2)$, and $x_i^1$ becomes the endpoint $u_a(i,2)$ of $R_a(i,2)$ for each $i\in[k]$. Thus the proof of this lemma is concluded if Claim 7 is proved, which is done below:

\begin{itemize}
    \item  
  We first show that $\{z_i^1,z_i^2 \} = \{\hat{z}_i^1,\hat{z}_i^2\}$ for each $i\in[k]$. Let $\Phi=\sum_{s \in {\mathcal{S}'}}\degl(s)$ and $\colPhi= \sum_{s \in {\colSprime}}\degl(s)$. Observe that $\degl({\mathcal{S}'})=\degl({\colSprime})$ implies $\Phi=\colPhi$.
   
We start by observing that   \[\degl(\hat{z}_i^1) + \degl(\hat{z}_i^2) \leq (d_i^1 + 2) + (d_i^2 + 1) + 2.\]
There are $d_i^1$ rays from $V_i^1$ and also $R_a(i,1)$ and $R_b(i,1)$.
There are $d_i^2$ rays from $V_i^2$ and also $R_b(i,2)$.
There is also $E_i^a$ which could form two rays.

We next show that $E_i^a$ cannot form two rays.   
Assume for contradiction that is does. Since $\colT$ is connected,  $z_i^1,w_i^1,\hat{w}_i^1,x_i^1$ and all the endpoints of the $d_i^1$ rays are identical, and $z_i^2,w_i^2,x_i^2$ and all the endpoints of the $d_i^2$ rays are identical. 

Now, if $\colT[E_i^a]$ would be the disjoint union of two rays of length less than $a$ with sources $z_i^1$ and $z_i^2$ then those rays are  non-external rays of length less than $a$, contradicting Claim 6. 
We have now shown   
    \begin{equation}\label{eq:potential}
        \degl(\hat{z}_i^1) + \degl(\hat{z}_i^2) \leq (d_i^1 + 2) + (d_i^2 + 2)\,.
    \end{equation}

Next, note that by definition of the $d_i^j$ (see Figure~\ref{fig:G}), the following holds:
\begin{equation}\label{eq:leafdegrees}
    (d_i^1 + 2) + (d_i^2 + 2) = \degl(s_i^1) + \degl(s_i^2) 
\end{equation}

We have now shown that 
$$   \degl(\hat{z}_i^1) + \degl(\hat{z}_i^2) \leq \degl(s_i^1) + \degl(s_i^2).$$
Finally, we will show that $z_i^1$ and $z_i^2$ are sources to finish the first bullet point.

Consider $z_i^1$, and recall that is has degree at least $3$ by Claim 5, and assume for contradiction that it is not a source of $\colT$. Then $z_i^1=x_i^1$, and $\colT[E_i^a]$ is a path, and $x_i^2$ is source (since it is the only vertex in $V(\colT)\cap V_i^2$ that might have degree at least $3$, except for $z_i^2$). Note that this also implies that $z_i^2$ is a source. Thus $\{\hat{z}_i^1,\hat{z}_i^2\}=\{x_i^2,z_i^2\}$. In this case, we have \[\degl(\hat{z}_i^1) + \degl(\hat{z}_i^2) \leq d_i^2 + 1 < \degl(s_i^1) + \degl(s_i^2)\,.\]
Consequently, using~\eqref{eq:potential} and~\eqref{eq:leafdegrees}, we have $\colPhi < \Phi$, which is  a contradiction. Thus $z_i^1$ is a source of $\colT$, and a similar argument shows that $z_i^2$ is a source of $\colT$ as well.
 \item We  now prove the remaining items. In what follows, using the previous bulleted item, we can assume that w.l.o.g.\ $\hat{z}_i^1=z_i^1$ and $\hat{z}_i^2=z_i^2$ for all $i\in[k]$. First, recall that we ordered the $s_i^j$ by their leaf-degrees, that is \[\degl(s_1^1) \geq \degl(s_1^2) \geq \dots \geq \degl(s_k^2)\geq 2\,.\]

If $x_1^1$ were  equal to $z_1^1$, then $\colT$ can only be connected if there is only one vertex in $V_1^1$, that is, all edges incident to $V_1^1$ are in fact incident to $z_1^1$. However, in that case, we have $\degl(z_1^1)=\degl(s_1^1)+1$ (by construction of $\hat{G}$), and thus the multi-sets cannot be equal anymore. Hence $x_1^1 \neq z_1^1$. 

If $x_1^1$  had degree $2$, then  there would have been a ray of length at least $a+1$ that originates in $V_1^2$ (otherwise $\colT$ would have been disconnected). However, this ray would neither be external, nor among the rays in $\hat{\calR}$, contradicting either Lemma~\ref{lem:longRays} or the previous sequence of claims. Finally, if $x_1^1$ had degree at least $3$, then $\colT$ would have contained more sources than $T$, which also yields a contradiction.

This shows that $x_1^1$  has degree $1$. However, this implies that $\colT$ can only contain one vertex in $V_i^2$; otherwise $\colT$ would be disconnected. 
Note that we have just proved the remaining items of Claim 7 for $i=1$. Additionally, we have shown that $\degl(z_1^1)=\degl(s_1^1)$ and $\degl(z_1^2)=\degl(s_1^2)$
Hence we can remove those two numbers from the multi-sets and continue recursively with $i=2$. This concludes the proof of Claim 7, and thus the proof of the overall lemma.
\end{itemize}

\end{proof}

We are now ready to conclude the case for trees of unbounded fork number.
\begin{lemma}\label{lem:success_forks}
Let $\calT$ be a recursively enumerable class of trees of unbounded fork number. Then $\oplus\subsprob(\calT)$ is $\oplus\W{1}$-hard.
\end{lemma}
\begin{proof}
We proceed similarly to Lemma~\ref{lem:success_Hgadgets}. However, we have to take care of some subtleties. First, we start with a class $\mathcal{C}$ of cubic \emph{bipartite} graphs of unbounded treewidth.
Next, we wish to rely on   Lemma~\ref{lem:main_forks}  to obtain the identity 
\[\oplus\embs{(\fracture{Q}{\tau},c_\tau)}{(G,c)} = \oplus\colsubs{T}{(\hat{G},\hat{\gamma})},\]
where $\tau$ is the fracture defined in Definition~\ref{def:tau_caseF}. Unfortunately, Lemma~\ref{lem:main_forks} only yields the above identity if, for each $v\in V(Q)$, $|c^{-1}(v)|$ is odd, that is, each colour class of vertices of $G$ has odd cardinality. However, this property  can easily be achieved. Let $(G',c')$ be the $Q$-coloured graph obtained from $(G,c)$ by adding to each even colour class one fresh isolated vertex. Since $\fracture{Q}{\tau}$ does not have isolated vertices, this operation does not change the number of colour-preserving embeddings. In combination with Lemma~\ref{lem:main_forks} we thus obtain
\[\oplus\embs{(\fracture{Q}{\tau},c_\tau)}{(G,c)} = \oplus\embs{(\fracture{Q}{\tau},c_\tau)}{(G',c')} = \oplus\colsubs{T}{(\hat{G'},\hat{\gamma})}.\]
From here on, we can proceed analogously to the proof of Lemma~\ref{lem:success_Hgadgets}.
\end{proof}

\subsection{The Dichotomy Theorem for Trees}\label{sec:collect_trees}
We are now able to prove Theorem~\ref{thm:main}, i.e., an exhaustive and explicit parameterised complexity classification for counting trees modulo $2$:
\main*
\begin{proof}
The fixed-parameter tractability result, as well as the fact that $\oplus\subsprob(\calT)$ is always contained in $\oplus\W{1}$ were both shown in~\cite{CurticapeanDH21}. Hence, it remains to prove $\oplus\W{1}$-hardness if $\calT$ is not matching splittable.

By Lemma~\ref{lem:non-matchsplit-trees} each class $\calT$ of trees that is not matching splittable has unbounded $\hal$-number, unbounded star number, or unbounded fork number. Finally, each of these three cases yields $\oplus\W{1}$-hardness as established by Lemmas~\ref{lem:success_Hgadgets}, \ref{lem:success_stars}, and~\ref{lem:success_forks}.
\end{proof}

\section{Conclusion and Open Questions}

Given a class $\calH$ of patterns, the problem
$\oplus\subsprob(\calH)$ asks, given as input a graph $H\in \calH$ and an arbitrary graph~$G$, to count the subgraphs of~$G$ that are isomorphic to~$H$. 

This work is motivated by the conjecture of Curticapean, Dell and Husfeldt (Conjecture~\ref{conj:subsmod2})
that $\oplus \subsprob(\calH)$ is FPT if and only if $\calH$ is \emph{matching splittable}.

Recall that  
the \emph{matching-split number} of  $H$ is the minimum size of a set $S\subseteq V(H)$ such that $H\setminus S$ is a matching. The class $\calH$ is  \emph{matching splittable} if there is a positive integer $B$ such that the matching-split number of any $H\in\calH$ is at most $B$.  

In this work,
Theorem~\ref{thm:main_hereditary} proves the conjecture for every hereditary class of graphs.
Theorem~\ref{thm:main} proofs the conjecture for every class $\calH$ of trees.

Clearly, the most important task for related future work is to fully resolve the conjecture. Our work has shown that the use of edge-colours, formalised by the framework of fractured graphs, makes it possible to bypass the problem caused by  automorphism groups that have even cardinality. We think that this approach will be useful  for future work.




\bibliography{subgraphs}

\end{document}